\documentclass[12pt, notitlepage]{article}
\usepackage{amsfonts, amsmath, amssymb, amsthm, color, hyperref, enumitem, fancyhdr, relsize, tikz, graphicx, mathrsfs}
\usepackage[utf8]{inputenc}
\usepackage{lipsum}

\title{\scshape 
Constructing Nearby Commuting Matrices for Reducible Representations of $su(2)$ with an Application to Ogata's Theorem  \sffamily }
\author{\scshape \sffamily David Herrera\footnote{Rutgers University. dh708@math.rutgers.edu}
}
\date{\today}

\newtheorem{prop}{Proposition}[section]

\newtheorem{thm}[prop]{Theorem}
\newtheorem{lemma}[prop]{Lemma}

\theoremstyle{definition}
\newtheorem{remark}[prop]{Remark}
\newtheorem{example}[prop]{Example}

\newtheorem{defn}[prop]{Definition}

\newcommand{\R}{\mathbb{R}}
\newcommand{\C}{\mathbb{C}}
\newcommand{\N}{\mathbb{N}}

\newcommand{\Z}{\mathbb{Z}}

\newcommand{\tcw}{\textcolor{white}}

\newcommand{\lam}{\lambda}
\newcommand{\Lam}{\Lambda}

\renewcommand{\Im}{\operatorname{Im}}
\renewcommand{\Re}{\operatorname{Re}}

\renewcommand{\d}{\,d}

\newcommand{\sgn}{\operatorname{sgn}}

\newcommand{\diag}{\operatorname{diag}}
\newcommand{\ws}{\operatorname{ws}}
\newcommand{\bws}{\operatorname{b-ws}}
\newcommand{\diam}{\operatorname{diam}}

\newcommand{\bp}{\begin{pmatrix}}
\newcommand{\ep}{\end{pmatrix}}
\newcommand{\spn}{\operatorname{span}}

\newcommand{\sround}{\operatorname{s-rnd}}

\newcommand{\dark}{}

\begin{document}
\dark

\maketitle

\abstract{Resolving a conjecture of von Neumann, Ogata's theorem in \cite{Ogata} showed the highly nontrivial result that arbitrarily many matrices corresponding to macroscopic observables with $N$ sites and a fixed site dimension $d$ are asymptotically nearby commuting observables as $N \to \infty$.

In this paper, we develop a method to construct nearby commuting matrices for normalized highly reducible representations of $su(2)$ whose multiplicities of irreducible subrepresentations exhibit a certain monotonically decreasing behavior.

We then provide a constructive proof of Ogata's theorem for site dimension $d=2$ with explicit estimates for how close the nearby observables are. Moreover, motivated by the application to time-reversal symmetry explored in \cite{LS}, our construction has the property that   real macroscopic observables are asymptotically nearby real commuting observables.}

\section{Introduction}
\label{Introduction}

Following, \cite{HastingsLoring}, we say that matrices $A_1, \dots, A_k\in M_n(\C)$ are $\delta$-almost commuting if $\|[A_i,A_j]\| \leq \delta$ for each $i$ and $j$, where $\|-\|$ is the operator norm. 
We say that $A_1, \dots, A_k$ are $\varepsilon$-nearly commuting if there are commuting matrices $A_i'$ such that $\|A_i'-A_i\| \leq \varepsilon$ for each $i$. 
We are interested primarily in the case where the $A_i$ are self-adjoint and the $A_i'$ are also self-adjoint.

Lin's theorem states that two $n\times n$ almost commuting self-adjoint contractions are nearby commuting self-adjoint matrices, independently of $n$. 
More precisely, Lin showed that there is a function $\varepsilon = \varepsilon(\delta)$ with $\lim_{\delta\to0^+}\varepsilon(\delta) = 0$ so that if self-adjoint contractions $A, B \in M_n(\C)$ are $\delta$-almost commuting 
then they are $\varepsilon(\delta)$-nearly commuting and the nearby commuting matrices $A', B'$ can be chosen to be self-adjoint. 
It is important to note that $\varepsilon(\delta)$ is independent of $n$.

Before Lin's theorem was proved, it was shown that it is not always true that three or more almost commuting self-adjoint matrices are nearly commuting (\cite{Choi}, \cite{VoicHeisenberg}, \cite{Davidson}).
Since Lin's theorem's proof, various refinements and generalizations of Lin's theorem have been proved. For example, Kachkovskiy and Safarov in \cite{KS} obtained the optimal homogeneous estimate of $\varepsilon(\delta) = Const.\delta^{1/2}$. As with Lin's theorem, this result  uses the fundamental fact that an operator $S$ is almost normal (meaning that its self-commutator $[S^\ast, S]$ has small norm) if and only if its real and imaginary parts are almost commuting.

\vspace{0.05in}

The question of whether almost commuting self-adjoint matrices are nearly commuting was raised by arguments of von Neumann (\cite{Neumann}) in the context of observables of large quantum systems.
The idea is that although the observables of a large quantum system may be almost commuting, they usually do not exactly commute. So, they cannot be measured simultaneously with respect to every state (or even any state). However, if these observables were nearly commuting then by choosing a collection of nearby commuting observables, we could instead measure these approximations simultaneously. von Neumann hypothesized that observations of macroscopic systems in practice were the result of such approximate measurements.

Proved in 2011 (\cite{Ogata}), Ogata's theorem (Theorem \ref{OgataThm} below) is a generalization of Lin's theorem for finitely many almost commuting self-adjoint matrices as an answer to the conjecture of von Neumann mentioned above. The rest of this section focuses on introducing Ogata's theorem and a result of this paper which extends Ogata's theorem.

For $A \in M_d(\C)$, define $T_N(A)$ as the normalized average of the $N$ many
$N$-fold tensor products of $A$ with $N$-1 identity matrices $I_d\in M_d(\C)$: 
\begin{align}
\nonumber
T_N(A) = \frac{1}{N}\left(A \otimes I_d \otimes \cdots \otimes I_d + I_d\otimes A\otimes I_d \otimes \cdots \otimes I_d + \cdots + I_d \otimes \cdots \otimes I_d \otimes A\right).
\end{align}
When $A$ is self-adjoint it is associated with an observable for a $d$-dimensional system and $T_N(A)$ is the  associated macroscopic observable.
Then, for $A_1, \dots, A_k \in M_d(\C)$ self-adjoint  define the macroscopic observables $H_{i,N} = T_N(A_i) \in M_{d^N}(\C)$.
Because $[H_{i,N}, H_{j,N}] = \frac{1}{N}T_N([A_i, A_j]),$ the $H_{i, N}$ are almost commuting. Ogata's theorem states that the $H_{i, N}$ are nearby commuting self-adjoint matrices:
\begin{thm}\label{OgataThm}
For $A_1, \dots, A_k \in M_d(\C)$ self-adjoint, there are commuting self-adjoint matrices $Y_{i,N}\in M_{d^N}(\C)$ so that $\|T_N(A_i)-Y_{i,N}\|\to 0$ as $N\to\infty$.  
\end{thm} 
In line with von Neumann's motivation for the almost-nearly commuting matrices problem, Ogata's Theorem has had applications to the theory of quantum statistical mechanics as explored by various authors (\cite{GoldsteinTE}, \cite{GoldsteinMTE}, \cite{Tasaki}).
As an example, \cite{HalpernNAT} and \cite{Halpern} apply Ogata's theorem to construct what the authors of those papers call an approximate microcanonical subspace. Due to the nonconstructive proof of Ogata's theorem, objects constructing using Ogata's theorem are also not constructive, as observed in Remark 7.1 of \cite{Khanian}. One consequence of this is that one cannot know if the result of Ogata's theorem is non-trivial for reasonably sized systems. 

\vspace{0.1in}

We now discuss the estimates of $\|H_{i,N} - Y_{i,N}\|$ and the extension of Ogata's theorem we prove in this paper.
Because the $H_{i,N}$ satisfy $\|[H_{i,N}, H_{j,N}]\| \leq Const. N^{-1}$, the optimal estimate for Lin's theorem in \cite{KS} implies that if $k = 2$, there are nearby commuting self-adjoint matrices within a distance of $Const. N^{-1/2}$.
Based on the proof of Ogata's theorem in \cite{Ogata} which guarantees that $\varepsilon = o(1)$ as $N \to \infty$, we cannot infer if this or a similar estimate holds for more than two matrices.

In Theorem \ref{mainthm} of this paper, we construct nearby commuting matrices for certain normalized direct sums of irreducible representations of $su(2)$. 
As a consequence of this, we provide a constructive proof of Ogata's theorem for $d = 2$ with an explicit constant and an asymptotic rate of decay of $N^{-1/7}$. More precisely, we prove:
\begin{thm}\label{OgataTheorem}
Let $\sigma_i$ be the norm $1/2$ Pauli spin matrices in Equation (\ref{pSpin}). There are commuting self-adjoint matrices $Y_{i,N} \in M_{2^N}(\C)$ such that
\begin{align}\|T_N(\sigma_1)-Y_{1,N}\|, \|T_N(\sigma_2)-Y_{2,N}\| &\leq 6.29 \,N^{-1/7}, \nonumber\\
\|T_N(\sigma_3)-Y_{3,N}\| &\leq 1.09\,N^{-3/7}
\nonumber
\end{align}
where $Y_{1,N}, iY_{2,N}$, and $Y_{3,N}$ are real.

Therefore, there is a linear map $Y_N:M_2(\C)\to M_{2^N}(\C)$ such that the $Y_{N}(A)$ commute for all $A\in M_2(\C)$,
\[Y_{N}(A^\ast)=Y_{N}(A)^\ast,\]
\[Y_{N}(A^T)=Y_{N}(A)^T,\]  and
\[\|T_N(A)-Y_{N}(A)\| \leq 17.92 \|A\|\,N^{-1/7}.\]

Consequently, $Y_{N}$ preserves the property of being self-adjoint, skew-adjoint, symmetric, antisymmetric, real, or imaginary.
\end{thm}
For $A$ self-adjoint, we thus obtain an explicit estimate for how close the commuting observables $Y_{N}(A)$ are to the macroscopic observables $T_N(A)$ as well as a construction of the $Y_{N}(A)$. In terms of how Ogata's theorem is presented in \cite{Ogata}, the  commuting observables $Y_{1,N}$, $Y_{2,N}$, and $Y_{3,N}$ are nearby the macroscopic observables of the $x$, $y$, and $z$ components of the total magnetization for a quantum spin system of $N$ sites of dimension $d=2$.

The transpose symmetry of $Y_{N}$ due to this extension of Ogata's theorem may be of interest given the attention given to structured nearby commuting matrices in  \cite{KTheoryandS} and \cite{LS}, which apply it to the theory of topological insulators.

The explicit estimates obtained, the additional structure of the matrices, and the simplification of Ogata's original argument for this case are some of the contributions of this construction.
However, due to the use of the Clebsch-Gordan change of basis and the large size of the matrices, it is unclear how useful the construction would be for generating or manipulating the constructed nearby commuting matrices.

As an example of the estimate from the theorem above, a three dimensional array of $N = (10^{10})^3$ particles gives a very small error compared to $\|A\|$. So, the estimate obtained is nontrivial for $N$ in the range of applications. See Remark \ref{useful} for more details. 
Our method can also provide an exponent of $-1/5$ by using $\cite{KS}$, however the explicit constant is not given and $Y_N$ may not have the transpose symmetry. See Theorem \ref{OptimalResult}. 

\vspace{0.1in}

\noindent\underline{Overview of Paper}: 

In Section \ref{Intro2}, we present a more full exposition of the history of the almost/nearly commuting matrix problem and the physical significance of its application to macroscopic observables. More details about the operator $T_N$ are presented. Then the approach to our extension of Ogata's theorem is introduced, followed by a detailed outline of the steps of the proof of Theorem \ref{mainthm} and Theorem \ref{OgataTheorem}.

In Section \ref{Prelim}, we review the basic representation theory of $su(2)$. We also develop other representation theoretic estimates and constructions that will be used later in the paper. Our proof of Theorem \ref{OgataTheorem} relies on framing the problem in terms of tensor representations of $su(2)$ so that almost commuting self-adjoint matrices can be constructed for the macroscopic observables associated to the Pauli matrices.

In Section \ref{GEL-Section}, we discuss weighted shift matrices and our version of Berg's gradual exchange lemma from \cite{BergNWS}. Berg's gradual exchange lemma  provides a way to perform a small perturbation of a direct sum of weighted shift operators to cause the orbits to interchange. 
This section includes an introduction to our weighted shift diagrams.

In Section \ref{ANWS}, we adapt Berg's construction from \cite{BergNWS} of a nearby normal matrix for an almost normal weighted shift matrix. Our adaptation of Berg's result is aimed at obtaining an optimal estimate in terms of  $\|\,[S^\ast, S]\,\|$ with the additional structure that when the almost normal matrix $S$ is real, the nearby normal constructed will be real as well.

In Section \ref{GEP-Section}, a method is developed to obtain almost invariant projections of direct sums of weighted shift matrices that can be used to make almost reducing subspaces. 
This method and the construction of nearby commuting matrices using it are referred to as the gradual exchange process. 
Suppose that $S$ is a direct sum of weighted shift matrices and $A$ is a direct sum of diagonal matrices. Under some conditions on $A$ and $S$, we construct nearby commuting matrices $A'$ and $S'$ using the gradual exchange process.  Several figures are included to illustrate the algorithm. 

In Section \ref{MainTheorem-Section}, we prove Theorem \ref{mainthm}, a constructive result with estimates concerning nearby commuting matrices for normalized direct sums of certain irreducible representations of $su(2)$. 
As a consequence of this result, a constructive proof of Theorem \ref{OgataTheorem} is obtained. 

\vspace{0.1in}

In this paper, all operators are assumed to act in finite dimensional (complex) Hilbert space. The norm $\|-\|$ is the operator norm.
All projections are assumed to be self-adjoint (alias Hermitian). Consequently, when we say that multiple projections are ``orthogonal projections'' we mean that their ranges are orthogonal subspaces.  

For an operator $T$, $R(T)$ denotes the range of $T$. If $F$ is a projection, $R(F)$ is invariant under $T$ if $T(R(F)) \subset R(F)$ or, equivalently, $(1-F)TF = 0$. The projection $F$ commutes with $T$ exactly when both $R(F)$ and $R(F)^{\perp} = R(1-F)$ are invariant subspaces of $T$. 

The adjoint (alias conjugate transpose) of the matrix $S\in M_n(\C)$ is denoted by $S^\ast$. The transpose of $S$ is denoted by $S^T$.
The real and imaginary parts of a matrix $S$ are defined to be $\Re(S) = \frac12(S+S^\ast)$ and $\Im(S) = \frac1{2i}(S-S^\ast)$, respectively.
If $A \in M_n(\C)$ then $\sigma(A)$ denotes the spectrum of $A$, namely the set of eigenvalues of $A$. If $\Omega \subset \C$ and $A$ is a normal matrix, then $E_{\Omega}(A)$ is the spectral projection of $A$ with respect to $\Omega$. The self-commutator of $S$ refers to $[S^\ast, S]$.

\section{Ogata’s Theorem for Macroscopic Observables}\label{Intro2}

We now survey some of the developments of the almost/nearly commuting matrices problem, ending with Ogata's theorem that macroscopic observables are nearly commuting.

Rosenthal in 1969 (\cite{Rosenthal}) wrote a paper raising awareness of the problem of almost/nearly commuting matrices and Halmos (\cite{Halmosproblems}) in 1976 included it in his list of open problems about Hilbert space operators in 1976.
Only partial results were know at the time. 
It was known that nearby commuting matrices did exist (\cite{Subnormal}, \cite{Lux}), unlike in the infinite dimensional case (\cite{berg-olsen}).

There were, however, no results showing that $\delta$ could be chosen independent of the size of the matrices $A, B$.
This is important for multiple reasons. 
For those interested in approximation problems of bounded operators on infinite dimensional Hilbert spaces, a dimension-independent result can be used to obtain results about compact operators. 
See \cite{Handbook} for more about this. 

If one instead is interested in von Neumann's original context, then one would either need a dimensionless result or good control of the dimensional dependence.
The result by Pearcy and Shields (\cite{P-S}) in 1978 gave an estimate of $\varepsilon = Const.n^{1/2}\delta^{1/2}$ if one of the matrices is self-adjoint and, in 1990, Szarek (\cite{Szarek}) improved the dimensional dependence to $\varepsilon = Const.n^{1/13}\delta^{2/13}$ if both matrices are self-adjoint. 

However, these results do not tell us that two sequences of self-adjoint matrices $A_N, B_N$ are nearly commuting if the size of these matrices grows much faster than $\|[A_N,B_N]\|$ converges to zero. This is the case for two macroscopic observables $H_{1,N}, H_{2,N} \in M_{d^N}(\C)$ discussed below because they have size $d^N$ and commutator with norm $O(1/N)$ as $N \to \infty$.

In 1983, a short paper by Voiculescu (\cite{Voiculescu}) provided two sequences $U_N, V_N$ of almost commuting unitaries that are not nearby commuting unitary matrices. Using similar methods, Davidson (\cite{Davidson}) in 1985 provided two sequences of matrices $A_N, B_N$ with $A_N$ self-adjoint and $B_N$ normal that are not nearby commuting matrices $A'_N, B'_N$ with $A'_N$ self-adjoint.

These counter-examples also provide counter-examples to the problem for $k\geq 3$ almost commuting self-adjoint matrices. 
If we define $A_{1,N} = \Re U_N, A_{2,N} = \Im U_N, A_{3,N} = \Re V_N, A_{4,N} = \Im V_N$ then Voiculescu's result shows that in general four almost commuting self-adjoint matrices $A_{i,N}$ may not be (simultaneously) nearly commuting. This is because if $A'_{i,N}$ were nearby commuting self-adjoint matrices then $U'_N = A_{1,N}'+iA_{2,N}', V'_N = A_{3,N}'+iA_{4,N}'$ are commuting normal matrices close to $U_N, V_N$ which can be perturbed to commuting unitaries.

A consequence of Davidson's result is that if we define $A_{1,N} = A_N$, $A_{2,N} = \Re B_N$, $A_{3,N} = \Im B_N$ then $A_{i,N}$ are three almost commuting self-adjoint matrices that are not nearly commuting.
Earlier in 1981, Voiculescu (\cite{VoicHeisenberg}) also had a less explicit proof of this by investigating some of the properties of the $C^\ast$-algebra of the Heisenberg group.

In \cite{Szarek}, Szarek states that the key consequence of his result that we mentioned above is that the problem of two almost commuting self-adjoint matrices is ``completely different'' than the (explicit) counter-examples that existed at the time.
In fact, although such negative results existed for different types of matrices, in 1995 (\cite{Lin}) Lin showed that two almost commuting self-adjoint matrices are nearby commuting self-adjoint matrices. Then in 1996, Friis and R{\o}rdam (\cite{F-R}) provided a simplified proof of this result of Lin. The proof of Lin's theorem, however, was left nonconstructive and without explicit control of $\varepsilon = \varepsilon(\delta)$. Extending this result has garnered interest in recent years (\cite{Hastings1sted}, \cite{Hastings}, \cite{Self-Commutator}, \cite{LS}, \cite{KS}, \cite{DH}, \cite{VectorWise}). As mentioned in Section \ref{Introduction}, in \cite{KS} there is a proof that one can choose $\varepsilon(\delta) = Const.\delta^{1/2}$.

The notion of almost commuting operators associated with observables being near actually commuting observables is discussed and used in a 1929 paper by von Neumann, translation provided in \cite{Neumann}. A specific passage in the beginning of the article states:
\begin{quote}
Still, it is obviously factually correct that in macroscopic measurements the coordinates and momenta are measured simultaneously – indeed, the idea is that that becomes possible through the inaccuracy of the macroscopic measurement, which is so great that we need not fear a conflict with the uncertainty relations.\\
...\\
We believe that the following interpretation is the correct one: in a macroscopic measurement of coordinate and momentum (or two other quantities that cannot be measured simultaneously according to quantum mechanics), really two physical quantities are measured simultaneously and exactly, which however are not exactly coordinate and momentum. They are, for example, the orientations of two pointers or the locations of two spots on photographic plates– and nothing keeps us from measuring these simultaneously and with arbitrary accuracy, only their relation to the really interesting physical quantities ($q_k$ and $p_k$) is somewhat loose, namely the uncertainty of this coupling required by the laws of nature corresponds to the uncertainty relation[.]
\end{quote}

This analysis of an aspect of the measurement problem presumes that such nearby commuting self-adjoint observables exist. 
Ogata's theorem confirms a mathematical formulation of the statement that macroscopic observables are nearby commuting observables with error going to zero as the uncertainty obstruction goes to zero.

An interesting counter-factual twist in the story might have been if von Neumann's physical argument was correct without Ogata's theorem being true. This certainly could be the case for certain observables of macroscopic objects defined under other assumptions. 
In such a scenario, it would be interesting if the error of measurement of these commuting observables did not go to zero as the uncertainty obstruction vanishes, but instead the error of such a measurement was numerically much smaller than would be detected macroscopically.

However, even with knowing Ogata's theorem, there may be limitations of its applicability due to our lack of knowledge of how close the exactly commuting observables $Y_{i, N}$  can be chosen to the given macroscopic observables $T_N(\sigma_i)$. This case has much in common with the speculation of a world where Ogata's theorem did not hold.
In particular, based on the proof in \cite{Ogata}, it is conceivable that Ogata's theorem might only be non-trivial for $N$ much larger than what is seen in any physical application. It is conceivable then that reality may reject our description of macroscopic observables by Ogata's theorem not being capable of providing suitable estimates. (However, it may still allow von Neumann's intuitive argument to be realized using a different mathematical formalism.) 
Our extension of Ogata's theorem shows that the estimates in Ogata's theorem are indeed useful for $d = 2$ and so the speculative musings of this paragraph are defeated in this case.

A mathematical formulation of the ``macroscopic measurements'' in the above quote are macroscopic observables as defined and discussed below. See Section II B. of \cite{Macroscopic} for more about this. In appendix D of \cite{Ogata}, Ogata provides a generalization of Ogata's theorem for translation invariant local interactions for a quantum spin system. Different generalizations are also possible.

We phrase the result in terms of the linear operators $T_N:M_d(\C)\to M_{d^N}(\C)$ defined by:
\[T_N(A) = \frac{1}{N}\sum_{k=0}^{N-1}I_d^{\otimes (N-1-k)}\otimes A \otimes I_d^{\otimes k}.\]
A self-adjoint matrix $A$ on $\C^d$ can be viewed as an observable for a small finite dimensional system, so then $T_N(A)$ is a normalization of the observable for many copies of this small system. Alternatively, one can distribute the factor of $1/N$ so as to view each of the small systems as having the observable $\frac1NA$ and the macroscopic observable being $T_N(A)$, as discussed in \cite{Macroscopic}.

We now list some properties of $T_N$.
When $A$ is diagonalizable, we see that \begin{align}\label{T_N spectrum}
\sigma(T_N(A)) = \frac{1}{N}\sum_{k=0}^{N-1}\sigma(A).
\end{align}
Thus, the spectrum of $T_N(A)$ is a discrete approximation of the convex hull of $\sigma(A)$. $T_N$ also satisfies $T_N(A^\ast) = T_N(A)^\ast$, $T_N(A^T) = T_N(A)^T$, and $T_N(UAU^{\ast}) = U^{\otimes N}T_N(A)U^{\ast\otimes N}$. There is additionally a symmetry due to permuting the tensor product factors. Note that $T_N$ is not multiplicative.

Regardless, because of Equation (\ref{T_N spectrum}), $\|T_N(A)\| = \|A\|$ when $A$ is normal and in general $\|T_N(A)\| \leq \|A\|$ by definition. Applying
\[\|A\| \leq \|\Re A\| + \|\Im A\| \leq 2\|A\|\] to $T_N(A)$, we see that
\begin{align}\label{norm-equiv}
\frac{1}{2}\|A\|\leq \|T_N(A)\| \leq \|A\|.
\end{align}

Because 
\[\left[I_d^{\otimes (N-1-j)}\otimes A \otimes I_d^{\otimes j},I_d^{\otimes (N-1-k)}\otimes B \otimes I_d^{\otimes k}\right] = \left\{ \begin{array}{ll} I_d^{\otimes (N-1-k)}\otimes [A,B] \otimes I_d^{\otimes k}, & j = k\\
0, & j \neq k\end{array}\right.,\]
we see that \begin{align}\label{T_Ncommutator}
[\,T_N(A),T_N(B)\,] = \frac{1}{N}T_N(\,[A,B]\,).
\end{align}

So, given any bounded collection of matrices in $M_d(\C)$, applying $T_N$ provides  sequences of almost commuting matrices for $N \to \infty$. 
Two almost commuting self-adjoint matrices are nearby commuting self-adjoint matrices by Lin's theorem. The analogous statement is not true for more than two almost commuting matrices as discussed in the introduction.
However, Ogata's theorem (Theorem \ref{OgataThm}) provides an extension of Lin's theorem in this special case of arbitrarily many  macroscopic observables.

\begin{remark}
Note that the statement of Ogata's theorem in \cite{Ogata} is for $N = 2n + 1$. However, because \[T_{N+1}(A) = \frac{N}{N+1}  T_N(A)\otimes I_d + \frac{1}{N+1}I_d^{\otimes(N-1)}\otimes A,\] having shown the existence of nearby commuting matrices for $N$ odd, it follows for $N+1$ by choosing \[Y_{i,N+1} = \frac{N}{N+1}Y_{i,N}\otimes I_d.\]
This gives us the formulation we have in the introduction.
\end{remark}

\vspace{0.05in}

We now outline an approach to proving Ogata's theorem. 
Although the result that we prove using this method is for $d = 2$, we only assume this in the discussion below when necessary.

It is sufficient to prove Ogata's theorem for self-adjoint $A_1, \dots, A_k$ being a $\C$-basis for $M_d(\C)$.
In particular, constructing nearby commuting matrices is only an interesting problem for $k \leq d^2$ due to the following reduction. Suppose that the $A_i$ are linearly independent and that we can find nearby commuting matrices $Y_{i, N}$ for $T_N(A_i)$.
If we have a matrix $A \in M_d(\C)$
that can be expanded as $A = \sum_{i=1}^k c_i A_i$
then define 
\begin{align}\label{Ydef}
Y_{N}(A) = Y_{N}\left(\sum_{i=1}^k c_i A_i\right)= \sum_{i=1}^k c_i Y_{i,N}.
\end{align}

We then see that for any $A, B\in M_d(\C)$ in the span of the $A_i$, it holds that $Y_{N}(A)$ and $Y_{N}(B)$ commute. If the constructed $Y_{i,N}$ are self-adjoint, then  $Y_{N}(A)$ is self-adjoint whenever $A$ is. 
Moreover, because all norms on finite dimensional spaces are equivalent, there is a constant $C$ only depending on the $A_i$ such that 
\begin{align}\label{Yestimate}
\|T_N(A) - Y_{N}(A)\| \leq \max_{1\leq i \leq k}\|T_N(A_i)-Y_{i,N}\|\sum_{i=1}^k |c_i| \leq \left(C \max_{1\leq i \leq k}\|T_N(A_i)-Y_{i,N}\|\right)\|A\|
\end{align}
converges to zero uniformly as $N \to \infty$ for $\|A\|$ bounded. Because $T_N(I_d) = I_{d^N},$ if $A_i$ for $i = i_0$ is a multiple of the identity, then we need only focus on constructing nearby commuting matrices for the other $A_i$ and can ignore $i=i_0$ in $\sum_i |c_i|$.

We now specialize to the case $d = 2$. 
We choose the specifically useful basis $A_i$ of $M_2(\C)$ given by $\sigma_1, \sigma_2, \sigma_3, \frac12 I$, where we use the following convention for the Pauli spin matrices:
\[
\sigma_1= \frac12\bp 0 & 1\\1&0  \ep,\;\; \sigma_2 = \frac12\bp 0 & i\\ -i &0 \ep,\;\; \sigma_3=\frac12\bp -1&0\\0&1 \ep.
\]

For any $A \in M_2(\C), $  write $A = \sum_i c_i A_i$. 
Using a well-known identity for the norm of the trace-free self-adjoint matrix $c_1\sigma_1+ c_2\sigma_2+c_3\sigma_3$, we have \begin{align}\label{pauliNorm}\|A\| =\frac12\sqrt{|c_1|^2+|c_2|^2+|c_3|^2} + \frac{|c_4|}2 .\end{align}
So, by the Cauchy-Schwartz inequality,
\begin{align}
\nonumber
\sum_i|c_i| \leq \sqrt{3}\sqrt{|c_1|^2 + |c_2|^2 + |c_3|^2} + |c_4| \leq 2\sqrt{3}\|A\| .\end{align}
This inequality is sharp exactly when $|c_1| = |c_2| = |c_3|$ and $c_4 = 0$. This gives $C = 2\sqrt{3}$ in Equation (\ref{Yestimate}). In our proof of Theorem \ref{OgataTheorem}, we will have that $\|T_N(A_i) - Y_{i,N}\|$ for $i = 1$, $2$ is much larger than this expression for $i=3$, so we will obtain a value of $C$ close to $2\sqrt2$.

Because we chose $A_4 = \frac12I_2$,  we only need to construct nearby commuting matrices for $A_i$ being $\sigma_1, \sigma_2, \sigma_3$, as stated above. (We would not include $|c_4|$ in Equation (\ref{Yestimate}) in this case.)
So, we then focus on constructing nearby commuting matrices for $T_N$ applied to $\sigma_3=A_3$ and $\sigma_+ = A_1 + iA_2$.

\vspace{0.1in}

\noindent\underline{Outline of Construction}: 
The key perspective used to construct nearby commuting matrices for $T_N(\sigma_3)$ and $T_N(\sigma_+)$ is to use the representation theory of $su(2)$ discussed in Section \ref{Prelim}. For any irreducible representation $S^\lam$ of $su(2)$, we have $S^\lam(\sigma_3)$ and $S^\lam(\sigma_+)$ given explicitly as a diagonal and a weighted shift matrix. The distribution of the multiplicities of the irreducible subrepresentations $S^\lam$ in the tensor representation $(S^{1/2})^{\otimes N}$ is discussed in Lemma \ref{1/2mult}. Because $T_N = \frac{1}{N}(S^{1/2})^{\otimes N}$, this simultaneously gives $T_N(\sigma_3)$ as a direct sum of diagonal matrices and $T_N(\sigma_+)$ as a direct sum of weighted shift matrices, up to a unitary change of basis. 

Section \ref{GEL-Section} and Section \ref{ANWS} discuss the needed results for weighted shift matrices in preparation for the gradual exchange process, which is the purpose of Section \ref{GEP-Section}. This construction is more general than the context of the proof of Ogata's theorem. Suppose that $A_r \in M_{n_r}(\C)$ are diagonal and $S_r \in M_{n_r}(\C)$ are weighted shift matrices, where the eigenvalues of the diagonal matrices $A_r$ have a certain nested structure.
The gradual exchange process lemma (Lemma \ref{gep}) provides a construction of nearby commuting matrices for $A = \bigoplus_r A_r$ and  $S = \bigoplus_r S_r$.  The next two paragraphs go into some more detail about the results used in this lemma.

Lemma \ref{gep} is built up through Lemma \ref{proto-gep} and Lemma \ref{proto-gep2}, which construct almost invariant subspaces that are localized with respect to the spectrum of $A$ and are almost invariant under $S$ in a particular way. Lemma \ref{proto-gep} is proved by building a braided pattern of exchanges using Berg's gradual exchange lemma (Lemma \ref{GELws}) for the direct sum of two weighted shift matrices.  Lemma \ref{proto-gep2} generalizes Lemma \ref{proto-gep} by handling the case that not all the diagonal matrices $A_r$ have the same size.

The subspaces constructed in Lemma \ref{proto-gep2} are used in the proof of Lemma \ref{gep} to construct nearby commuting matrices $A'$ and $S'$.
Berg's construction of a nearby normal matrix for an almost normal weighted shift matrix (the focus of Section \ref{ANWS}) is used in this last step to construct $S''$ from $S'$. 
For this last step, it is used that the matrices $A'$ and $S'$ constructed are actually a direct sum of diagonal matrices and a direct sum of weighted shift matrices, though with a different basis than $A$ and $S$ are expressed as a direct sum and with a different block structure.

Section \ref{MainTheorem-Section} is focused on completing the construction of nearby commuting matrices for $\frac{1}{N}S(\sigma_3)$ and  $\frac{1}{N}S(\sigma_+)$ for various reducible representations $S$ of $su(2)$.
Using various estimates for the entries of $S^\lam(\sigma_+)$ gotten in Lemma \ref{d-ineq}, Lemma \ref{gep} is directly applied  to obtain in Lemma \ref{Snearby}. 
Given certain estimates for the irreducible representations making up $S$, this lemma provides a construction of commuting matrices $A'$ self-adjoint and $S''$ normal nearby $\frac1NS(\sigma_3)$ and $\frac1NS(\sigma_+)$. This then provides commuting self-adjoint $\Re(S''), \Im(S''), A'$ nearby $\frac1NS(\sigma_1)$, $\frac1NS(\sigma_2)$, $\frac1NS(\sigma_3)$.

Work done in Example \ref{Ex1} is collected into Lemma \ref{Ex2Lemma} which is then  optimized and extended to cover trivial cases as Lemma \ref{bigLstepLemma}.
This lemma provides nearby commuting matrices when $S = S^{\lam_1}\oplus \cdots\oplus S^{\lam_m}$ has an optimized fixed spacing between the $\lam_i$. 
By breaking up more natural reducible representations  into direct sums of representations of this form, one obtains the main theorem (Theorem \ref{mainthm}). From that we obtain Ogata's Theorem for $d=2$ stated in the introduction (Theorem \ref{OgataTheorem}).

\section{Representation Theory Preliminaries}\label{Prelim}

Here we review some of the standard properties of representations of the lie algebra $su(2)$ as well as some further properties of these representations that will be useful later. The standard material can be found in \cite{Hall} or \cite{RepPhys}. All lie algebra representations discussed will be assumed to be skew-Hermitian, coming from unitary representations of $SU(2)$. All direct sums are orthogonal. 

Consider the Pauli spin matrices (with eigenvalues $\pm1/2$) with the convention that $\sigma_3$ is diagonal with increasing eigenvalues:
\begin{align}\label{pSpin}
\sigma_1= \frac12\bp 0 & 1\\1&0  \ep,\;\; \sigma_2 = \frac12\bp 0 & i\\ -i &0 \ep,\;\; \sigma_3=\frac12\bp -1&0\\0&1 \ep.
\end{align}
These matrices span, with real coefficients, the trace-free self-adjoint matrices in $M_2(\C)$. The Pauli spin matrices satisfy the commutation relations
\[[\sigma_i, \sigma_j] = i\sum_k \epsilon_{ijk}\sigma_k,\]
where 
\[\epsilon_{ijk} =
\left\{\begin{array}{ll} \sgn (i\, j\, k), & i, j, k \mbox{ are distinct} \\
0, & \mbox{otherwise}\end{array}\right..
\]
Note also that the $\sigma_i$ anticommute:
\[\sigma_1\sigma_2+\sigma_2\sigma_1 = \sigma_2\sigma_3+\sigma_3\sigma_2 = \sigma_3\sigma_1+\sigma_1\sigma_3 = 0.\]

An arbitrary element of $su(2)$ can be represented as $i$ multiplied by the self-adjoint $c_1\sigma_1 + c_2\sigma_2 + c_3\sigma_3$ for $c_i\in\R$. This is the so-called defining representation of $su(2)$. By removing a factor of $i$, any representation $S$ of $su(2)$ is equivalent to a linear map $\tilde{S}$ defined on the $\C$-span of $\sigma_1, \sigma_2, \sigma_3$ with the same commutation relations
\[\left[\tilde S(\sigma_i), \tilde S(\sigma_j)\right] = i\sum_k \epsilon_{ijk}\tilde S(\sigma_k).\]
So, we identify any representation $S$ of $su(2)$ with its linear extension linear $\tilde{S}$.

Up to unitary equivalence, there is a unique irreducible representation of $su(2)$ of each dimension.  
For $\lam$ a non-negative integer or half-integer, the unique irreducible representation $S^\lambda$ on $\C^{2\lam+1}$ can be explicitly expressed as follows. 

Let $\sigma_+ = \sigma_1 + i\sigma_2$ and $\sigma_- = \sigma_+^\ast$. Note that
\begin{align}
\nonumber
\sigma_+= \bp 0 & 0\\1&0  \ep,\;\; \sigma_- = \bp 0 & 1\\ 0 &0 \ep.
\end{align}
Let \[d_{\lam,m} = \sqrt{\lam(\lam+1)-m(m+1)}=\sqrt{(\lam-m)(\lam+m+1)}, \;\; -\lam \leq m < \lam.\] The condition that $\lam$ and $m$ are both integers or both half-integers will be expressed as $\lam - m \in \Z$. Then 
\[S^{\lam}(\sigma_3) = 
\bp 
-\lam& & & & \\
 &-\lam+1& & & \\
 & & -\lam + 2 & &\\
 & & &\ddots& \\ 
 & & & & \lam \\
\ep, \;\; S^\lam(\sigma_+) =  
\bp 
0  & & & &\\
d_{\lam,-\lam}&0& & &\\
 &d_{\lam, -\lam+1}&0& &\\
 & &\ddots& \ddots&\\
 & & &d_{\lam,\lam-1}&0
\ep\]
and $S^{\lambda}(\sigma_-) = S^{\lambda}(\sigma_+)^\ast$. Then extend $S^\lam$ to $su(2)$ by linearity. 

In particular, if $v_{-\lam}, \dots, v_{\lam}$ are the standard basis vectors for $\C^{2\lam + 1}$, then 
\[S^{\lam}(\sigma_3)v_m = mv_m, \;\; S^{\lam}(\sigma_+)v_{m} = d_{\lam,m}v_{m+1}, \;\; S^{\lam}(\sigma_-)v_{m} = d_{\lam,m-1}v_{m-1}.\]
The trivial representation $S^{0}$ on $\C^1$ is given by $S(\sigma_i) = 0$.
The first nontrivial irreducible representation is the $2(1/2)+1=2$ dimensional representation $S^{1/2}$, the ``defining representation'', given by $S^{1/2}(\sigma_i) = \sigma_i$.

It is important to note that in representation theory $\lam$ is often called the ``weight'' of the representation $S^\lam$. However due to our usage of the term ``weight'' in Definition \ref{weightdef}, we will instead always refer to $d_{\lam, i}$ as the weights of the weighted shift matrix $S^\lam(\sigma_+)$ and will not refer to $\lam$ as a ``weight''. To distinguish between these two usages, we will use the common physics terminology that $S^\lam$ is ``the irreducible spin-$\lam$ representation'' if necessary. 

We now proceed to discuss some of the properties of the weights $d_{\lam, i}$ of the representation $S^\lam$. Note that
\begin{align}
\label{lamnorm}
\|S^{\lam}(\sigma_i)\|=\lam, \;
\|S^{\lam}(\sigma_{\pm})\| 
\leq 2\lam.
\end{align}
In particular, we see that for $i\neq j$, $\frac1\lam S^{\lam}(\sigma_i), \frac1\lam S^{\lam}(\sigma_j)$ are almost commuting with
\begin{align}\label{repAlmostCommuting}
\left\|\left[\frac1\lam S^{\lam}(\sigma_i), \frac1\lam S^{\lam}(\sigma_j) \right]\right\| =  \frac{1}{\lam}.
\end{align}

We state some estimates concerning the weights $d_{\lam, i}$ in the following lemma.
In particular, $(i)$ below provides a refinement of the bound of $\|S^\lam(\sigma_{\pm})\| = \max_i d_{\lam, i}$ in Equation (\ref{lamnorm}).
\begin{lemma}\label{d-ineq}
Suppose that $|i|\leq \mu \leq \lam$ are such that $\lam-i, \mu-i \in \Z$.
\begin{enumerate}[label=(\roman*)]
\item We have 
\begin{align*}
d_{\mu,i}\leq d_{\lam,i}\leq  \lam + \frac12 \leq 2\lam.
\end{align*}

\item If $\lam - |i| \leq M$ then
\begin{align*}
d_{\lam, i}\leq \sqrt{2\lam(M+1)}.
\end{align*}

\item
If $|\lam - \mu| \leq L$ then
\begin{align*}
d_{\lam, i}-d_{\mu,i}\leq \sqrt{2\lam L}.
\end{align*}

\item If $|\lam-\mu| \leq L$ and $l > 0$ is given, then at least one of  \[d_{\lam, i} - d_{\mu, i} \leq \sqrt{\lam}\frac{2L}{\sqrt{l}},\] \[\d_{\lam, i} \leq \sqrt{2\lam(l+1)}\]
hold. Consequently,
\begin{align}\label{d-Gbound}
d_{\lam, i} &- d_{\mu, i} + C \max(d_{\lam, i}, d_{\mu, i})\nonumber
\\ &\leq \max\left(\sqrt{\lam}\frac{2L}{\sqrt{l}}+C(\lam + 1/2), \sqrt{2\lam L}+C\sqrt{2\lam(l+1)}\right)
\end{align}

\item If $|\lam - \mu| \leq L$ then
\begin{align*}
d_{\lam,i}^2-d_{\mu,i}^2 \leq 2\lam L.
\end{align*}

\item $\|\, [S^\lam(\sigma_+)^\ast, S^\lam(\sigma_+)] \,\| = 2\lam.$

\end{enumerate}
\end{lemma}
\begin{remark}
For a fixed $\lam$, the graph of $d_{\lam, i}$ as a function of $i$ are points on a semicircle with center $-1/2$ and radius about $\lam+1/2$. See Figure \ref{All}. 
\begin{figure}[htp]  
    \centering
    \includegraphics[width=8cm]{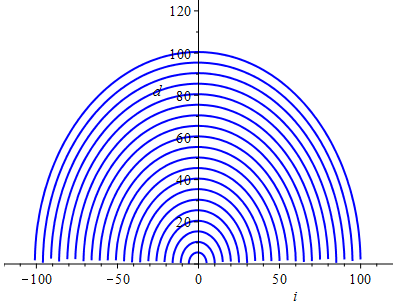}
    \caption{\dark \label{All} For each $\lam$, the points of $(i, d_{\lam, i})$ for $i = -\lam, -\lam+1, \dots, \lam-1$ all lie on a single semicircle. This Figure is an illustration of the weights $d_{\lam, i}$ for $\lam = 5, 10, \dots, 100$. }
\end{figure}
The maximum value of $d_{\lam,i}$ is asymptotically $\lam$, however it is always bounded by $2\lam$. This is $(i)$.
When $|i|$ is close to $\lam$, $d_{\lam, i}$ is small. This is $(ii)$. In other words, near the boundary of the circle, the weights are comparable to a smaller power of $\lam$. In particular, if $i = -\lam$ or $i = \lam-1$, $d_{\lam,i} = \sqrt{2\lam}$. 

When $\mu$ is close to $\lam$ then $d_{\lam, i}-d_{\mu,i}$ is small compared to $\lam$. However, if we put a separation of $M$ between $|i|$ and $\lam$ then this difference can be made smaller since it corresponds to taking the difference between values of consecutive semicircles away from the edges of the semicircles. This is the Claim in the proof. See Figure \ref{Diff}.
\begin{figure}[htp]  
    \centering
    \includegraphics[width=8cm]{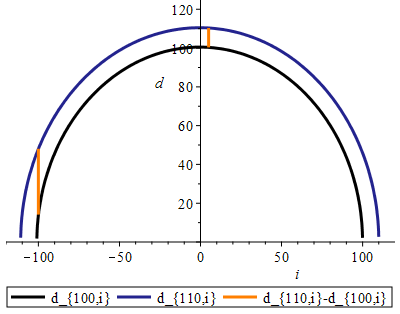}
    \caption{\label{Diff}\dark Illustration of the difference $d_{\lam, i}-d_{\mu,i}$ when $i$ is close to $-\mu$ and when $i$ is much smaller than $\mu$. Note that this difference is the \emph{vertical} distance between the arcs, not the radial distance.}
\end{figure}

As stated above, when $i$ corresponds to a point away from the boundary of the semicircle, one obtains an improved estimate for the differences of weights. When $i$ corresponds to a point near the boundary of the semicircle, one obtain an improved estimate for the size of the weight. This is $(iv)$.
As above, all notions of ``small'' or ``close'' should be interpreted in terms of the size of $\lam$. In particular $\sqrt{2\lam L}$ is much smaller than $\lam$ when $L$ is much smaller than $\lam$.

The similarity between $(ii)$ and $(iii)$ is due to the fact that $d_{\lam,i} \leq d_{\lam,\mu} + d_{\mu,i}$ because $d_{\lam,i}^2 = d_{\lam,\mu}^2 + d_{\mu,i}^2$. So, a bound for $d_{\lam, \mu}$ gives a bound for the difference $d_{\lam,i} - d_{\mu,i}$. This can be seen in the proof. 
Also, the pervasive ``$+1$'' is due to the small asymmetry of the terms $d_{\lam, i}$ with respect to $i \mapsto -i$. 
\end{remark}
\begin{proof}
\begin{enumerate}[label=(\roman*)]
\item  The first inequality follows since 
\[d_{\lam, i}^2 = \left(\lam(\lam+1)+\frac14\right) - \left(i(i+1)+\frac14\right) = \left(\lam+\frac12\right)^2-\left(i+\frac12\right)^2.\]
So, one obtains $\max_i d_{\lam, i} \leq \lam + 1/2$ with equality when $\lam$ is a half-integer.

\item If $0 \leq i < \lam$ then
\[d_{\lam,i} = \sqrt{(\lam-|i|)(\lam+|i|+1)}\leq \sqrt{M(2\lam)}.\]
If instead $-\lam \leq i < 0$ then $i = -|i|$ so \[d_{\lam,i} = \sqrt{(\lam+|i|)(\lam-|i|+1)}\leq \sqrt{2\lam(M+1)}.\]

\item 
If $\lam = \mu$ then the stated inequality is trivial so suppose that $\mu < \lam$. We calculate
\begin{align} 
\nonumber
d_{\lam, i}&= \sqrt{\lam(\lam+1)-i(i+1)} \leq \sqrt{\lam(\lam+1)-\mu(\mu+1)} + \sqrt{\mu(\mu+1)-i(i+1)}\\
&= \sqrt{(\lam-\mu)(\lam+\mu+1)} + d_{\mu,i} \leq \sqrt{L(2\lam)} + d_{\mu,i}.\nonumber
\end{align}
So, we obtain the desired inequality.

\item 
Given the Claim below, choose $M = l$. If $\lam - |i|\leq l$ then we obtain the second inequality by $(ii)$ above. If $\lam - |i| > l$ then we obtain the first inequality by the Claim below. To obtain Equation (\ref{d-Gbound}), we apply the same case analysis along with the unconditional bounds in $(i)$ and $(iii)$.
So, we only need to show:

\noindent \underline{Claim}:
Suppose  $|\lam - \mu| \leq L$. If  $\lam - |i| > M$ then
\begin{align}
\nonumber
d_{\lam, i}-d_{\mu,i}\leq \sqrt{\lam}\frac{ 2L}{\sqrt{M}}.
\end{align}

\noindent \underline{Proof of Claim}:
As before, suppose $\mu < \lam$. We calculate
\[d_{\lam, i} - d_{\mu, i} = \frac{d^2_{\lam, i} - d^2_{\mu, i}}{d_{\lam, i}+d_{\mu, i}}
= \frac{\lam(\lam+1)-\mu(\mu+1)}{d_{\lam, i}+d_{\mu, i}}
\leq \frac{(\lam-\mu)(\lam+\mu+1)}{d_{\lam, i}}.\]
Suppose $\lam-|i| > M$. If $i \geq 0$ then 
\[d_{\lam,i}  = \sqrt{(\lam+|i|+1)(\lam-|i|)} > \sqrt{\lam M}.\]
If $i < 0$ then
\[d_{\lam,i}  = \sqrt{(\lam-|i|+1)(\lam+|i|)} > \sqrt{M\lam}.\]

So,
\[d_{\lam, i} - d_{\mu, i} < \frac{L(2\lam) }{\sqrt{\lam M}} = \frac{2\sqrt{\lam} L}{\sqrt{M}}.\]

\item We have
\begin{align*}
d_{\lam,i}^2-d_{\mu,i}^2 &= \lambda(\lambda+1)-i(i+1) - \mu(\mu+1)+i(i+1) \\
&= (\lam+\mu+1)(\lam-\mu).
\end{align*}
If $\lam = \mu$ then $(\lam+\mu+1)(\lam-\mu) = 0 < 2\lam$. If $\mu < \lam$ then $\lam+\mu+1 \leq 2\lam$.

\item For $-\lam \leq i  \leq  \lam-2$, 
\[|d_{\lam, i+1}^2-d_{\lam, i}^2| = |(i+1)(i+2)-i(i+1)| = 2|i+1| \leq 2\lam.\]
Also, 
\[d_{\lam, -\lam}^2 = d_{\lam, \lam-1}^2 = 2\lam.\]
So, \[\|\,[S^\lam(\sigma_+)^\ast, S^\lam(\sigma_+)]\,\| = \max\left(d_{\lam, -\lam}^2,\, \max_{-\lam \leq i \leq \lam-2}|d_{\lam, i+1}^2-d_{\lam, i}^2|,\, d_{\lam, \lam-1}^2\right) = 2\lam.\]

\end{enumerate}
\end{proof}

We now recall some general properties of the tensor products of the irreducible representations of $su(2)$.
The reason we are interested in this is that if we have two representations $S_1$ on $\C^{n_1}$ and $S_2$ on $\C^{n_2}$, then their tensor product representation is expressed as
\[S_1\otimes S_2(\sigma_i) = S_1(\sigma_i)\otimes I_{n_2} + I_{n_1}\otimes S_2(\sigma_i).\] So, we can view $T_N(\sigma_i)$ in the statement of Ogata's theorem as the scaled matrix tensor product $\displaystyle\frac{1}{N}(S^{1/2})^{\otimes N}(\sigma_i)$. 
From this perspective, understanding how to break down this tensor product representation into irreducible representations will give us a handle on some of the underlying structure of $T_N(\sigma_i)$. 

Suppose that $\lambda_1 \leq \lambda_2$. Then the tensor product representation satisfies
\[S^{\lambda_2}\otimes S^{\lambda_1} \cong S^{\lambda_2-\lambda_1}\oplus S^{\lambda_2-\lambda_1+1}\oplus \cdots \oplus S^{\lambda_2 + \lambda_1}.\]
This means that there is a unitary matrix $U$ such that for all $i$, \[U^\ast\left(S^{\lambda_2}\otimes S^{\lambda_1}(\sigma_i)\right)U = S^{\lambda_2-\lambda_1}(\sigma_i)\oplus S^{\lambda_2-\lambda_1+1}(\sigma_i)\oplus \cdots \oplus S^{\lambda_2 + \lambda_1}(\sigma_i).\]
The unitary matrix can be expressed in terms of Clebsch-Gordan coefficients. These coefficients can be chosen to be real. Algorithms for the calculation of such coefficients have been well-studied. See for instance \cite{CG}. 

The repeated tensor product of representations can be gotten by using this result along with standard manipulations of tensor products. In particular,
\begin{align*}
(S^{1/2})^{\otimes3} &\cong S^{1/2}\otimes(S^{1/2}\otimes S^{1/2}) \cong S^{1/2}\otimes(S^0 \oplus S^1) \cong S^{1/2}\otimes S^0 \oplus S^{1/2}\otimes S^{1} \\
&\cong S^{1/2}\oplus S^{1/2} \oplus S^{3/2} \cong 2S^{1/2}\oplus S^{3/2}.
\end{align*}

So, we see that $S^{1/2}$ has multiplicity $2$ and $S^{3/2}$ has multiplicity $1$ in the decomposition of the tensor representation into irreducible representations. By similar calculations, the representation $S=(S^{\lambda})^{\otimes N}$ can be calculated explicitly in terms of Clebsch-Gordan coefficients for any value of $N$. With that as a given, we focus on the distribution of multiplicities that occur when we write such a tensor representation as a direct sum of irreducible representations for general $N$.

Recall that the eigenvalues of $S^{\lambda}(\sigma_3)$ are $-\lambda, \dots, \lambda$.
By analyzing this, we obtain the following standard property that is used in the proof of the tensor product property given above. (See Theorem C.1 of \cite{Hall}.)
Observe that $(S^{1/2})^{\otimes N}$ is a direct sum of irreducible representations $S^{\lambda}$ where all the $\lambda$ are integers if $N$ is even and all the $\lambda$ are half-integers if $N$ is odd. In particular, the eigenvalues of $(S^{1/2})^{\otimes N}(\sigma_3)$ will be integers if $N$ is even and will be half-integers if $N$ is odd. 
\begin{lemma}
Suppose that $S = n_0S^{0}\oplus n_{1/2}S^{1/2}\oplus \cdots n_{k}S^{k}$ is a representation of $su(2)$. Then the multiplicity of the eigenvalue $m$ of $S(\sigma_3)$ is $\sum_{i \geq 0} n_{|m|+i}$, where the sum is over integral $i$.

Conversely, if the eigenvalue $m$ of $S(\sigma_3)$ has multiplicity $k_m=k_{|m|}$ then the representation multiplicities $n_j$ can be reconstructed as $n_{m} = k_{|m|} - k_{|m|+1}$.
\end{lemma}
\begin{proof}
For the first statement, the eigenvalues of $S^{\lam}(\sigma_3)$ are $-\lambda, \dots, \lambda$. So, $S^{\lam}(\sigma_3)$ has an eigenvalue $m$ if $|m| \leq \lam$ and $\lam - m$ is an integer. Therefore, there is a non-negative integer $i$ such that $\lam = |m| + i$. Because such eigenvalues appear with multiplicity one, the first result then follows.

The converse follows directly from the first part.
\end{proof}

A simple way to express the multiplicities of eigenvalues is to identify the representation $S^{\lam}$ with the polynomial $x^{-\lam} + x^{-\lam + 1} + \cdots + x^{\lam-1} + x^{\lam}$ in the variables $x^{1/2}, x^{-1/2}$. The coefficient of the $x^m$ term is the multiplicity of the eigenvalue $m$ of $S^{\lam}(\sigma_3)$. When performing the direct sum of representations, this corresponds to adding the respective polynomials. The correspondence remains valid because the multiplicities and coefficients both add.
Likewise, the product of the polynomial corresponding to irreducible representations corresponds to tensor products of the irreducible representations. To see this consider the case that $j_1 \leq j_2$:
\begin{align*}
(x^{-j_1}&+x^{-j_1+1} + \cdots x^{j_1-1}+x^{j_1})(x^{-j_2}+x^{-j_2+1} + \cdots x^{j_2-1}+x^{j_2})\\
&= (x^{-j_1-j_2}+\cdots+x^{j_1-j_2})+(x^{-j_1-j_2+1}+\cdots+x^{j_1-j_2+1})+(x^{-j_1-j_2+2}+\cdots+x^{j_1-j_2+2})\\
&\;\;\;\;\;\;+\cdots +(x^{-j_1+j_2}+\cdots+x^{j_1+j_2})\\
&=x^{-j_1-j_2}+2x^{-j_1-j_2+1}+\cdots+(2j_1+1)x^{j_1-j_2}+(2j_1+1)x^{j_1-j_2+1}\\
&\;\;\;\;\;\;+\cdots+(2j_1+1)x^{j_2-j_1-1}+(2j_1+1)x^{j_2-j_1}+\cdots+2x^{j_1+j_2-1}+x^{j_1+j_2}\\
&=(x^{-j_1-j_2}+\cdots+x^{j_1+j_2})+(x^{-j_1-j_2+1}+\cdots+x^{j_1+j_2-1})+\cdots+(x^{j_1-j_2}+\cdots+x^{-j_1+j_2}).
\end{align*}
Hence, by the distributive property of multiplication and tensor products, the algebraic identification holds for all such polynomials.
This provides a method to easily calculate the multiplicities of the representations for computer algebra systems and also a simple closed form expression for $(S^{1/2})^{\otimes N}$. 

In particular, taking powers of $x^{-1/2} + x^{1/2}$ and using the binomial formula gives the following result. We interpret $\binom{N}{s}$ to be zero if $s$ is not an integer in $[0, N]$ and summations of the form $\sum_{k=a}^b$ where $b-a \in \Z$ to be the sum over $k=a, a+1 \dots, b$.
\begin{lemma}
For $0\leq \lam = N/2, N/2-1, \dots,$ the multiplicity of $S^\lam$ in $(S^{1/2})^{\otimes N}$ is \[\binom{N}{\lam+N/2} - \binom{N}{\lam+1+N/2}.\]
\end{lemma}
\begin{proof}
We calculate
\begin{align*}
(x^{-1/2} + x^{1/2})^{N} = \sum_{k=0}^N\binom{N}{k}x^{-\frac{N-k}{2}}x^{\frac{k}{2}} = \sum_{k=0}^N\binom{N}{k}x^{k - N/2} = \sum_{m=-N/2}^{N/2}\binom{N}{m+N/2}x^m.
\end{align*}
So, the multiplicity of the $S^\lam$ representation is
$\binom{N}{\lam+N/2} - \binom{N}{\lam+1+N/2}$. 
\end{proof}

Using the previous result, we can then investigate the behavior of the multiplicities. A graph of the multiplicities for $N = 1000$ is depicted in Figure \ref{5}.
\begin{figure}[htp]  
    \centering
    \includegraphics[width=8cm]{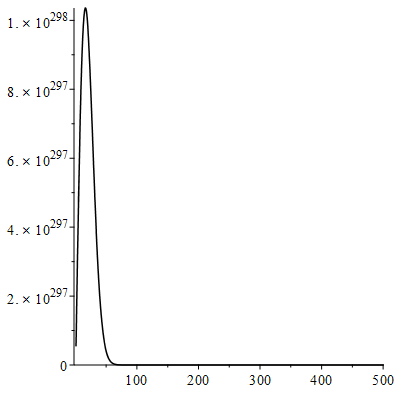}
    \caption{\label{5}\dark Illustration of multiplicities of $S^{\lam}$ for $\lam \in \Z$ of $(S^{1/2})^{\otimes N}$ for $N = 1000$. For this value of $N$, $\sqrt{N}/2 \approx 15.8$.}
\end{figure}

In particular, the multiplicities are increasing until the inflection point of the binomial distribution then afterward it decreases. Although numerical explorations suggest a rapid decrease of the multiplicities, since we are only investigating the operator norm, our method will only involve using that the coefficients strictly decrease after $O(\sqrt{N})$. A further discussion of properties of differences of binomial coefficients can also be found in \cite{Bin}, which influenced the statement of the following.

\begin{lemma}\label{1/2mult}
The multiplicity of $S^\lam$ in $(S^{1/2})^{\otimes N}$ is zero if $2\lam$ has a different parity than $N$. For $2\lam$ having the same parity as $N$, the multiplicity $n_\lambda$ of $S^\lam$ satisfies 
\[ 
\left\{\begin{array}{ll}
n_\lam < n_{\lam+1}, & \lam < \lam_\ast \\
n_\lam = n_{\lam+1}, & \lam = \lam_\ast \\
n_\lam > n_{\lam+1}, & \lam > \lam_\ast \\
\end{array}\right.,
\]
where
\[\lam_\ast =  \frac{\sqrt{N+2}}{2}-1 \leq \frac12N^{1/2}.\]
\end{lemma}
\begin{proof}
We use 
\[\binom{n}{k+1} = \frac{n!}{(k+1)!(n-k-1)!} = \binom{n}{k}\cdot\frac{n-k}{k+1}.\]
Therefore, 
\begin{align*}
    \left(\binom{n}{k}-\binom{n}{k+1}\right) &-\left(\binom{n}{k+1}-\binom{n}{k+2}\right) =  \binom{n}{k} - 2\binom{n}{k}\frac{n-k}{k+1} + \binom{n}{k+1}\frac{n-k-1}{k+2} \\
    &= \binom{n}{k}\cdot\left(1-  2\frac{n-k}{k+1}+\frac{n-k-1}{k+2}\cdot\frac{n-k}{k+1}\right)\\
    &= \binom{n}{k}\cdot\frac{(k+2)(k+1)-2(n-k)(k+2)+(n-k-1)(n-k)}{(k+2)(k+1)}\\
    &=\binom{n}{k}\cdot\frac{4k^2+(8-4n)k+(n^2-5n+2)}{(k+1)(k+2)}.
\end{align*}
Finding the (potentially irrational) values of $k$ such that this expression equals zero, we obtain
\[k = \frac{1}{2}n-1 \pm \frac{1}{8}\sqrt{(4n-8)^2-16(n^2-5n+2)} = \frac{1}{2}n-1 \pm \frac{1}{8}\sqrt{16n+32}.\]

By the previous lemma, the difference of coefficient multiplicities is 
 \[n_{\lam+1}-n_{\lam}=\left(\binom{N}{\lam+1+N/2} - \binom{N}{\lam+2+N/2}\right)-\left(\binom{N}{\lam+N/2} - \binom{N}{\lam+1+N/2}\right).\]
 Compared to the calculations above, we have $n=N$ and $k = \lambda + N/2$. So, the multiplicities begin decreasing after $\lam_\ast= \displaystyle\frac{\sqrt{N+2}}{2}-1$ as stated in the statement of the lemma.
\end{proof}

\section{Gradual Exchange Lemma}
\label{GEL-Section}

A key component for the construction in later sections will be I. D. Berg's Gradual Exchange Lemma, sometimes referred to as ``Berg's technique''. 
This method has been used in various arguments to prove results for matrices and also normal and nilpotent operators on a separable Hilbert space (\cite{BergGEL}, \cite{DavidsonRotAlgs}, \cite{MarcouxOrbits}, \cite{MarcouxQD}, \cite{HerreroBlockDiag}). Also, in addition to the proof provided by Loring in \cite{KTheoryand}, Loring remarked that Davidson knew how to use Berg's gradual exchange (by an argument similar to that found in \cite{DavidsonRotAlgs}) to provide a construction of nearby commuting matrices for the modified version of Voiculescu's almost unitaries: $U_n \oplus U_n^\ast, V_n \oplus V_n$.

The lemma has appeared in different forms. A nice paper containing  reflections on the different uses and generalizations (with many diagrams) is Loring's \cite{LoringPseudoActions}.  The argument we present below is a simple modification of Berg's original argument, although recast in terms of perturbing matrix blocks instead of a basis. It is similar to the argument in Lemma 2.1 of \cite{LoringPseudoActions}.
Comparing this with the version stated in \cite{DavidsonRotAlgs}, one sees that the main difference is that the perturbation is real and the constant of the second term of the estimate is $\pi/2$ instead of the usual $\pi$ because we only require that $w_{N_0+1}' = -v_{N_0+1}$ instead of $w_{N_0+1}' = v_{N_0+1}$.

We first give a definition of weighted shift operators.
\begin{defn}\label{weightdef}
Suppose that an orthonormal basis $v_1, \dots, v_n$ is given. We call a linear operator $A$ diagonal with respect to this basis, expressed as $\diag(a_1, \dots, a_n) = \diag(a_i)$, if $Av_i = a_iv_i$. 

We call a linear operator $S$ a weighted shift operator with respect to this basis, expressed as $\ws(c_1, \dots, c_{n-1}) = \ws(c_i)$, if $Sv_i = c_iv_{i+1}$. 
We can express the action of $S$ as:
\begin{align}\label{Sarrows}
S:v_1\overset{c_1}{\rightarrow}v_2\overset{c_2}{\rightarrow}\cdots\overset{c_{n-2}}{\rightarrow}v_{n-1}\overset{c_{n-1}}{\rightarrow}v_n\rightarrow0.
\end{align}
If the basis is not mentioned, the basis is assumed to be the ``standard basis''.

By multiplying the basis vectors by phases, we can choose each $c_i$ to be non-negative. This is discussed in more detail in Example \ref{ws-basisChange}. At this point it need only be said that if all the weights are real, then the phases can be chosen to be $\pm1$. \end{defn}
\begin{defn}
Suppose that $S = \ws(c_1, \dots, c_{n-1}) = \ws(c_i)$ is a weighted shift operator with respect to the basis $v_1, \dots, v_n$.
We refer to the lines spanned by the vectors $v_{k}, v_{k+1}, \dots, v_n$ as the ``orbit'' of $v_k$ under $S$.
We may refer to the vectors $v_k, \dots, v_n$ belonging to the orbit of $v_k$ under $S$.

If all the weights $c_k, \dots, c_{n-1}$ are non-zero, this coincides with the lines: $\spn(v_k)$, $\spn(S v_k)$, $\spn(S^2 v_k)$, $\dots$, $\spn(S^{n-k}v_k)$. In this case, we could call the weighted shift ``irreducible''. 

Note that this definition of orbit digresses from a typical notion of ``orbit'' from Dynamical systems (such as in \cite{DynStystms}, \cite{ErgodicThyandDynSystm}, \cite{RandRecurrDynSystm}) if the weighted shift is not irreducible. In particular, our definition of orbit more closely aligns with what \cite{DynStystms} calls a ``forward-invariant set''. 

In particular, the fact that we have called $v_k, \dots, v_n$ \emph{the} orbit of $v_k$ indicates a choice made when writing $S = \ws(c_1, \dots, c_{n-1})$ as it may be possible to decompose $v_{k}, \dots, v_{n}$ as the disjoint union of orbits of irreducible weighted shift operators.
Generically, the weights $c_i$ will all be non-zero so that this definition coincides with the standard notion of orbit.
\end{defn}

We now describe the diagrams in Figure \ref{WeightedShiftDiagrams}.
\begin{figure}[htp]     \centering
    \includegraphics[width=15cm]{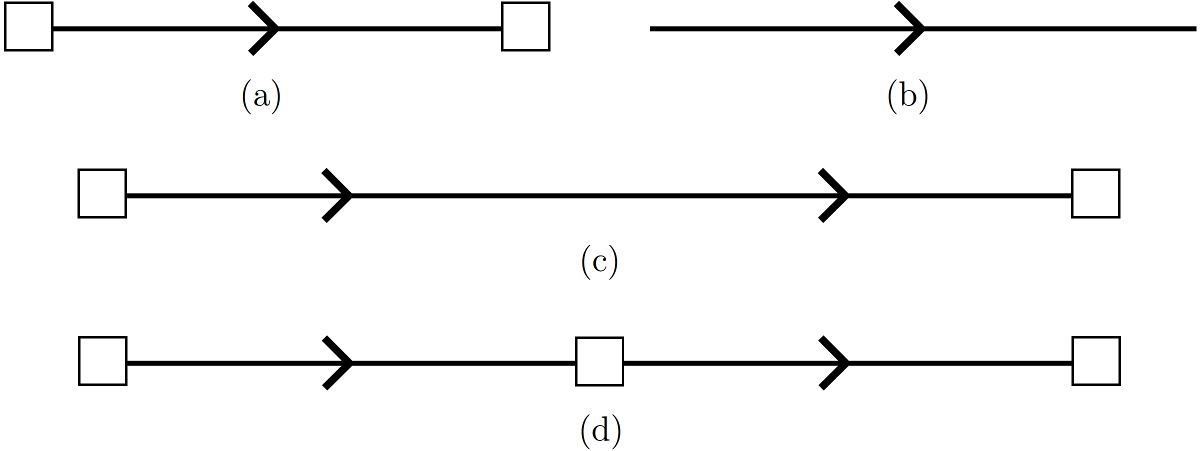}
    \caption{\label{WeightedShiftDiagrams}\dark
    Illustration of several weighted shift diagrams.}
\end{figure}
 Figure \ref{WeightedShiftDiagrams}(a) is an illustration of the weighted shift matrix $S = \ws(c_1, \dots, c_{n-1})$ with respect to the orthonormal vectors $v_1, \dots, v_n$. It can be thought of as a graphical illustration of Equation (\ref{Sarrows}). Moving from left to right along the horizontal line segment corresponds to increasing the index of the vectors $v_i$. The vector $v_1$ is depicted by the square on the left and $v_n$ is depicted by the square on the right. For the purposes of this paper, we illustrate $v_1$ and $v_n$ in the diagram while suppressing explicit depictions of $v_2, \dots, v_{n-1}$. Note that the values of the weights and the size of $n$, while being important, are not illustrated in the diagram either.

Strictly speaking, this weighted shift diagram is a continuous illustration of a discrete system, similar to previous diagrams using Berg's interchange method. See \cite{LoringPseudoActions} for diagrams that are discrete, which involve drawing a point for each $v_i$ and involve a ``$\cdots$'' in numerous places for complicated diagrams. Similar discrete diagrams sometimes appear in illustrations of the irreducible representations of $su(2)$ and other contexts. For instance, see Figure 8.1 of \cite{Woit}, Figure 4.1 and 9.4 of \cite{Hall}, or quivers as in \cite{JourneyThroughRT}. Figure 8.1 of \cite{Woit} illustrates the weighted shift matrix $S^{\lam}(\sigma_+)$ and the diagonal matrix $S^\lam(\sigma_3)$ in the same diagram. 

Figure \ref{WeightedShiftDiagrams}(b) is an illustration of the same weighted shift matrix as \ref{WeightedShiftDiagrams}(a) on a subset $v_{i_0}, \dots, v_{i_1}$ where $1 < i_0 < i_1 < n$. This can be expressed as
\begin{align}\label{Sarrows2}
S:v_{i_0}\overset{c_{i_0}}{\rightarrow}v_{i_0+1}\overset{c_{i_0+1}}{\rightarrow}\cdots\overset{c_{i_1-2}}{\rightarrow}v_{i_1-1}\overset{c_{i_1-1}}{\rightarrow}v_{i_1},
\end{align}
where Equation (\ref{Sarrows2}) only indicates the action of $S$ on the relevant vectors and is silent on whether we are viewing $S$ acting as  a weighted shift starting at $v_{i_0}$ and whether its orbit ends with $v_{i_1}$ or what $Sv_{i_1}$ is. The weighted shift diagram in Figure \ref{WeightedShiftDiagrams}(b) does not include an initial square, indicating that we are not viewing $v_{i_0}$ as initiating a complete orbit (but a sub-orbit). It also does not end in a square, indicating that the orbit of $v_{i_0}$ is not being viewed as ending with $v_{i_1}$.

 Figure \ref{WeightedShiftDiagrams}(c) is an illustration of $S = \ws(c_1, \dots, c_{n-1})$ where $c_{\tilde n} = 0$ for some $1 < \tilde{n} < n-1$. So, $S v_{\tilde n} = 0v_{\tilde n+1}$.
 Note that the diagram itself gives no indication that the $\tilde n$-th weight is zero.
 Also note that the second arrow in the diagram has no additional meaning and is added for aesthetic reasons related to Figure \ref{WeightedShiftDiagrams}(d).  Because of our terminology we view $v_1, \dots, v_n$ as the $S$-orbit of $v_1$ and the illustration reflects this with only having the squares for the first and last vectors.
    
 Figure \ref{WeightedShiftDiagrams}(d) is an illustration of the same operator $S$ as in \ref{WeightedShiftDiagrams}(c), except that we now view $c_{\tilde{n}} = 0$ as breaking $S$ into two weighted shift operators $\ws(c_1, \dots, c_{\tilde{n}-1})$ with respect to the vectors $v_1, \dots, v_{\tilde{n}}$ and $\ws(c_{\tilde{n}+1}, \dots, c_{n-1})$ with respect to the vectors $v_{\tilde n+1}, \dots, v_n$. With this choice of perspective, we view $v_1, \dots, v_{\tilde{n}}$ as the $S$-orbit of $v_1$.
 
 The distinction between (c) and (d) is based on the decision to view $S$ as a single weighted shift matrix
 \[S: v_1 \overset{c_1}{\rightarrow} \cdots \overset{c_{\tilde n -1}}{\rightarrow} v_{\tilde n} \overset{0}{\rightarrow} v_{\tilde n + 1} \overset{c_{\tilde n+1}}{\rightarrow}\cdots \overset{c_{n-1}}{\rightarrow}v_n\rightarrow0\]
 or as a weighted shift on two invariant subspaces:
 \[S: v_1 \overset{c_1}{\rightarrow} \cdots\overset{c_{\tilde n -1}}{\rightarrow} v_{\tilde n} \rightarrow 0, \;\;\; S:  v_{\tilde n + 1} \overset{c_{\tilde n+1}}{\rightarrow}\cdots \overset{c_{n-1}}{\rightarrow}v_n\rightarrow0.\]
 
 We now present our first formulation of the gradual exchange lemma in terms of vectors. We will later formulate this in terms of weighed shift operators.
\begin{lemma}\label{pre-gel}
Let $\{v_k, w_k\}_{k=1,\dots,N_0+1}$ be a collection of orthonormal vectors in a Hilbert space $\mathcal H$ and $S$ be a linear operator on $\mathcal H$ such that for $k = 1, \dots, N_0$, $Sv_k = a_kv_{k+1}$, $Sw_k = b_kw_{k+1}$ for some constants $a_k, b_k$. 

Then there is a linear operator $S'$ such that $S'^{N_0}v_1$ is a multiple of $w_{N_0+1}$, $S'^{N_0}w_1$ is a multiple of $v_{N_0+1}$, and 
\[\|S'-S\| \leq \max_{k\in[1, N_0]}\left(\frac12|a_k-b_k| + \frac{\pi}{2N_0}\max(|a_k|, |b_k|)\right).\]

Moreover, there are rotated orthonormal vectors $v_k', w_k'$ with
$\spn(\{v_k', w_k'\}) =$ $\spn(\{v_k, w_k\})$ for $k = 1, \dots, N_0+1$, 
$S'v_k' = \frac{a_k+b_k}2v'_{k+1}$ and $S'w_k' = \frac{a_k+b_k}2w'_{k+1}$ for $k = 1, \dots, N_0$, 
and  $v_1' = v_1, v_{N_0+1}' = w_{N_0+1}, w_1' = w_1, w_{N_0+1}' = -v_{N_0+1}$. Also, $S'-S$ is supported on and has range in $\spn\left(\bigcup_{k=1}^{N_0+1}\{v_k, w_k\}\right)$. 
\end{lemma}
\begin{proof}
We can restrict $S$ to $V = \spn(v_1, w_1, \dots, v_{N_0+1}, w_{N_0+1})$ and will leave $S$ alone on $V^\perp$. We will identify $V$ with $\C^{2(N_0+1)}$.

Let the standard basis vectors $e_i$ of $\C^{2(N_0+1)}$ be identified with a basis of $V$ by $v_k \sim e_{2k-1}, w_k \sim e_{2k}$. We can then write $S$ as a matrix of the form
\[S = \begin{pmatrix} 0&&&&\ast\\
C_1&0&&&\ast\\
&C_2&\ddots&&\ast\\
&&\ddots&0&\ast\\
&&&C_{N_0}&\ast\\
&&&&\ast\\
\end{pmatrix},\]
where $C_k = \diag(a_k, b_k) \in M_2(\C)$ and the column of $\ast$'s depicts the action of $S$ on $\spn(v_{N_0+1}, w_{N_0+1})\oplus V^\perp$. The rows correspond to the spaces $\spn(v_{1}, w_{1})$, $\dots$, $\spn(v_{N_0+1}, w_{N_0+1})$, $V^\perp$

Let $0_2$ be the zero vector in $\C^2$ and $0_2^{\oplus \ell}$ denote the $\ell$-fold direct sum of $0_2$.
So, the basis vectors $v_k, w_k$ can be identified with direct sums of vectors in $\C^2$ by padding the standard basis vectors $\bp 1\\0\ep$, $\bp 0\\1\ep$ in $\C^2$ with $2N_0$ zeros appropriately:
\[v_k = 0_{2}^{\oplus (k-1)}\oplus \bp 1\\0\ep \oplus 0_{2}^{\oplus (N_0+1-k)}, \;\; w_k = 0_{2}^{\oplus (k-1)}\oplus \bp 0\\1\ep \oplus 0_{2}^{\oplus (N_0+1-k)}.\]
So, the results of repeatedly multiplying  $v_1$ and $w_1$ by $S$ correspond to the action of the matrix product $C_k \cdots C_1$ on the standard basis vectors in $\C^2$.

Since the product $C_k \cdots C_1$ is diagonal, the main idea of the proof is that if we introduce a small rotation into the terms $C_i$ then we can eventually have the product $C_k \cdots C_1$ be of the form $\bp 0 & \ast \\ \ast & 0\ep$ which would be what is required to interchange the orbits.

Let $R_\theta$ be the rotation matrix $\bp \cos \theta & -\sin\theta \\ \sin\theta & \cos\theta\ep$. Note that  $\|C_k - (a_k+b_k) I_2/2\| \leq |a_k - b_k|/2$. 
Let $S'$ act as the block weighted shift operator on $V$ with weights $C_k' = \frac{a_k+b_k}2 R_{\pi/2N_0}$ and equal to $S$ on $V^\perp$.

Then using $C_{N_0}'\cdots C_1' = \left(\prod_{k=1}^{N_0}\frac{a_k+b_k}2\right) R_{\pi/2} = \left(\prod_{k=1}^{N_0}\frac{a_k+b_k}2\right)\bp 0 & -1 \\ 1 & 0\ep$ we see that $S'$ satisfies the primary conditions of the lemma with 
\begin{align*}
\|S' - S\| &= \max_k \|C_k' - C_k\| \leq \max_k\left(\left\|C_k-\frac{a_k+b_k}2I_2\right\|+\left\|\frac{a_k+b_k}2I_2-C_k'\right\|\right)\\
&\leq\max_k\left(\frac{|a_k-b_k|}2+ \frac{|a_k|+|b_k|}2\left|1-e^{i\pi/2N_0}\right|\right)  \\
&\leq \max_k\left(\frac12|a_k-b_k|
+ \frac{\pi }{2N_0}\max(|a_k|, |b_k|)\right).
\end{align*}

Further, because $R_\theta$ is a real orthogonal matrix, we can define 
\[v_k' = 0_{2}^{\oplus (k-1)}\oplus R_{(k-1)\pi/2N_0}\bp 1\\0\ep \oplus 0_{2}^{\oplus (N_0+1-k)},\] \[w_k' = 0_{2}^{\oplus (k-1)}\oplus R_{(k-1)\pi/2N_0}\bp 0\\1\ep \oplus 0_{2}^{\oplus (N_0+1-k)}\]
to have the required properties from the second part of the statement of the lemma.
\end{proof}
\begin{remark}
Note that in \cite{BergGEL}, there is a phase factor close to $1$ that appears as well to remove the $-1$ term in $R_{\pi/2}$ so that $w_{N_0+1}' = v_{N_0+1}$.
This is unnecessary for our purposes. 

Moreover, because our change of basis: $v_k, w_k \to v_k', w_k'$ is performed by a real orthogonal matrix, this will provide additional structure for the matrices that we later obtain for Ogata's theorem. So, our modification of the construction is preferred. 
\end{remark}

\vspace{0.05in}

We will now modify the gradual exchange lemma put in terms of direct sums of weighted shift operators. Because we will be interested in applying the gradual exchange lemma to direct sums of almost normal weighted shift operators, we will want the perturbation using the gradual exchange lemma to not change the norm of the self-commutator much. See the next section for more about this. The only thing that we need here is to state that if $S = \ws(c_1, \dots, c_{n-1})$ on $\C^n$ then the norm of the self-commutator of $S$ can be expressed as
\[\|\,[S^\ast, S]\,\| = \max\left(|c_1|^2, |c_{n-1}|^2, \max_{i \in [1, n-2]}||c_{i+1}|^2-|c_i|^2|\right).\]

The following is what will be referred to as the gradual exchange lemma.
\begin{lemma}\label{GELws}
Let $S_1 = \ws(a_i)$ with respect to an orthonormal basis $v_i$ of $\C^{n_1}$ and $S_2 = \ws(b_i)$ with respect to an orthonormal basis $w_i$ of $\C^{n_2}$. Assume that $a_i, b_i \geq 0$.
Let $i_0< i_1$ be indices in $[1, \min(n_1, n_2)]\cap \N$ satisfying $\#[i_0, i_1]\cap \N \geq N_0 + 1$.

Then there are $S_1', S_2'$ and orthonormal vectors $v_{i_0}', \dots, v_{i_1}', w_{i_0}', \dots, w_{i_1}'\in \C^{n_1}\oplus \C^{n_2}$
with the following properties:
\begin{enumerate}[label=(\roman*)]
\item \[v_{i_0}' = v_{i_0}\oplus 0, v_{i_1}' = 0\oplus w_{i_1}, w_{i_0}' = 0\oplus w_{i_0}, w_{i_1}' = -v_{i_1}\oplus 0,\] and $\spn(v_i\oplus 0,0\oplus w_i) = \spn(v_i',w_i')$ for $i = i_0, \dots, i_1$.
\item
$S_1' = \ws(a_i')$ with respect to \[v_1\oplus0, \dots, v_{i_0-1}\oplus0, v_{i_0}', \dots, v_{i_1}', 0\oplus w_{i_1+1}, \dots, 0\oplus w_{n_2}\] and $S_2'=\ws(b_i')$ with respect to \[0\oplus w_1, \dots, 0\oplus w_{i_0-1}, w_{i_0}', \dots, w_{i_1}', -v_{i_1+1}\oplus0, \dots, -v_{n_1}\oplus0.\] 
\item  For $i \leq i_0$, $a_i'=a_i$ and $b_i' =  b_i$. For $i \in (i_0, i_1)$,  the $a_i$ and $b_i$ are convex combinations of the $a_i, b_i$.
For $i \geq i_1$, $a_i' = b_i$ and $b_i' = a_i$.
\item The perturbation $S'-S$ is supported on and has range in $\displaystyle\spn\left(\bigcup_{i\in[i_0, i_1]}\{v_i\oplus0, 0\oplus w_i\}\right)$.
\item If  $S = S_1\oplus S_2$ and $S'=S_1'\oplus S_2'$ then
\[\|S'-S\| \leq \max_{i\in[i_0, i_1)}\left(|a_i - b_i| + \frac{\pi}{2N_0}\max(|a_i|, |b_i|)\right)\]
and
\[
\|\,[S'^\ast, S']\,\|\leq \|\,[S^\ast, S]\,\|+\frac1{N_0}\max_{i\in(i_0, i_1]}||b_i|^2-|a_i|^2|.
\]
\end{enumerate}
\end{lemma}
\begin{proof}
We apply Lemma \ref{pre-gel} to the at least $N_0+1$ vectors $v_i\oplus 0$ and $0\oplus w_i$ from the statement of this lemma for $i = i_0, \dots, i_1$.
This provides what we will call $\tilde S$ expressed as the direct sum of $\tilde S_1$ and $\tilde S_2$ as follows. 

This provides vectors which we call $v_{i_0}', \dots, v_{i_1}'$ with the properties that $\tilde S$ acts as \[\tilde Sv_{i_0}' = \frac{a_{i_0}+b_{i_0}}{2}v_{i_0+1}', \dots, \tilde Sv_{i_1-2}' = \frac{a_{i_1-2}+b_{i_1-2}}{2}v_{i_1-1}',  \tilde Sv_{i_1-1}' = \frac{a_{i_1-1}+b_{i_1-1}}{2}v_{i_1}',\]
\[\tilde S (v_1\oplus 0) = S (v_1\oplus 0) = a_1(v_2\oplus 0), \dots, \tilde S (v_{i_0-1}\oplus 0) = S (v_{i_0-1}\oplus 0) = a_{i_0-1}(v_{i_0}\oplus 0),\]
and
\[\tilde S (0\oplus w_{i_1}) = S (0\oplus w_{i_1}) = b_{i_1}(0\oplus w_{i_1+1}), \dots, \tilde S( 0\oplus w_{n_2-1}) = b_{n_2-1}(0\oplus w_{n_2}), \tilde S(0\oplus w_{n_2})=0. \]
Because $v_{i_0}' = v_{i_0}\oplus 0$ and $v_{i_1}' = 0\oplus w_{i_1}$, we have
\[\tilde S_1 = \ws\left(a_1, \dots, a_{i_0-1}, \frac{a_{i_0}+b_{i_0}}{2}, \dots, \frac{a_{i_1-1}+b_{i_1-1}}{2}, b_{i_1}, b_{i_1+1},\dots, b_{n_2-1}\right),\]
with respect to the orthonormal 
\[v_1\oplus0, \dots, v_{i_0-1}\oplus0, v_{i_0}', \dots, v_{i_1-1}',  v_{i_1}', 0\oplus w_{i_1+1}, \dots, 0\oplus w_{n_2-1}, 0\oplus w_{n_2}.\]

The lemma also provides vectors which we call $w_{i_0}', \dots, w_{i_1}'$ with the properties that $\tilde S$ acts as \[\tilde Sw_{i_0}' = \frac{a_{i_0}+b_{i_0}}{2}w_{i_0+1}', \dots, \tilde Sw_{i_1-2}' = \frac{a_{i_1-2}+b_{i_1-2}}{2}w_{i_1-1}',  \tilde Sw_{i_1-1}' = \frac{a_{i_1-1}+b_{i_1-1}}{2}w_{i_1}',\]
\[\tilde S (0\oplus w_1) = S (0\oplus w_1) = b_1(0\oplus w_2), \dots, \tilde S (0\oplus w_{i_0-1}) = S (0\oplus w_{i_0-1}) = b_{i_0-1}(0\oplus w_{i_0}),\]
and
\[\tilde S ( v_{i_1}\oplus0) = S ( v_{i_1}\oplus0) = a_{i_1}( v_{i_1+1}\oplus0), \dots, \tilde S( v_{n_1-1}\oplus0) = a_{n_1-1}(v_{n_1}\oplus0), \tilde S( v_{n_1}\oplus0)=0. \]
Because $w_{i_0}' = 0\oplus w_{i_0}$ and $w_{i_1}' = -v_{i_1}\oplus0$, we have \[\tilde S_2 = \ws\left(b_1, \dots, b_{i_0-1}, \frac{a_{i_0}+b_{i_0}}{2}, \dots, \frac{a_{i_1-1}+b_{i_1-1}}{2}, -a_{i_1}, a_{i_1+1},\dots, a_{n_1-1}\right)\]
with respect to the orthonormal  
\[0\oplus w_1, \dots, 0\oplus w_{i_0-1}, w_{i_0}', \dots, w_{i_1-1}', w_{i_1}', v_{i_1+1}\oplus0, \dots, v_{n_1-1}\oplus0, v_{n_1}\oplus0.\] 

By changing the basis of this second mixed list of vectors through introducing and propagating a negative sign to the vectors after $w_{i_1}'$, we see that we can express $\tilde S_2$ unchanged as a weighted shift matrix with all non-negative weights:
\[\tilde S_2 = \ws\left(b_1, \dots, b_{i_0-1}, \frac{a_{i_0}+b_{i_0}}{2}, \dots, \frac{a_{i_1-1}+b_{i_1-1}}{2}, a_{i_1}, a_{i_1+1},\dots, a_{n_1-1}\right)\]
with respect to 
\[0\oplus w_1, \dots, 0\oplus w_{i_0-1}, w_{i_0}', \dots, w_{i_1-1}', w_{i_1}', -v_{i_1+1}\oplus0, \dots, -v_{n_1-1}\oplus0, -v_{n_1}\oplus0.\] 

We will now alter the weights of $\tilde S_1$ and $\tilde S_2$ so that the weights change more gradually while interchanging orbits. This will provide the operators $S_1'$ and $S_2'$.
Note that $i_1 - i_0 \geq N_0$. Define 
\[t_i = \left\{\begin{array}{ll}
0 & i < i_0 \\
\frac{i-i_0}{i_1-i_0} & i_0 \leq i \leq i_1 \\
1 & i > i_1
.\end{array}\right..\] So the $t_i$ satisfy $0 \leq t_i \leq 1$, $t_{i_0} = 0$, $t_{i_1}= 1$, and $|t_{i+1}-t_i| \leq 1/N_0$.

Define $a_i', b_i'$ to be non-negative satisfying
\[|a_i'|^2 = (1-t_i)|a_i|^2 + t_i|b_i|^2 = |a_i|^2 + t_i(|b_i|^2-|a_i|^2),\]
\[|b_i'|^2 = t_i|a_i|^2 + (1-t_i)|b_i|^2 = |b_i|^2 + t_i(|a_i|^2-|b_i|^2).\]
Now, change the weights of $\tilde S_1$ and $\tilde S_2$ to be $a_i'$ and $b_i'$ to obtain $S_1'$ and $S_2'$, respectively.

We now verify the statements of the lemma. (i) and (ii) are clear from our discussion of $\tilde S_1$  and $\tilde S_2$ in the beginning of the proof. 

Because $0\leq t_i \leq 1$, we have that $|a_i'|^2$ and $|b_i'|^2$ are each convex combinations of $|a_i'|^2$ and $|b_i'|^2$. Because $a_i, a_i', b_i, b_i'$ are all non-negative, we have that $a_i'$ and $b_i'$ belong to the interval $[\min(a_i, b_i), \max(a_i, b_i)]$ for $i$ in $[i_0, i_1]$. This and the above comments about $t_i$ show (iii). 
 
Because $S' - S = (S'-\tilde S) + (\tilde S - S)$, we see that (iv) holds as well by construction.
 
Because the $a_i', b_i'$ are convex combinations of the $a_i, b_i$, they are then within a distance of $|b_i-a_i|/2$ from $(a_i+b_i)/2$.
So, \[\|S'-\tilde S\| \leq \frac12\max_{i \in [i_0, i_1)}|b_i-a_i|.\] This then provides the estimate for $\|S'-S\|$. 

We now obtain the other estimate of (v). For a sequence $c_i$, let $\Delta$ denote the forward difference operator: $\Delta c_i = c_{i+1} - c_i$.
Notice that
\begin{align*}
\Delta|a'|^2_i &= \Delta|a|^2_i + t_{i+1}(|b_{i+1}|^2-|a_{i+1}|^2) - t_i(|b_i|^2-|a_i|^2)
\\
&= \Delta|a|^2_i +t_{i}(\Delta|b|^2_i-\Delta|a|^2_i) + (t_{i+1}-t_i)(|b_{i+1}|^2-|a_{i+1}|^2).
\end{align*}

So, 
\begin{align*}
\|\,[S_1'^\ast, S_1']\,\| &= \max\left(|a_1'|^2, |a_{n
_2}'|^2, \max_i\Delta|a'|^2_i\right)
\\&\leq \max\left(\|\,[S^\ast, S]\,\|, \|\,[S^\ast, S]\,\|  + \max_{i}\left((t_{i+1}-t_i)||b_{i+1}|^2-|a_{i+1}|^2|\right)\right)\\
&\leq \|\,[S^\ast, S]\,\|+\frac1{N_0}\max_{i\in(i_0, i_1]}||b_i|^2-|a_i|^2|,
\end{align*}
because $t_{i+1} = t_i$ unless $i_0\leq i< i_1$.
Interchanging the roles of $a_i$ and $b_i$ provides
\begin{align*}
\|\,[S_2'^\ast, S_2']\,\| \leq \|\,[S^\ast, S]\,\|+\frac1{N_0}\max_{i\in(i_0, i_1]}||b_i|^2-|a_i|^2|.
\end{align*}
This then provides the second inequality in the statement of the lemma.
\end{proof}
\begin{remark}
Note that we need not propagate the negative signs to the vectors $v_i$ for $i > i_1$ in our construction. What this amounts to is having a single negative weight $-b_{i_1}'$ for $S_2'$.

Note that when applying the gradual exchange lemma repeatedly on orthogonal subspaces, one can  apply the lemma as stated. This is done in detail for a simple case in Example \ref{GEL_Arrows}. Alternatively, one can apply the construction from the lemma without propagating negative signs as mentioned above, given that no weights that the lemma is applied to are ever negative.

With this modification, one may then propagate negative signs once after all the applications of the gradual exchange method to avoid relabeling or keeping track of which vectors inherit negative signs due to repeated applications along a single orbit. This difficulty comes up in the construction in Remark \ref{noNegatives} and is avoided due to this alternative in Example \ref{gep-Example}.
\end{remark}
\begin{remark}\label{weightinterchange}
The result of applying the gradual exchange lemma can be seen as perturbing
\begin{align*}
S: v_1\oplus0\overset{a_1}{\rightarrow} \cdots\overset{a_{i_0-1}}{\rightarrow} v_{i_0}\oplus0\overset{a_{i_0}}{\rightarrow} \cdots\overset{a_{i_1-1}}{\rightarrow} &v_{i_1}\oplus0\overset{a_{i_1}}{\rightarrow} v_{i_1+1}\oplus0\overset{a_{i_1+1}}{\rightarrow} \cdots\overset{a_{n_1-1}}{\rightarrow} v_{n_1}\oplus0\rightarrow0\\
S: 0\oplus w_1\overset{b_1}{\rightarrow} \cdots\overset{b_{i_0-1}}{\rightarrow} 0\oplus w_{i_0}\overset{b_{i_0}}{\rightarrow} \cdots\overset{b_{i_1-1}}{\rightarrow} &0\oplus w_{i_1}\overset{b_{i_1}}{\rightarrow} 0\oplus w_{i_1+1}\overset{b_{i_1+1}}{\rightarrow} \cdots\overset{b_{n_1-1}}{\rightarrow} 0\oplus w_{n_1}\rightarrow0
\end{align*}
to
\begin{align*}
S': v_1\oplus0 \overset{a_{1}}{\rightarrow}\cdots\overset{a_{i_0-1}}{\rightarrow} &v_{i_0}'\overset{a_{i_0}'}{\rightarrow} \cdots\overset{a_{i_1-1}'}{\rightarrow}  v_{i_1}'\overset{a_{i_1}'}{\rightarrow} 0\oplus w_{i_1+1}\overset{b_{i_1+1}}{\rightarrow} \cdots \overset{b_{n_2-1}}{\rightarrow}0\oplus w_{n_2}\rightarrow0\\
S': 0\oplus w_1 \overset{b_{1}}{\rightarrow}\cdots\overset{b_{i_0-1}}{\rightarrow}  &w_{i_0}'\overset{b_{i_0}'}{\rightarrow} \cdots\overset{b_{i_1-1}'}{\rightarrow} w_{i_1}'\overset{b_{i_1}'}{\rightarrow} -v_{i_1+1}\oplus0\overset{a_{i_1+1}}{\rightarrow} \cdots \overset{a_{n_1-1}}{\rightarrow} -v_{n_1}\oplus0\rightarrow0
\end{align*}
with the properties specified in the statement of the lemma. 

This is illustrated in the weighted shift diagram of Figure \ref{GEL_Illustration}. 
\begin{figure}[htp]     \centering
    \includegraphics[width=15cm]{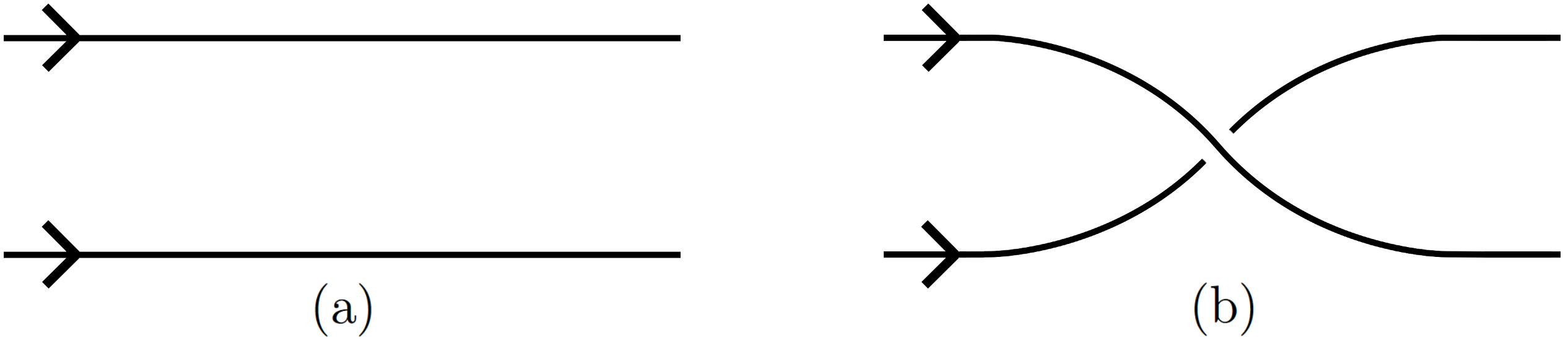}
    \caption{\label{GEL_Illustration}\dark
    Illustration of applying the gradual exchange lemma to two weighted shift diagrams.}
\end{figure}
Figure \ref{GEL_Illustration}(a) depicts weighted shift diagrams for $S$ and Figure \ref{GEL_Illustration}(b) depicts $S'$. The main focus of the diagrams is to illustrate that the  orbits of $S'$ begin with some vectors initially in the orbit of $S_1$ and $S_2$ then eventually are in the orbit of $S_2$ and $S_1$, respectively.
\end{remark}

\begin{remark}
This result applies to two weighted shifts whenever we have intervals of indices of length $i_1-i_0$ for each of the weighted shifts on which we apply the gradual exchange. The first index of these intervals need not be the same. We see this by simply relabeling the indices so that the first ``$v$'' vector is $v_{i_a}$, the first ``$w$'' vector is $w_{i_b}$, and the interval over which we apply the gradual exchange lemma begins with the same index $i_0$. 

Then the modification to $S$ would be as follows: 
\begin{align*}
S: &\;v_{i_a}\oplus0\overset{a_{i_a}}{\rightarrow} \cdots\overset{a_{i_0-1}}{\rightarrow} v_{i_0}\oplus0\overset{a_{i_0}}{\rightarrow} \cdots\\
S: &\;0\oplus w_{i_b}\overset{b_{i_b}}{\rightarrow} \cdots\overset{b_{i_0-1}}{\rightarrow} 0\oplus w_{i_0}\overset{b_{i_0}}{\rightarrow} \cdots.
\end{align*}
With this modification, $S$ and $S'$ are analogous to that of Remark \ref{weightinterchange}.
\end{remark}

\begin{example}\label{GEL_Arrows}
Because of
(iv) and (v) in the gradual exchange lemma, we can apply this lemma repeatedly to some direct sum of weighted shift operators without an increase in the norm of the perturbation or the self-commutator as long as no vectors are repeated in the different applications of the gradual exchange lemma.
As an example, consider $S_1 = \ws(a_i)$, $S_2 = \ws(b_i)$, $S_3 = \ws(c_i)$ with respect to $e_1, \dots, e_{2k}$ for $i = 1, \dots, 2k:=2(N_0+1)$. 
The action of $S = S_1\oplus S_2\oplus S_3$ is expressed in Figure \ref{GEL_Arrows1}(a).

\begin{figure}[htp]     \centering
    \includegraphics[width=15cm]{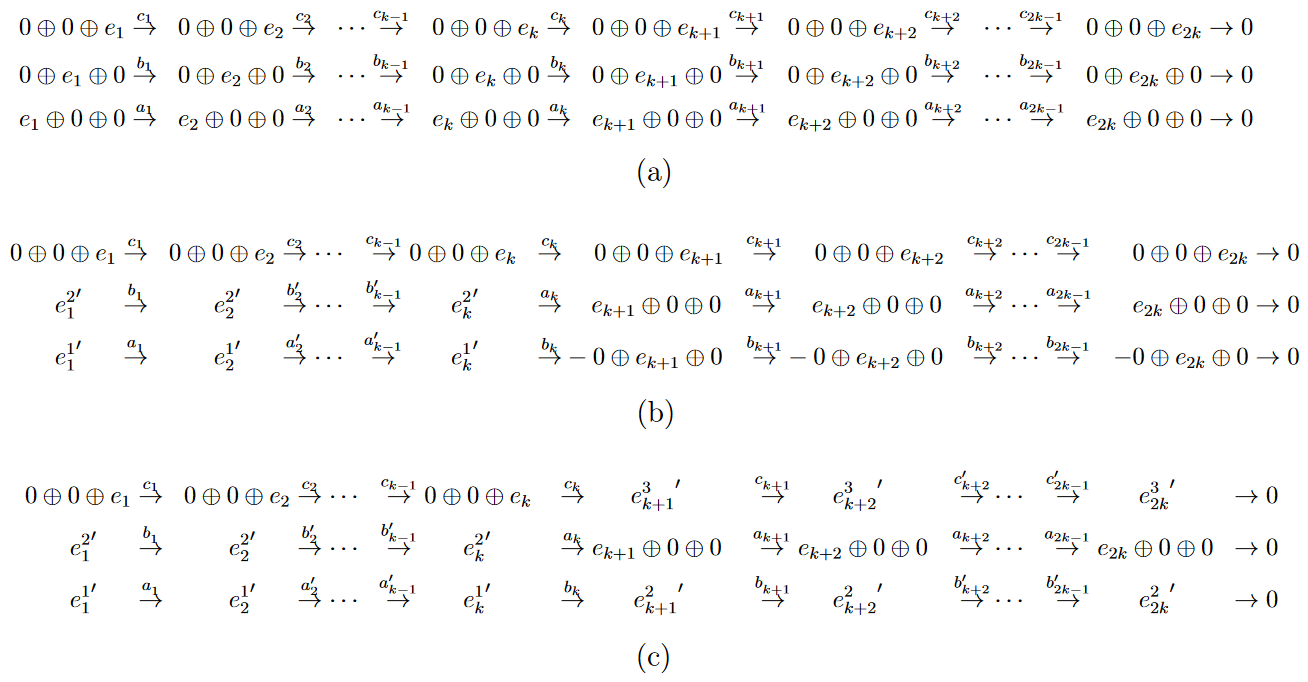}
    \caption{\label{GEL_Arrows1}\dark
    (a) is an illustration for Example \ref{GEL_Arrows} of the vectors and weights of $S_1\oplus S_2\oplus S_3$ on three invariant subspaces on which $S$ acts as a weighted shift with weights written above the arrows. (b) is an illustration of how the vectors and weights changed when applying the gradual exchange lemma to $S_2, S_1$ over the first $k$ vectors. (c) is an illustration of how the vectors and weights changed when applying the gradual exchange lemma to the result of the previous application over the orbits of $S_3, S_2$ over the latter $k$ vectors.}
\end{figure}

We now apply the gradual exchange lemma to $S_2, S_1$ over the vectors corresponding to $i = 1, \dots, k$ and to $S_3, S_2$ over the vectors corresponding to $i = k+1, \dots, 2k$. Note that the order of $S_2, S_1$ vs. $S_1, S_2$ is important inasmuch as it indicates which orbit's vectors inherit negative signs after the interchange. The second orbit listed inherits the negative signs.

We first apply the gradual exchange lemma to the orbits of $S_2$ and $S_1$ over the first interval to obtain Figure \ref{GEL_Arrows2}(b). This provides vectors
\[{e_1^2}'=0\oplus e_1\oplus0, {e_2^2}', \dots, {e_{k-1}^2}', {e_k^2}'= e_k\oplus0\oplus0\]
and
\[{e_1^1}'=e_1\oplus0\oplus0, {e_2^1}', \dots, {e_{k-1}^1}', {e_k^1}'=-0\oplus e_k\oplus0\] and weights \[b_1'=b_1, b_2' \dots, b_{k-1}', b_k'=a_k, \;\;\; a_1'=a_1, a_2' \dots, a_{k-1}', a_k'=b_k.\]

Then we apply the gradual exchange lemma to the orbits of $S_3, S_2$ over the second interval to obtain Figure \ref{GEL_Arrows2}(c). This provides vectors
\[{e_{k+1}^3}'=0\oplus0\oplus e_{k+1}, {e_{k+2}^3}', \dots, {e_{2k-1}^3}', {e_{2k}^3}'=-0\oplus e_{2k}\oplus0\]
and
\[{e_{k+1}^2}'=-0\oplus e_{k+1}\oplus0, {e_{k+2}^2}', \dots, {e_{2k-1}^2}', {e_{2k}^2}'=-0\oplus0\oplus e_{2k}\] 
and weights \[c_{k+1}'=c_{k+1}, c_{k+2}' \dots, c_{2k-1}', c_{2k}'=b_{2k}=0, \;\;\; b_{k+1}'=b_{k+1}, b_{k+2}' \dots, b_{2k-1}', b_{2k}'=b_{2k}=0.\]

We refer to the operator in  Figure \ref{GEL_Arrows2}(c)  gotten by applying the gradual exchange lemma twice as $S'$. Notice that the perturbations in each application of Lemma \ref{GELws} are supported on and have range in orthogonal subspaces in accordance with Lemma \ref{GELws}(iv). Recall that $\#[1,k] = k = N_0+1$. So, the estimate for $\|S'-S\|$ is gotten as the maximum of the estimates from the two applications:
\begin{align*}\|S'-S\| &\leq \max\left(\max_{i\in[1, k)}\left(|a_i - b_i| + \frac{\pi}{2N_0}\max(|a_i|, |b_i|)\right),\right.\\
&\;\;\;\;\;\;\;\;\;\;\;\;\;\;\;\left.\max_{i\in[k+1, 2k)}\left(|b_i - c_i| + \frac{\pi}{2N_0}\max(|b_i|, |c_i|)\right)\right)\\
&\leq \max_{i}\left(\max(|a_i - b_i|, |b_i-c_i|) + \frac{\pi}{2N_0}\max(|a_i|, |b_i|, |c_i|)\right).
\end{align*}

The estimate for the self-commutator of $S'$ is not based on analyzing a perturbation of $S$ but instead the weights of $S'$. We then see that because each application of the gradual exchange lemma leaves the first and last weight in each orbit unchanged, the difference of the squares of weights in an orbit are those of one of the isolated applications of the gradual exchange lemma. So, we see that the norm of the self-commutator due to repeated applications is the maximum of the separate estimates:
\begin{align*}
\|\,[S'^\ast, S']\,\|&\leq \max\left(\|\,[S^\ast, S]\,\|+\frac1{N_0}\max_{i\in(1, k]}||b_i|^2-|a_i|^2|,\right.\\
&\;\;\;\;\;\;\;\;\;\;\;\;\;\;\;\left.\|\,[S^\ast, S]\,\|+\frac1{N_0} \max_{i\in(k+1, 2k]}||c_i|^2-|b_i|^2|\right)\\
&\leq \|\,[S^\ast, S]\,\|+\frac1{N_0}\max_{i}\max\left(||b_i|^2-|a_i|^2|, ||c_i|^2-|b_i|^2|\right).
\end{align*}
\begin{figure}[htp]     \centering
    \includegraphics[width=15cm]{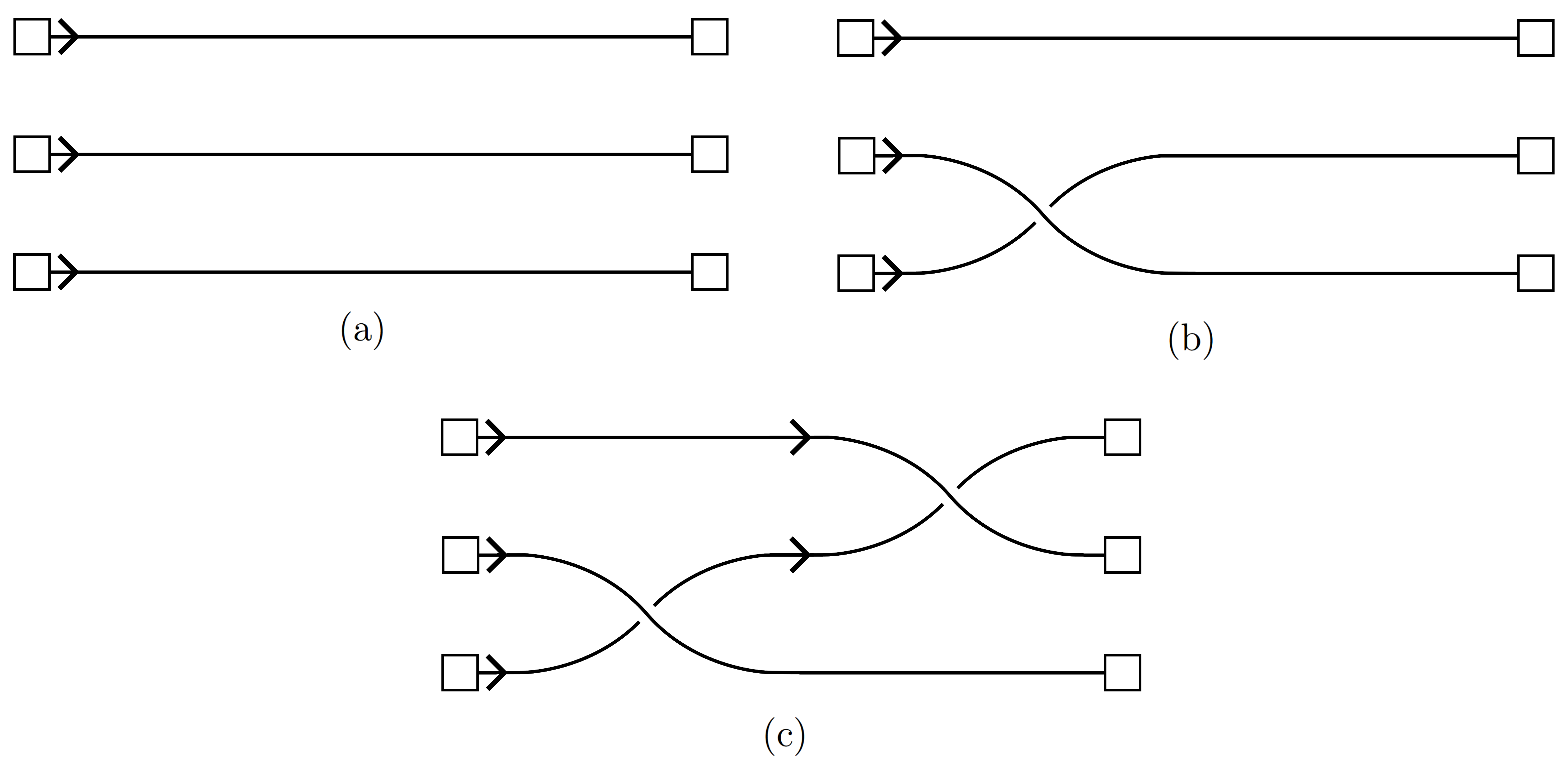}
    \caption{\label{GEL_Arrows2}\dark illustrations of the weighted shifts from Figure \ref{GEL_Arrows1} using weighted shift diagrams.
    }
\end{figure}
\end{example}

\section{Almost Normal Weighed Shifts}\label{ANWS}

Recall that the optimal upper bound by Kachkovskiy and Safarov in \cite{KS} for how nearby an almost normal matrix $S$ is to a normal matrix $N$ is:
\[\|N-S\| \leq C\,\|\,[S^\ast ,S]\,\|^{\alpha}\]
with $\alpha = 1/2$. It is not possible for such an estimate to hold with a different value of $\alpha$ without restrictions on the norm of $S$ for scaling reasons. A scaling-invariant form of this inequality obtained for $\|S\|=1$ would give
\[\|N-S\| \leq C_\alpha\|S\|^{1-2\alpha}\,\|\,[S^\ast ,S]\,\|^{\alpha}.\]
The main result of this section is Theorem \ref{BergResult} which contains an estimate of this type for $\alpha = 1/3$ for a weighted shift matrix $S$ with also the special property that $N$ can be chosen to be real when $S$ is real. 

Loring and S{\o}rensen in \cite{LS} showed the following structured Lin's theorem: if two almost commuting real self-adjoint matrices are real then they are nearby two actually commuting real self-adjoint matrices. They also showed that a real almost normal matrix is nearby a real normal matrix as well. However, these proofs are not constructive and do not provide any estimates.

In this section we present a refined version of Berg's constructive result in \cite{BergNWS} of Lin's theorem for an almost normal weighted shift matrix. Berg's construction when framed in terms of obtaining a result of this form would provide $\alpha =1/4$ due to the effect of small weights in some of the inequalities as described later. By refining the construction and estimates we obtain this result for $\alpha = 1/3$. Our modification of the construction and calculations also provide much smaller numerical constants with a structured result.

See \cite{BergNWS} for the details of Berg's original argument and also \cite{LoringPseudoActions} for a discussion and illustration of how Berg's formulation of the gradual exchange concept is applied in this construction. 
Before we can say much more about the rest of this section, we make some definitions.

\begin{defn}
Given an orthonormal basis $v_1, \dots, v_n \in \C^n$, we define the bilateral weighted shift operators $T = \bws(c_1, \dots, c_n)$ to be the linear operator on $\C^n$ which satisfies $Tv_i = c_i v_{i+1}$. We use the convention that the vectors $v_k$ and weights $c_k$ are indexed cyclically.
We can express the action of $T$ as:
\[T:v_1 \overset{c_1}{\rightarrow}v_2\overset{c_2}{\rightarrow}\cdots\overset{c_{n-2}}{\rightarrow}v_{n-1}\overset{c_{n-1}}{\rightarrow}v_n\overset{c_n}{\rightarrow}v_{n+1}=v_1.\]

If $T =\bws(c_1,\dots, c_{n-1}, 0)= \ws(c_1,\dots, c_{n-1})$ then we say that $T$ is a (unilateral) weighted shift. In the previous section we expressed this as
\[T:v_1 \overset{c_1}{\rightarrow}\cdots\overset{c_{n-1}}{\rightarrow}v_n\rightarrow0,\] but expressed as a bilateral weighted shift this is:
\[T:v_1 \overset{c_1}{\rightarrow}\cdots\overset{c_{n-1}}{\rightarrow}v_n\overset{0}{\rightarrow}v_1.\]
\end{defn}
\begin{example}\label{ws-basisChange}
Let $T = \bws(c_1, \dots, c_n)$.
In the basis $\beta =  (v_1, \dots, v_n)$, $T$ is expressed as the matrix
\[[T]_\beta = 
\bp
0&&&&c_n\\
c_1&0&&&\\
&c_2&0&&\\
&&\ddots&\ddots&\\
&&&c_{n-1}&0
\ep.  \]
If we view $v_1, \dots, v_n$ as the standard basis, then we can think of $T$ as being this matrix. Otherwise, we can think of $T$ being unitarily equivalent to this matrix.
A simple example of this is that $T$ is unitarily equivalent to the matrix obtained by cyclically permuting the weights of $T$.

Consider the following change of basis obtained by multiplying the vectors $v_n$ by the phases $\omega_n \in \C$ with $|\omega_n|=1$.
Using the basis $\beta_\omega = (\omega_1v_1, \omega_2v_2\dots, \omega_nv_n)$, $T$ is seen to be unitarily equivalent to
\[[T]_{\beta_\omega} = 
\bp
0&&&&\frac{\omega_n}{\omega_1}c_n\\
\frac{\omega_1}{\omega_2}c_1&0&&&\\
&\frac{\omega_2}{\omega_3}c_2&0&&\\
&&\ddots&\ddots&\\
&&&\frac{\omega_{n-1}}{\omega_n}c_{n-1}&0
\ep.  \]

In particular, if we choose $\omega_1 = 1$ and define $\omega_k$ recursively by
\[\omega_{k+1} = \left\{ \begin{array}{ll} 
\frac{c_k}{|c_k|}\omega_k, & c_k \neq 0\\
1, & c_k =0\\
\end{array} \right.\]
then we obtain
\[[T]_{\beta_\omega} = 
\bp
0&&&&\omega c_n\\
|c_1|&0&&&\\
&|c_2|&0&&\\
&&\ddots&\ddots&\\
&&&|c_{n-1}|&0
\ep, \]
where if $c$ is the product of the $c_k\neq 0$ for $k < n$ then $\omega = c/|c|$.

We make a few observations. If one of the weights is zero, as in the case of a unilateral weighted shift, then all the weights can be made non-negative in this manner. If all the $c_k$ are real then $\omega_k=\pm1$, so upon conjugation by a diagonal matrix with diagonal entries $\pm1$ the weights  can be made all positive except perhaps the last. We can make all the weights positive exactly when the product of all the $c_k$ is positive.
\end{example}

\begin{example}
We use the same notation as in the previous example.
The self-commutator $[T^\ast, T] = T^\ast T-TT^\ast$ has a matrix representation of
\[[T^\ast T-TT^\ast]_\beta  = 
\bp
|c_1|^2-|c_n|^2&&&&\\
&|c_2|^2-|c_1|^2&&&\\
&&|c_3|^2-|c_2|^2&&\\
&&&\ddots&\\
&&&&|c_n|^2-|c_{n-1}|^2
\ep  \]
so that $\|[T^\ast, T]\| = \max_k||c_{k+1}|^2-|c_k|^2|$. 
So, $T$ is normal if all the $c_k$ have the same absolute value and $T$ is almost normal if the $|c_k|^2$ change slowly.

In particular, if $T$ is a unilateral weighted shift operator then
\[\|[T^\ast, T]\| = \max\left(|c_1|^2, |c_{n-1}|^2, \max_{1\leq k \leq n-2}||c_{k+1}|^2-|c_k|^2|\right)\] and $T$ is normal only if $T = 0$ identically.

A standard example of an almost normal unilateral weighed shift matrix is used in \cite{Davidson} where the weights of $S$ start near zero, slowly increase to one, then decrease back to zero. We see that such a matrix is nearby a normal matrix by Lin's theorem. However, any nearby normal matrix cannot be a bilateral weighted shift matrix in the same basis since all the weights would need to have the same absolute value. It also cannot be a bilateral weighted shift matrix in any other basis since then all the singular values of the normal matrix should be the same, which is not a possible property of a small perturbation of $S$. We will show in this section that an almost normal weighted shift is nearby a direct sum of normal weighted shift matrices in some bases.

\end{example}

We now complete our introduction to this section by discussing the results that we obtain.
Lemma \ref{gel-forNearbyNormal} and Lemma \ref{HelpingBerg} can be seen as an adaption of Berg's original argument. There are two main differences. First, our implementation of the ``gradual exchange'' idea in Lemma \ref{gel-forNearbyNormal} has a simpler definition, has a tighter estimate, and does not involve complex numbers at the expense of having the negative sign in $\eta_{k_0} = -v_0$. 

The second difference is that Berg expressed his estimates in terms of $\max_k||c_{k+1}|-|c_k||$. The motivation for this is based in the characterization of a normal operator as one that satisfies $\|Nv\| = \|N^\ast v\|$ for all vectors $v$. We showed above that the norm of the self-commutator $[S^\ast , S]$ equals $\max_k||c_{k+1}|^2-|c_k|^2|$. Although Berg's construction produces an estimate of the form
\[\|N-S\| \leq Const.\sqrt{\max_k||c_{k+1}|-|c_k||}\]
for $\|S\|=1$ and $\max_k||c_{k+1}|-|c_k||$ small enough, this result produces an estimate in terms of the self-commutator having exponent $\alpha = 1/4$ due to
\[\max_k||c_{k+1}|-|c_k|| \leq \sqrt{\max_k||c_{k+1}|^2-|c_k|^2|}.\]
Because \[|c_{k+1}|^2-|c_k|^2 = (|c_{k+1}|-|c_k|)(|c_{k+1}|+|c_k|),\]
the inequality above is asymptotically sharp when the difference $||c_{k+1}|-|c_k||$ has a similar size as the sum $|c_{k+1}|+|c_k|$. This can happen when, for instance, $|c_{k+1}|$ is much larger than $|c_{k}|$.

Then in Theorem \ref{BergResult} we present a version of a condition of Theorem 2 of \cite{BergNWS} that does not require the operator to have norm $1$ or have any requirement on the size of the self-commutator.
This includes a result with exponent $\alpha = 1/3$ 
and also an estimate with $\alpha = 1/2$ with a scaling-invariant factor that is large when there are weights of the matrix that are much smaller than the norm. 

We now proceed to the results of this section. 
The proof of \cite{BergNWS} was formulated in terms of a recursive algorithm. We isolate this part as the following lemma so that the entire proof in Lemma \ref{HelpingBerg} is expressed as a single step. The modification of Berg's construction here can be seen as applying the gradual exchange lemma to two portions of $S$. 

\begin{lemma}\label{gel-forNearbyNormal}
Suppose that $S$ is a linear map on $\C^n$ such that there are orthonormal vectors $v_i, w_j$ in $\C^n$ with $Sv_i=bv_{i+1}, Sw_j = bw_{j+1}$ for $i = 0, \dots, k_0$ and $j = 0, \dots, k_0+1$.

Let $\alpha_k = \cos\left(\frac{\pi k}{2k_0}\right)$ and $\beta_k = \sin\left(\frac{\pi k}{2k_0}\right)$ and define \[\xi_k = \alpha_kv_k+\beta_kw_k, \;\;\; \eta_k = -\beta_kv_k+\alpha_kw_k\] for $k = 0, \dots, k_0$.
Note that $(\xi_k, \eta_k)$ is gotten by rotating $(v_k, w_k)$ in a two dimensional subspace by $\pi k/2k_0$ so that $(\xi_0, \eta_0) = (v_0, w_0)$ and $(\xi_{k_0}, \eta_{k_0}) = (w_{k_0}, -v_{k_0})$.
Let $S'$ be the linear operator that satisfies
\[S'\xi_k= a\xi_{k+1}, S'\eta_k = b\eta_{k+1}, 0 \leq k \leq k_0-1,\]
\[S'\xi_{k_0}= S'w_{k_0} = aw_{k_0+1},  S'\eta_{k_0}= -S'v_{k_0} = -Sv_{k_0},\]
and equals $S$ on the orthogonal complement of the span of the $v_k$ and $w_k$, $k = 1, \dots, k_0$. 

Then 
\[\|S'-S\| \leq |b-a| + |b|\frac{\pi}{2k_0}.\]
\end{lemma}
\begin{proof}
Let $\mathcal U_k = \spn(v_k, w_k)$ for $k = 0, \dots, k_0$. Notice that both $(v_k, w_k)$ and $(\xi_k, \eta_k)$ form orthonormal bases for $\mathcal U_k$.  We first claim that \[\|S' - S\| = \max_k\|(S'-S)P_{\mathcal U_k}\|.\] 
Notice that $S'-S$ is only non-zero on the span on the $\mathcal U_k$. Also, $S'-S$ maps $\mathcal U_{k}$ into $\mathcal U_{k+1}$ for $0 \leq k \leq k_0-1$ and $S'-S$ maps $\mathcal U_{k_0} $ into the span of $w_{k_0+1}$ as seen below. So, the restrictions $(S'-S)P_{\mathcal U_k}$ have orthogonal ranges which is enough to prove this claim.

We now continue with calculating $\|(S'-S)P_{\mathcal U_k}\|$ for $k = k_0$:
\[(S'-S)\xi_{k_0}= (a-b)w_{k_0+1}\]
\[(S'-S)\eta_{k_0}=0.\]
So, $\|(S'-S)P_{\mathcal U_{k_0}}\| \leq |b-a|$.
This also shows that $S'-S$ maps $\mathcal U_{k_0} $ into the span of $w_{k_0+1}$ as referenced above.

Recall the real orthogonal rotation matrix
\[R_\theta = \bp \cos(\theta) & -\sin(\theta) \\ \sin(\theta) & \cos(\theta)\ep\] which satisfies $R_{\theta}R_{\varphi} = R_{\theta+\varphi}$ and has eigenvalues $e^{\pm\pi i\theta}$.
Notice that if $(e_1, e_2)$ is the standard basis of $\C^2$ then the coordinates of  $\xi_k$ and  $\eta_k$ with respect to $(v_k, w_k)$ are exactly those of $R_{\pi k/2k_0}e_1$ and $R_{\pi k/2k_0}e_2$, respectively. 

We now consider the case when $0 \leq k \leq k_0-1$. We will represent $S$ and $S'$ on $\mathcal U_k$ with the matrices 
$[SP_{\mathcal U_{k_0}}], [S'P_{\mathcal U_{k_0}}]$ 
with respect to the bases $(\xi_k, \eta_k)$ of $\mathcal U_k$ and $(v_{k+1}, w_{k+1})$ of $\mathcal U_{k+1}$. 
We obtain
\[[SP_{\mathcal U_{k}}] = \bp b\alpha_k& -b\beta_k\\b\beta_k &b\alpha_k \ep = bR_{\pi k/2k_0}\]
and
\[[S'P_{\mathcal U_{k}}] = \bp a\alpha_{k+1}& -b\beta_{k+1}\\ a\beta_{k+1} &b\alpha_{k+1} \ep = \bp (a-b)\alpha_{k+1}& 0\\ (a-b)\beta_{k+1} &0 \ep + bR_{\pi(k+1)/2k_0}.\]
So,
\begin{align*}
\|(S'-S)P_{\mathcal U_k}\| &\leq |b|\|R_{\pi k/2k_0}-R_{\pi(k+1)/2k_0}\|+ |b-a| = |b|\|I-R_{\pi/2k_0}\|+ |b-a| \\
&= |b||1-e^{\pi i/2k_0}| + |b-a| \leq \frac{\pi}{2k_0}|b| + |b-a|.
\end{align*}

\end{proof}

We now move to our modification of the main construction from \cite{BergNWS}. Note that an explicit construction is not provided there for the first step of the following lemma so we provide it for completeness. We also express our estimate in terms of $\|S\|$ because it will allow us to optimize the constant $C_\alpha$ later.
\begin{lemma} \label{HelpingBerg}
Suppose that $S \in M_n(\C)$ is a bilateral weighted shift matrix with weights $c_1, \dots, c_{n}$. Let $M \geq 4$ be an even integer. If
\[\|\,[S^\ast, S]\,\| < \frac{1}{M^{3}}\]
then there is a normal matrix $N$ such that  
\[\|N-S\| < \left(\|S\|\frac{\pi M}{M-2}+2\right)\frac{1}{M}.\]

Additionally, $N$ is a direct sum of weighted shift unitary matrices in another basis with $\|N\| \leq \|S\|$. In particular, the weights in all the direct sums are between $\min_k|c_k|$ and $\max_k|c_k|=\|S\|$. 

Also, if $S$ is real then $N$ is real and the basis in which $N$ is a direct sum of real unitary weighted shift matrices is obtained using a real orthogonal matrix. 

The same conclusion holds if instead of the commutator estimate above we have that all the weights $c_k$ satisfy $|c_k|\geq \sigma$ and we have the commutator estimate\[\|\,[S^\ast, S]\,\| < \frac{2\sigma}{M^{2}}.\]
\end{lemma}
\begin{proof}
The proof proceeds in four steps. Before step 1, we provide some inequalities used in the proof. In the first step we show that we can group the basis vectors into blocks that roughly correspond to level sets of the $|c_k|$. In the second step, we lay out how to perturb $S$ on certain pairs of basis vectors to obtain $N$. In the third step, we verify the norm inequality for $\|N-S\|$. In the fourth step we verify that $N$ is normal.

As with the weights $c_k$, all intervals of indices that we construct will be cyclically indexed by integers. 
Because all such intervals will be proper subsets of the set of all indices, 
it makes sense to use the terminology of ``first'' and ``last'' entry of such an interval to refer to the left-most and the right-most element due to the orientation of increasing the indices cyclically.

We first perform some estimates. We know that
\[||c_{k+1}|^2-|c_k|^2| < \frac1{M^3}.\]
This implies that 
\[\sqrt{||c_{k_1}|^2 - |c_{k_2}|^2|} < 
\sqrt{\frac{|k_1-k_2|}{M^3}}.\]
We relate this to an estimate for the differences of the absolute values of the weights. For $x, y \in \C$, 
\[||x|-|y||=\sqrt{(|x|-|y|)^2} \leq \sqrt{||x|-|y||(|x|+|y|)} =\sqrt{||x|^2-|y|^2|}.\]
Using this, we see that
\[||c_{k_1}| - |c_{k_2}||\leq \sqrt{||c_{k_1}|^2 - |c_{k_2}|^2|} < 
\sqrt{\frac{|k_1-k_2|}{M^3}}.\]

If we had the alternative restriction that $|c_k|\geq \sigma$ and $||c_{k+1}|^2-|c_k|^2| < 2\sigma/M^2$ then we would obtain the estimate: \[||c_{k_1}|-|c_{k_2}|| = \frac{||c_{k_1}|^2-|c_{k_2}|^2|}{|c_{k_1}|+|c_{k_2}|}\leq \frac{|k_1-k_2|\max_k ||c_{k+1}|^2-|c_k|^2|}{2\sigma}< \frac{|k_1-k_2|}{M^2}.\]       
So,
in either case
we have
\begin{align}\label{diffEst}
|k_1-k_2|\leq M \Rightarrow ||c_{k_1}|-|c_{k_2}||< \frac1M.
\end{align}

\underline{Step 1}: We now begin with the construction. Dividing $n$ by $M$ with remainder gives $q, d \in \N_0$ with $n = qM+d$ and $0 \leq d < M$. 

We first address the case where $n \leq 2M$. 
Because the distance is calculated cyclically, we see that the distance from $\max_j |c_j|$ to $\min_j |c_j|$ is less than $1/M$ by Equation (\ref{diffEst}).
We then change $c_k$ radially in $\C$ so that they all have the absolute value equal to $\frac{1}{2}(\max_j |c_j|+\min_j |c_j|)$. This provides a normal matrix $N$ with the desired properties and
\[\|N-S\| < \frac1{2M}.\]

We now assume that $n > 2M$.
Choose an integer $\tilde k$ so that $ |c_{\tilde k}| = \|S\|$. 
Then partition the sequence $1, \dots, n$ into the intervals $I_j'$ for $j = 0, \dots, q$ of consecutive integers as follows. 
We require all intervals to contain $M$ integers except the interval that contains $\tilde k$ which will contain $M+d$ integers. 
We will choose this particular interval so that there are $d$ integers to the left of $\tilde k$ and $M-1$ integers to its right. 
We relabel the basis vectors $e_k$ if necessary by cycling the indices (by at most $d$) so that $I_1'$ begins with $e_1$ to avoid any interval containing both $e_1$ and $e_n$ due to the shifting of the intervals when we included the additional $d$ indices in the interval containing $\tilde k$.

Because we assume that $n > 2M$, we then have that $I_j'$ are at least two consecutive disjoint intervals. Let $s_r = \|S\|-r/M$ for $0\leq r \leq r_0:=\lceil M\|S\|\rceil$. If $M\|S\|$ is an integer then $s_{r_0} = 0$. If it is not an integer, then $s_{r_0} < 0$. This provides a list of real numbers $s_r$:
\[\|S\|=s_0 > s_1 > \dots, s_{r_0-1} >0\geq s_{r_0}\]
spaced by $1/M$. For $c\in [0, \|S\|]$, we define the function $\sround(c):= \min\{s_r: s_r \geq c\}$ that ``rounds up'' to a nearby value of $s_r$. We know then that $\sround(c)=s_r$ for some $r$ and $\sround(c)-1/M< c \leq \sround(c)$.
We will replace all the weights in an interval $I_j'$ with a single absolute value $s_r$ now.

Let $A_j = \{|c_k|: k \in I_j'\}$. By Equation (\ref{diffEst}), we see that $\diam A_j < 1/M$ as follows. This is clearly true for the intervals $I_j'$ containing $M$ integers but also for the potentially longer interval since the index $\tilde k$ of a weight with maximum absolute value is less than $M$ away from the other integers in the interval.

So, we define $a_j = \sround(\min A_j)$. Then $a_j-1/M < \min A_j \leq a_j$ so that $A_j \cap (a_j-1/M, a_j] \neq \emptyset$ and $A_j \subset (a_j-1/M, a_j+1/M)$. 
In particular, $||c_k| - a_j| < 1/M$ for all $k \in I_j'$. Note that when $\max A_j = \|S\|$, because there is a distance of less than $M$ from a place where this maximum can take place this shows that $\min A_j > \|S\|-1/M$ so $a_j = s_0=\|S\|$.  Note also that $a_j \geq 0$ and the situation where $a_j = 0$ is only possible when both $s_{r_0} = 0$ and some weight in $I_j'$ equals zero.

Let $\underline{a}$ denote the smallest of the $a_j$. Choose a value $\underline{j}$ of $j$ so that $a_j=\underline{a}$ and then choose a value $\underline{k}$ of $k$ so that $\underline{k}$ lies in $I'_{\underline{j}}$. 
Using the change of basis like that indicated in Example \ref{ws-basisChange}, we see that $T$ is unitarily equivalent to a matrix with $c_k \geq 0$ except possibly $c_{\underline{k}}$. Each $k \neq \underline{k}$ lies in an interval $I_{j}'$ and we replace $c_k$ with $a_j$. 
We change $c_{\underline{k}}$ radially in $\C$ to have the absolute value equal to $a_{\underline{j}}$. 
Let $S_1$ denote this perturbation of $S$ so that
\[\|S_1-S\| < \frac1M.\]

If there is only one distinct value of $a_j$ then $S_1$ is normal and we are done. We will now assume that there are multiple distinct values of $a_j$.

We now show that consecutive weights $a_j$ are either equal or differ by at most $1/M$. 
Without loss of generality, suppose that $a_j < a_{j+1}$. Then there is a $k_j \in I_j'$ such that $|c_{k_j}| = \min A_j$. Because the intervals $I_{j}'$ and $I_{j+1}'$ are consecutive, there is an index $k_{j+1}$ of $I_{j+1}'$ that is within $M$ of $k_j$. So,
\[\min A_{j+1} \leq |c_{k_{j+1}}| \leq |c_{k_j}|+||c_{k_j}|-|c_{k_{j+1}}|| < \min A_j + \frac1M.\]
So, $a_{j+1} \leq a_j + 1/M$, which is what we wanted to show.

So, we now merge consecutive intervals $I_j'$ of the same weight $a_j$ to obtain reindexed intervals $I_j$ for $j = 1, \dots, q_0\leq q$ where the reindexed weights $a_j$ of the perturbed weighted shift matrix satisfy $a_{j+1} = a_j \pm1/M$.

\vspace{0.1in}

\underline{Step 2}:
We now need to determine how we will apply Lemma \ref{gel-forNearbyNormal}. The non-negative weights of $S_1$ are  spaced by $1/M$: $s_{n_0}> s_{n_0}-1/M > \dots > s_{n_1}$. 
The only weight that is potentially not non-negative is a single weight of minimal absolute value $s_{n_1}=\underline{a}$. 
Now, for a non-negative weight $b$, let $J_b$ be the level set  $\bigcup_{j: a_j \geq b} I_j.$ 
Then $J_b$ is the union of maximal sequences of consecutive intervals, each of the form $I_{j_0}, I_{j_0+1}, \dots, I_{j_1}$. We refer to $I_{j_0} \cup \cdots \cup I_{j_1}$ as a ``connected component'' of $J_b$ in analogy to how every open set in the unit circle is the disjoint union of countably many (connected) open arcs.

Define the integer $k_0 =
(M-2)/2$. So, each interval $I_j$ contains at least $2(k_0+1)$ integers.
Consider a connected component of $J_b $ with weight $b > \underline{a}$. 
Suppose that the connected component is formed by $I_{j_0}, \dots, I_{j_1}$. We will apply 
Lemma \ref{gel-forNearbyNormal} to obtain a perturbation of $S_1$ on the span of the first $k_0+1$ vectors of $I_{j_0}$, the last $k_0+1$ vectors of $I_{j_1}$, and the first vector of $I_{j_1+1}$. Namely, write $I_{j_0} = \{i_0, \dots, i_0'\}$ and $I_{j_1} = \{i_1, \dots, i_1'\}$ and observe that
$\#[i_0,i_0'] \geq M$ and $\#[i_1,i_1'] \geq M$. We define $v_i = e_{i_0+i}$  for $i = 0, \dots, k_0$ and $w_j= e_{i_1'-k_0+j}$  for $j = 0, \dots, k_0+1$.  Notice that $w_{k_0+1}$ is the first vector of $I_{j_1+1}$. We now apply Lemma \ref{gel-forNearbyNormal} to $S$ with $b$ and $a = b-1/M$. 

We do this for all such connected components of all such $J_b$ with $b > \underline{a}$. We claim that this provides the desired normal matrix $N$. 

\underline{Step 3}: 
We will obtain the estimate for  $\|N-S\|$. 
Because we perturb $S_1$ on orthogonal subspaces using Lemma \ref{gel-forNearbyNormal}, we see that 
\[\|N-S\|\leq \|N-S_1\| + \|S_1-S\| < \frac2M + \|S\|\frac{\pi}{M-2}.\]

\underline{Step 4}:
Because it is clear that $N$ satisfies the other conditions, we now prove that $N$ is normal as a direct sum of normal bilateral weighted shift operators. In order for each summand to be normal, it is necessary that each of these weighted shifts all have weights that have the same absolute value.

Consider a weight $b$.
There are three cases to consider. Compare the arguments for these three cases to Remark \ref{CaseRemark} which contains illustrations for them. 

We first consider a connected component of $J_b$ for $b > \underline{a}$ composed of $I_{j_0}, \dots, I_{j_1}$ which corresponds to cases 1 and 2 below. 
The value of $k_0$ was chosen so that $k_0+1 = \frac12M\leq \frac12\#I_{j_0}, \frac12\#I_{j_1}$.
Case 1 corresponds to when $I_{j_0}=I_{j_1}$ so that the $v_k$ and $w_k$ for $k=0,\dots, k_0$ are orthogonal vectors of the same interval $I_{j_0}$. 
Case 2 is when the intervals in question are distinct.
We now introduce some statements that apply for these first two cases.

Using the notation in Step 2, we can define the vectors $v_i$ and $w_j$ by rewriting the vectors $e_{i_0}, \dots, e_{i_1'}$ as
\begin{align}\label{vectors}
v_0, \dots, v_{k_0}, e_{i_0+k_0+1}, \dots, e_{i_1'-k_0-1}, w_0, \dots, w_{k_0}
\end{align}
and having $w_{k_0+1} = e_{i_1'+1}$. Because this connected component of $J_b$ is not all of the indices, we see that $e_{i_1'+1}$ belongs to $I_{j_1+1}$ and is thus orthogonal to the other vectors listed above. 
Note that by construction the vector  $e_{i_1'+1}$ is not included in any other applications of Lemma \ref{gel-forNearbyNormal} because it is the first vector of an interval that cannot be the first interval of a level set $J_b$ for any $b$.

The span of these vectors equals the span of these two groups of vectors:
\[e_{i_0+k_0+1}, \dots, e_{i_1'-k_0-1}, \eta_0, \dots, \eta_{k_0}\]
\[\xi_0, \dots, \xi_{k_0}.\]
Recall that $w_0 = \eta_0$ and $N\eta_{k_0} = -Sv_{k_0}=-be_{i_0+k_0+1}$.

Note that if this connected component has exactly $M$ indices (which can happen only in Case 1 below) then $i_0+k_0+1=i_1'-k_0$ so $e_{i_0+k_0+1}=w_0=\eta_0$ and $e_{i_1'-k_0-1}=v_{k_0}=-\eta_{k_0}$.
So, to avoid redundancies, it is best to think of the $e$ and $\eta$ list of vectors as just
\[\eta_0, \dots, \eta_{k_0}\]
where then $N\eta_{k_0}=-b\eta_0$.

Now, we know that $N$ acts on the second grouping of vectors as:
\begin{align}\label{Case1Teleport}
N: \xi_0 \overset{b-\frac1M}{\rightarrow}\xi_1\overset{b-\frac1M}{\rightarrow}\xi_2\overset{b-\frac1M}{\rightarrow}\cdots\overset{b-\frac1M}{\rightarrow} \xi_{k_0}\overset{b-\frac1M}{\rightarrow} w_{k_0+1}=e_{i_1'+1}.
\end{align}
The second grouping of vectors will be put together with vectors from $J_{b-1/M}$. We will now use this information directly for the first two cases.

\underline{Case 1}: In this first case, the connected component will not contain any interval $I_j$ of a higher weight $a_j$. We have the vectors in Equation (\ref{vectors}).
The vectors $v_k, w_k$ all correspond to vectors in $I_{j_0}$. With $M=2(k_0+1)$, we have at least this many indices in $I_{j_0}$: $i_0, \dots, i_0'$.

$N$ acts on the first grouping of vectors as a bilateral weighted shift with weights having absolute value $b$:
\begin{align}\label{Case1Loop}
N: e_{i_0+k_0+1}&\overset{b}{\rightarrow}e_{i_0+k_0+2}\overset{b}{\rightarrow} \cdots\overset{b}{\rightarrow} e_{i_1'-k_0-1}\overset{b}{\rightarrow}e_{i_1'-k_0}=\eta_0\overset{b}{\rightarrow}\eta_1\overset{b}{\rightarrow} \cdots\nonumber\\
&\overset{b}{\rightarrow}\eta_{k_0-1}\overset{b}{\rightarrow}\eta_{k_0}= -e_{i_0+k_0}\overset{-b}{\rightarrow}e_{i_0+k_0+1}.
\end{align}
So, the first grouping of vectors spans an invariant subspace of $N$ on which $N$ is normal. When $I_{j}$ has only $M$ indices, one should think of the above orbit of $N$ as
\[N:\eta_0\overset{b}{\rightarrow}\cdots\overset{b}{\rightarrow}\eta_{k_0}\overset{-b}{\rightarrow}\eta_0.\]

\underline{Case 2}: In this case, the connected component of $J_b$ will contain some intervals $I_j$ of higher weights $a_j > b$ and we also require that $b > \underline{a}$. We have the vectors in Equation (\ref{vectors}) with at least $M$ vectors between $v_{k_0}$ and $w_{0}$ coming from $J_{b+1/M}$. 
We decompose the middle block of vectors in (\ref{vectors}):
\begin{align*}
 e_{i_0+k_0+1}, \dots, e_{i_1'-k_0-1}
\end{align*}
as
\begin{align*}
&e_{j_1}, \dots, e_{j_2}, \overline{e}_{j_1^+}, \dots, \overline{e}_{j_2^+}, e_{j_3}, \dots, e_{j_4}, \overline{e}_{j_3^+}, \dots, \overline{e}_{j_4^+}, \cdots,\\ &e_{j_m}, \dots, e_{j_{m+1}}, \overline{e}_{j_m^+}, \dots, \overline{e}_{j_{m+1}^+}, e_{j_{m+2}}, \dots, e_{j_{m+3}}  
\end{align*}
where the block $\overline{e}_{j_r^+}, \dots, \overline{e}_{j_{r+1}^+}$ corresponds to each of the connected components of $J_{b+1/M}$ within the component of $J_b$ on which we are focusing. The remaining blocks of the form $e_{j_r}, \dots, e_{j_{r+1}}$ belong to $J_b$. Note that the first and/or last block of this form may be empty.

Based on Case 1 for $b+1/M$ or the (recursive) application of Case 2 for $b+1/M$, we obtain the passed-down vectors $\overline\xi_{k_r}, \dots, \overline\xi_{k_{r+1}}$ within the span of the block $\overline{e}_{j_r^+}, \dots, \overline{e}_{j_{r+1}^+}$ such that $\overline\xi_{k_{r}}=\overline{e}_{j_r^+}$ and by Equation (\ref{Case1Teleport}),
\begin{align}\label{Case2PartialLoop}
N:\overline{e}_{j_r^+}=\overline\xi_{k_r}\overset{b}{\rightarrow}\overline\xi_{k_r+1}\overset{b}{\rightarrow}\cdots\overset{b}{\rightarrow}\overline\xi_{k_{r+1}}\overset{b}{\rightarrow}e_{j_{r+2}}.
\end{align}

Now, for this case we will use the $\overline\xi$ and the $\eta$ vectors to make a closed orbit with the $e_{j}$ vectors of this block. The $\xi$ vectors will be passed down for use for $J_{b-1/M}$. So, putting together Equations (\ref{Case1Loop}) and (\ref{Case2PartialLoop}) we see that

\begin{align*}
\eta_0, \dots, \eta_{k_0}, 
\;&e_{j_1}, \dots, e_{j_2}, \overline\xi_{k_1}, \dots, \overline\xi_{k_2}, e_{j_3}, \dots, e_{j_4}, \overline\xi_{k_3}, \dots, \overline\xi_{k_4}, \cdots,\\ &e_{j_m}, \dots, e_{j_{m+1}}, \overline\xi_{k_m}, \dots, \overline\xi_{k_{m+1}}, e_{j_{m+2}}, \dots, e_{j_{m+3}} 
\end{align*}
form an invariant subspace for $N$ on which $N$ is a bilateral weighted shift with weights $\pm b$:

\begin{align*}
N: \eta_0 \overset{b}{\rightarrow} 
\cdots \overset{b}{\rightarrow} \eta_{k_0} \overset{-b}{\rightarrow}
\;&e_{j_1} \overset{b}{\rightarrow} 
\cdots \overset{b}{\rightarrow} e_{j_2} \overset{b}{\rightarrow}
\overline{e}_{j_1^+}=
\overline\xi_{k_1}\overset{b}{\rightarrow}  
\cdots \overset{b}{\rightarrow} \overline\xi_{k_2} \overset{b}{\rightarrow}\\
&e_{j_3}\overset{b}{\rightarrow} 
\cdots \overset{b}{\rightarrow} e_{j_4}\overset{b}{\rightarrow}\overline{e}_{j_3^+}= \overline\xi_{k_3} \overset{b}{\rightarrow} 
\cdots\overset{b}{\rightarrow} \overline\xi_{k_4}\overset{b}{\rightarrow} 
\cdots\overset{b}{\rightarrow}\\ &e_{j_m}\overset{b}{\rightarrow} 
\cdots\overset{b}{\rightarrow} e_{j_{m+1}}\overset{b}{\rightarrow}\overline{e}_{j_{m}^+}= \overline\xi_{k_m}\overset{b}{\rightarrow} 
\cdots\overset{b}{\rightarrow} \overline\xi_{k_{m+1}}\overset{b}{\rightarrow}\\
&e_{j_{m+2}}\overset{b}{\rightarrow}
\cdots\overset{b}{\rightarrow} e_{j_{m+3}}\overset{b}{\rightarrow} w_0 = \eta_{0} 
\end{align*}

Note that if one of the blocks of $e$ vectors is empty then the corresponding vectors would just be skipped in showing the orbit of $N$. For instance, if the first $e$ block is empty then we would instead have $\eta_{k_0} \overset{-b}{\rightarrow}
 \overline\xi_{k_1}$.

\underline{Case 3}: In this last case, $b = \underline{a}$. Focus on the intervals $I_j$ such that $a_j = \underline{a}$. The complement of the union of these intervals is $J_{\underline a + 1/M}$. Consider a connected component of $J_{\underline a + 1/M}$ as in Case 2. Consider the interval(s) $I_{\underline j_1}$ and $I_{\underline j_2}$ with weight $\underline a$ that are immediately before and after this connected component. When $J_{\underline a + 1/M}$ has one connected component, it is the case that $\underline j_1 = \underline j_2$ as in Figure \ref{wsSystem1}. Figure \ref{wsSystemExtended} illustrates a more general case.

Let $e_{i^\ell_0}, \dots, e_{i^\ell_1}$ be the vectors corresponding to $I_{\underline j_1}$ and $e_{i^r_0}, \dots, e_{i^r_1}$ be the vectors corresponding to $I_{\underline j_2}$. We can express the action of $N$ on these basis vectors as
\begin{align*}
N:&\;e_{i^\ell_0}\overset{\underline{a}}{\rightarrow} \cdots\overset{\underline{a}}{\rightarrow} e_{i^\ell_1}\overset{\underline{a}}{\rightarrow} e_{i^\ell_1+1}\\
N:&\;e_{i^r_0}\overset{\underline{a}}{\rightarrow} \cdots\overset{\underline{a}}{\rightarrow} e_{i^r_1}\overset{\underline{a}}{\rightarrow}e_{i^r_1+1}
\end{align*}
generically. It is possible that a single one of these weights is not positive but instead just has absolute value equal to $\underline{a}$.

We proceed in a way similar to Case 2 except that we do not change any of the vectors of the lowest weight because the original operator that we started with was a bilateral shift.
Based on Case 1 or the application of Case 2 for $\underline a+1/M$, we obtain vectors $\overline\xi_0, \dots, \overline\xi_{k_0}$ such that $\overline\xi_0=e_{i^\ell_1+1}$ and by Equation (\ref{Case1Teleport}),
\begin{align}
\nonumber
N:e_{i^\ell_1+1}=\overline\xi_0 \overset{\underline a}{\rightarrow}\overline\xi_1\overset{\underline a}{\rightarrow} \cdots\overset{\underline a}{\rightarrow}\overline\xi_{k_0}\overset{\underline a}{\rightarrow} e_{i^r_0}.
\end{align}

This shows that by including the vectors $\overline \xi_k$ that were passed down as follows:
\[e_{i^\ell_0}, \dots, e_{i^\ell_1}, \overline\xi_0, \dots, \overline\xi_{k_0}, e_{i^r_0}, \dots, e_{i^r_1}\]
then $N$ maps each vector in the list to the next multiplied by $\pm \underline{a}$ except perhaps the last vector as its image might be orthogonal the span of the vectors listed here.

However, once we have included all the vectors that were passed down from the connected components of $J_{\underline a + 1/M}$ we see that this provides a subspace on which $N$ acts as a bilateral weighted shift with weights having absolute value $\underline{a}$.

This completes the verification and also the proof of this lemma.

\end{proof}
\begin{remark}\label{CaseRemark}
\begin{figure}[htp]  
    \centering
    \includegraphics[width=14cm]{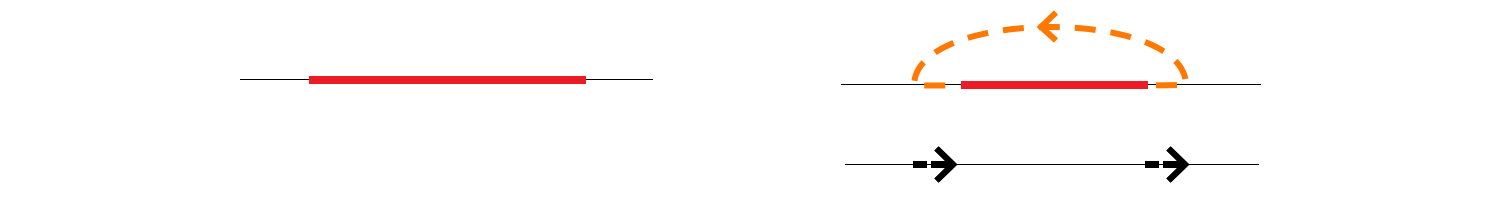}
    \caption{\dark \label{wsSimpleCase1}\dark 
    Illustration of Case 1 in the proof of Lemma \ref{HelpingBerg}.}
    \end{figure}

In this remark, we discuss Figures \ref{wsSimpleCase1}, \ref{wsSimpleCase2}, and \ref{wsSimpleCase3} as illustrations of the constructions in cases 1, 2, and 3, respectively, in the proof of Lemma \ref{HelpingBerg}.

\underline{Case 1}:
The red line on the left side corresponds to the vectors $e_i$ that correspond to a connected component of the interval $J_b$. One should think of $e_{i}$ as a point on this red line that moves from the left-most part of the red line to its right-most point as $i$ increases from $i_0$ to $i_1'$. The reason that we have singled out these specific basis vectors with a red line is that they have weight $b$ for $S'$. The thin black lines starting before and continuing after the red line segment correspond to basis vectors $e_i$ for $i < i_0$ and $i > i_1'$, respectively, and will have potentially different weights because they do not belong to this connected component of $J_b$.

The right side of this figure illustrates $N$ acting on the vectors $\eta_k$ and $e_i$ and the $\xi_k$. The orbit of $N$ in Equation (\ref{Case1Loop}) is illustrated in the top right side of this figure. The red line corresponds to $e_{i_0+k_0+1}, \dots, e_{i_1'-k_0-1}$ and the orange loop corresponds to the action of $N$ on the $\eta_k$. The weights of $N$ on this orbit are the same as the weights of the red line, namely $b$.

The action of $N$ on the $\xi_k$ is illustrated in the line diagram on the bottom right of this figure. 
The vectors $\xi_k$ belong to the span of the vectors $e_i$ that correspond to the beginning and ending portions of the red line that vertically line up with the two arrows in the diagram. 
Because $\xi_0=e_{i_0}$ and $\xi_{k_0}=e_{i_1'}$, we view the action of $N$ on the $\xi_k$ as a perturbation of $S$ with the orbit of $N$ to starting at $e_{i_0}$ and ``teleporting'' to $e_{i_1'}$  with the $\xi_k$ being orthogonal to the span of the $e_i$ that correspond to the red line above it (the vectors that are not equal to a $v_k$ or $w_k$). 
The positioning of this diagram below the other diagram on the right side is to illustrate that the weight of $N$ on the $\xi_k$ is $b-1/M < b$. This will be ``passed down'' to constructions in cases 2 and 3.

\begin{figure}[htp]  
    \centering
    \includegraphics[width=14cm]{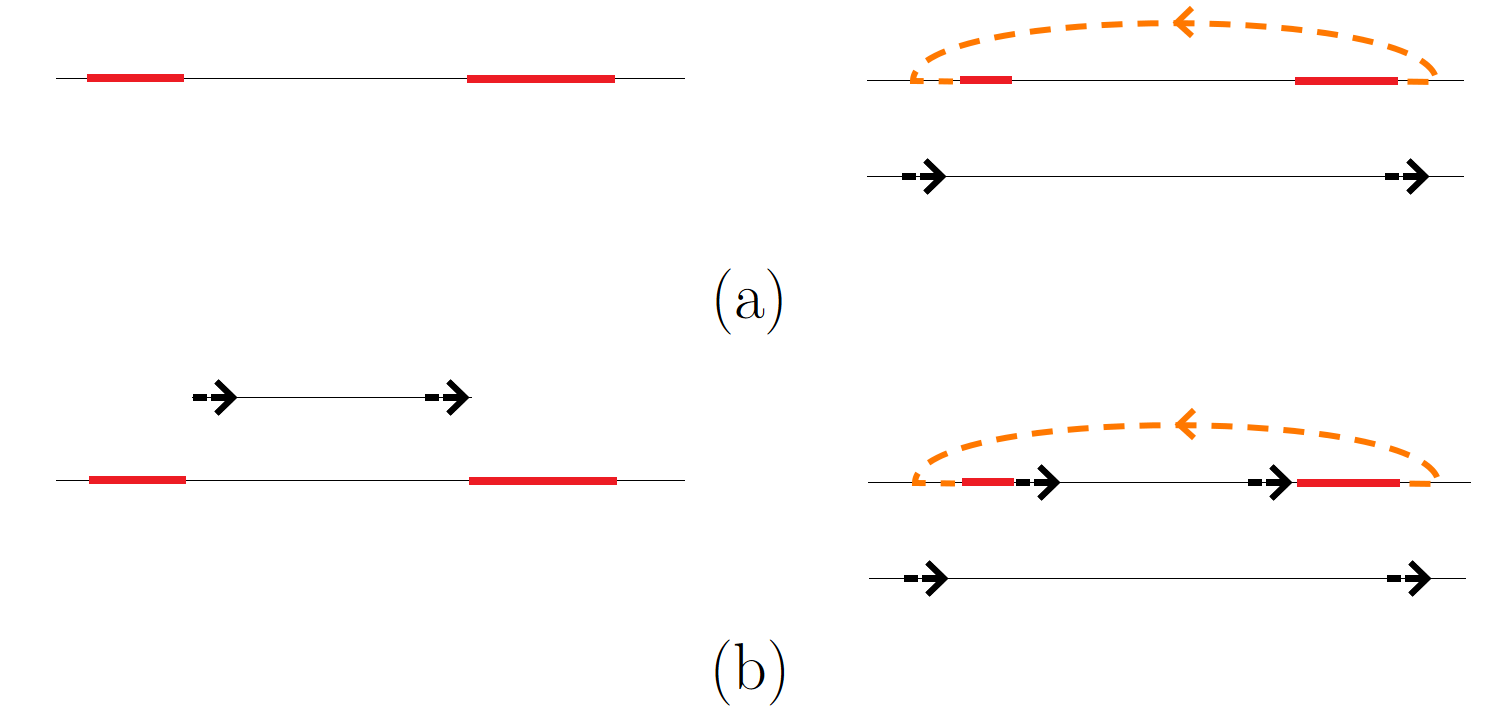}
    \caption{\dark \label{wsSimpleCase2} Illustration of Case 2 in the proof of Lemma \ref{HelpingBerg}.}
\end{figure}

\underline{Case 2}: Figure \ref{wsSimpleCase2}(a) is an illustration similar to that of Figure \ref{wsSimpleCase2} with the exception that there is a gap in the red line because the connected component of $J_b$ has vectors that have weight higher than $b$. The main difference here is that the top diagram on the right side of Figure \ref{wsSimpleCase2}(a) does not represent an invariant subspace of $N$ because the right-most point of the left subset of the red line indicates that $N$ will map that vector to the black line, which is outside the orbit that we are considering.

The resolution of the fact that we do not obtain an invariant subspace in (a) is to include two arrows composed of some $\xi_k$ with weight $b$ originating from $J_{b+1/M}$. The left side of (b) shows that we are including this so that on the right side of (b) will have a closed orbit. The bottom two arrows on the right side of (b) will have weight $b-1/M$ and will be passed down to the construction for $J_{b-1/M}$. 

Note that (b) illustrates the case where the portion of $J_{b+1/M}$ in the connected component of $J_{n}$ on which we are focused is made of only one connected component. For an example where the relevant portion of $J_{b+1/M}$ contains two connected components, see the second-to-the-bottom line in Figure \ref{wsSystem1}(a) and Figure \ref{wsSystem1}(b).

\underline{Case 3}: 
\begin{figure}[htp]  
    \centering
    \includegraphics[width=14cm]{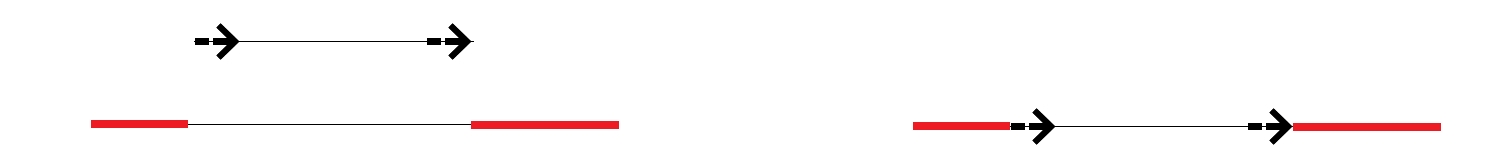}
    \caption{\dark \label{wsSimpleCase3} Illustration of Case 3 in the proof of Lemma \ref{HelpingBerg}. \dark}
\end{figure}
Case 3 does not have any change to the basis vectors in red. The only issues that can arise is when the there are gaps in the lowest weight intervals $I_j$ due to there being weights greater than $\underline{a}$. However, the $\xi_k$ that are passed down removes this difficulty. This is illustrated in the figure in that the passed down arrows with weights $(\underline{a}+1/M)-1/M=\underline{a}$. 

Note that in this illustration the red line on the right is not begun or ended by a black line. This indicates that the red line is a single segment (viewed cyclically) because it contains $e_1$ and $e_n$. The bottom row of Figure \ref{wsSystemExtended} illustrates a slightly more general scenario of having $J_{\underline{a}+1/M}$ with two connected components so that there are two lowest weight intervals $I_j$.
\end{remark}

\begin{example}
We now provide two visual examples of the construction of the normal matrix $N$ in Lemma  \ref{HelpingBerg}. Figure \ref{wsSystem1} provides an illustration of such an example, starting with the weights perturbed as described in the proof of the lemma in (a) and showing the constructed $N$ in (b) using the diagrams described in Remark \ref{CaseRemark}. 
\begin{figure}[htp]  
    \centering
    \includegraphics[width=14cm]{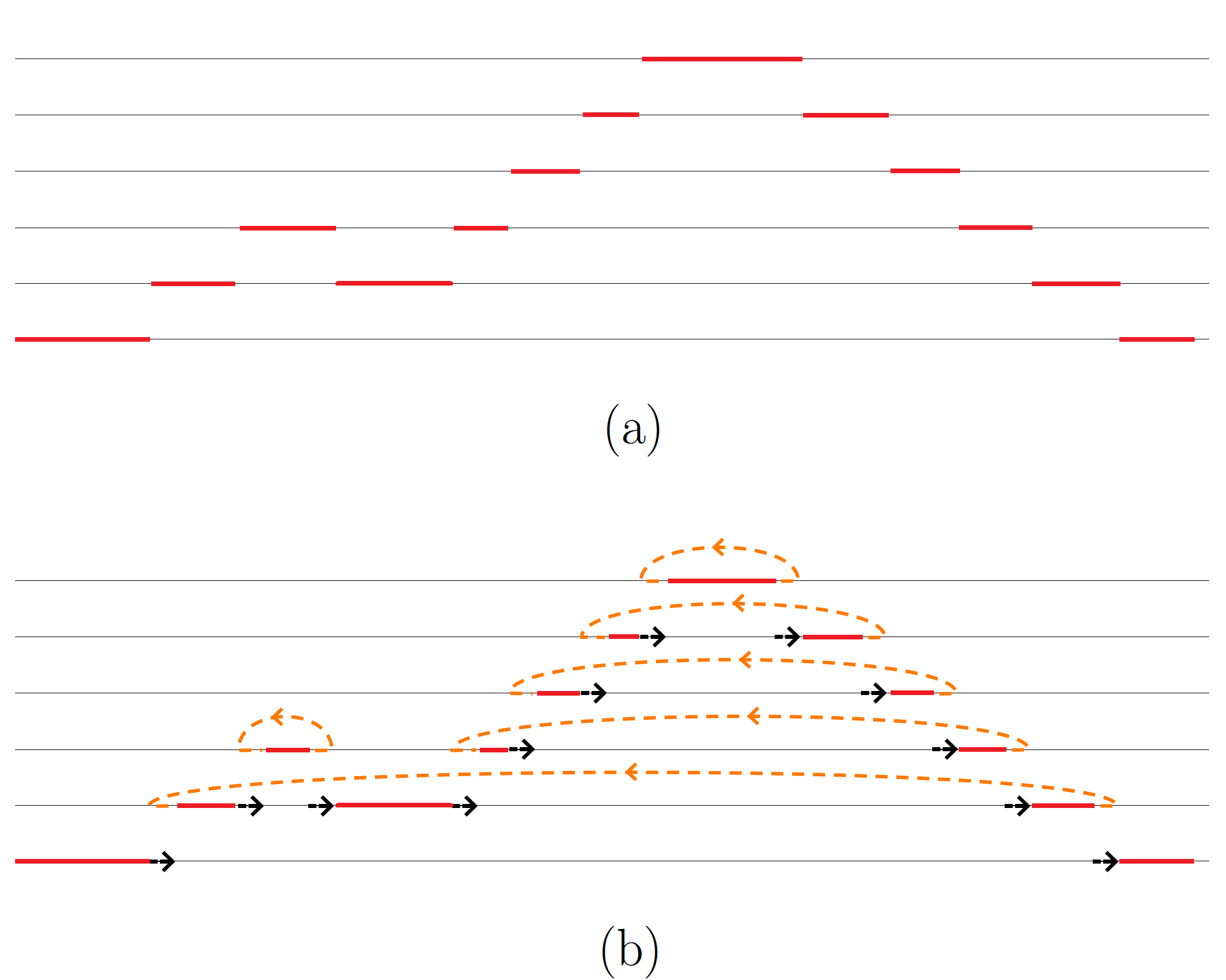}
    \caption{\dark \label{wsSystem1} Illustration of construction in Lemma  \ref{HelpingBerg}.  \dark}
\end{figure}

Figure \ref{wsSystemExtended} provides a more general example of the construction where $J_{\underline{a}+1/M}$ has two connected components.
\begin{figure}[htp]  
    \centering
    \includegraphics[width=16cm]{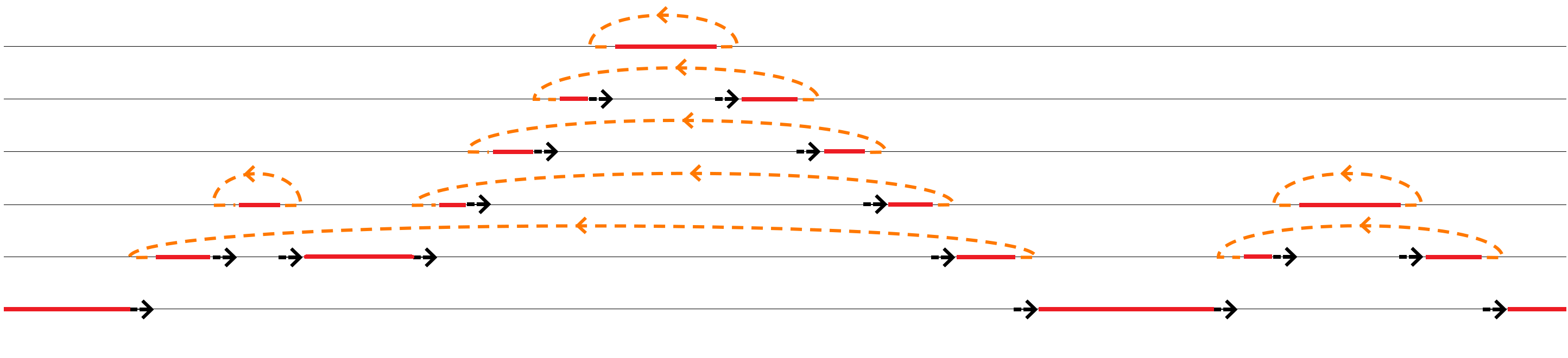}
    \caption{\dark \label{wsSystemExtended}  \dark
     Illustration of construction in Lemma  \ref{HelpingBerg}.}
\end{figure}
\end{example}

\begin{remark}\label{noNegatives}
The constructed normal matrix $N$ is a direct sum of bilateral weights shifts with weights $\pm b$ except the lowest weighted shift which may have a complex phase if the weights of $S$ were complex. 
However, it is possible to change the construction so that the bilateral shift with lowest absolute value weights is the only summand with a non-positive weight. 
Further, if $S$ is a unilateral shift, then all the weights of the summands of $N$ can be made non-negative even though only the lowest weight summand is a unilateral shift.

We presently have no need for this modification so for us such a modification would be purely aesthetic, but we discuss it nonetheless.
We modify the construction to minimize the number of negative signs left after our applications of Lemma \ref{gel-forNearbyNormal}. This same effect is accomplished by Berg's original construction due to the use of complex phases even if $S$ is real, but we opt for a different approach so that we obtain the structured result that $N$ is real if $S$ is.

One way to modify the proof is as follows.
First note that we will either use Lemma \ref{gel-forNearbyNormal} as stated or a modified form of Lemma \ref{gel-forNearbyNormal} that has a different definition of the rotated basis: $\tilde\xi_k=\alpha_kv_k - \beta_kw_k$ and $\tilde\eta_k=\beta_kv_k+\alpha_kw_k$ so that the $\tilde\eta_k$ satisfy $\tilde\eta_0 = w_0, \tilde\eta_{k_0}= v_{k_0}$ and the $\tilde\xi_k$ have ``the negative sign'': $\tilde\xi_0 = v_0, \tilde\xi_{k_0}=-w_{k_0}.$
We will apply one of the versions of the lemma so that the number of weights with a negative sign in invariant orbit of $N$ for Case 1 or Case 2 is even. This way, a simple change of variables in this invariant subspace for the $e_i, \eta_k$ orbit will result in all the weights being positive. Ultimately, the choice of which version of Lemma \ref{gel-forNearbyNormal} to apply will affect the choice for smaller weights due to the signs of the weights of the passed down vectors $\xi_k$.

Note that we can determine which passed down vectors will carry down a negative sign by noting that Case 1 always passes down a negative sign and Case 2 always passes down one (modulo two) negative sign more than the sum of the negative signs passed down to it.

We repeat this process where each $J_b$ for $b>\underline{a}$ will pass down some negative signs, at most one from each of its connected components. We then come to Case 3. This is the only place where we cannot remove the negative sign if $\underline{a}>0$. 

If $\underline{a}$ is close to zero, then we can replace it with zero with a small additional error. This is possible if $S$ is a unilateral weighted shift. If $\underline{a}$ is far away from zero then we might not be able to  remove a last remaining negative sign of the lowest weight bilateral shift with this method even with a perturbation.
\end{remark}

We return to Lin's theorem for a weighted shift matrix.
Reformulating the previous lemma, we obtain the following theorem. This first inequality is inherent to Theorem 2 of \cite{BergNWS} with $C = 100$ and exponent $1/4$. Additionally, this result applies to not just unilateral shifts and we have the two additional properties of $N$ stated at the end of the statement of the theorem. The ability to choose $N$ real is an improvement on the construction of Berg's original proof as well as the greatly reduced constant. We also obtain a second construction in a more specific case that provides the optimal exponent.
\begin{thm}\label{BergResult}
Suppose that $S \in M_n(\C)$ is a bilateral weighted shift matrix. Then there is a normal matrix $N$ such that
\begin{align}\label{BergIneq1}
\|N-S\| \leq C_{\alpha} \|S\|^{1-2\alpha}\|[S^\ast,  S]\|^{\alpha}
\end{align}
for $\alpha=1/3$ and $C_{1/3} < 5.3308.$
Further, $N$ is equivalent to a direct sum of bilateral weighted shift operators, $\|N\|\leq \|S\|$, and if $S$ is real then $N$ is real.

If the weights of $S$ all have absolute value at least $\sigma$ then $N$ can be chosen with the above properties but the alternate estimate\begin{align}\label{BergIneq2}\|N-S\| \leq 4.8573\sqrt{\frac{\|S\|\,}{\sigma}} \|[S^\ast, S]\|^{1/2}.\end{align}
\end{thm}
\begin{remark}
Note that Equation  (\ref{BergIneq1}) is asymptotically weaker than the optimal upper estimate $\|N-S\|\leq C_{1/2}\|[S^\ast, S]\|^{1/2}$  by using \[\|[S^\ast, S]\|^{1/2}=\|[S^\ast, S]\|^{1/2-1/3}\|[S^\ast, S]\|^{1/3} \leq (2\|S\|^2)^{1/6}\|[S^\ast, S]\|^{1/3}.\]  
Equation  (\ref{BergIneq2}) is also weaker than the optimal upper estimate since $\|S\|\geq \sigma$. However, when $\sigma/\|S\|$ is not too small Equation  (\ref{BergIneq2}) can be of great use due to the small constant. 

The proof of the optimal estimate in \cite{KS} does not provide a value of $C_{1/2}$, however it appears from the proof that it will be much larger than $C_{1/3}$ given above. For this reason, Equation  (\ref{BergIneq1}) will still be of use in addition to the simplicity of the construction of $N$ and the additional structure of $N$.

In our application to Ogata's theorem, we will have almost normal (unilateral) weighted shifts and hence will not be able to procure a usable lower bound for the absolute values of the weights. So, Equation  (\ref{BergIneq1}) with $\alpha = 1/3$ will be of use to us in later sections.
\end{remark}

\begin{proof}
Assume that $\|S\|\leq s$.  Let $x = \|[S^\ast, S]\|$. 
Note that if $M\geq 4$ is an even integer, then when $x < M^{-3}$ the normal matrix constructed in Lemma \ref{HelpingBerg} satisfies the properties therein. 

Let $M_0\geq4$ be a real number. Consider the case that $x \leq (M_0+2)^{-3}$ so that $M_0+2 \leq x^{-1/3}$ and define \[M = 2\left(\left\lceil \frac{x^{-1/3}}2 \right\rceil -1\right)\] so that $M$ is an even integer that satisfies
\[M < 2\left( \frac{x^{-1/3}}{2}\right)=  x^{-1/3},\] hence $x<M^{-3}$. Also, \[M \geq 2\left( \frac{x^{-1/3}}2 -1\right) \geq 2\left( \frac{M_0+2}{2} -1\right) = M_0.\]

Apply Lemma \ref{HelpingBerg} to obtain a normal matrix $N$ with the properties from that lemma. Because $t \mapsto t/(t-2) = 1 + 2/(t-2)$ for $t > 2$ is decreasing and $M\geq M_0$, we have
\[\frac{M}{M-2}\leq \frac{M_0}{M_0-2}.\] 
Since
\[\frac M2\geq\frac{x^{-1/3}}{2} -1 \geq \frac{x^{-1/3}}{2}-\frac{x^{-1/3}}{M_0+2} = \frac{M_0}{2(M_0+2)}x^{-1/3},\]
we have
\begin{align}\label{HelpingEstimate}
\|N-S\| < \left(s\frac{\pi M}{M-2}+2\right)\frac{1}{M} \leq \left(s\frac{\pi M_0}{M_0-2}+2\right)\frac{M_0+2}{M_0}x^{1/3}.\end{align}
We have obtained an estimate when $x \leq (M_0+2)^{-3}$. 

If $x > (M_0+2)^{-3}$ then  we can choose $N = 0$ so that
\[\|N-S\| = \|S\| \leq s \leq s(M_0+2)x^{1/3}.\]
So, putting this case together with Equation (\ref{HelpingEstimate}) we have some normal matrix $N$ such that
\[\|N-S\| \leq \max\left(s(M_0+2), \left(s\frac{\pi M_0}{M_0-2}+2\right)\frac{M_0+2}{M_0}  \right)x^{1/3} = f(s,M_0)x^{1/3}.\]

In general, when $\|S\| > 0$ apply this result to the rescaled $\tilde S = \frac r{\|S\|}S$ with norm $r$ to obtain a normal $\tilde N$. With $N = \frac{\|S\|}r\tilde N$, we have
\[\|N-S\| = \frac{\|S\|}r\|\tilde N - \tilde S\| \leq 
\frac{\|S\|}rf(r, M_0)\|[\tilde S^\ast, \tilde S]\|^{1/3} = \frac{f(r, M_0)}{r^{1/3}}\|S\|^{1/3}\|[S^\ast,  S]\|^{1/3}.\]

So, we want to choose $r$ and $M_0$ to minimize
\[\frac{f(r, M_0)}{r^{1/3}}= \max\left(r^{2/3}(M_0+2), \left(r^{2/3}\frac{\pi M_0}{M_0-2}+\frac2{r^{1/3}}\right)\frac{M_0+2}{M_0}  \right).\]
We choose $M_0=15.937$ and $r = 0.162$ to obtain the $f(r, M_0)r^{-1/3}< 5.3308$.

\vspace{0.05in}

We obtain the second result as follows. Let $x = \|[S^\ast, S]\|/2\sigma$ and define $M_0$ as above. We now change the definition of $M$ to instead have an exponent $\alpha = 1/2$. We assume that $x \leq (M_0+2)^{-2}$ so that $M_0+2 \leq x^{-1/2}$ and define \[M = 2\left(\left\lceil \frac{x^{-1/2}}2 \right\rceil -1\right)\] analogous to what is done above.
Then \[\|[S^\ast, S]\| = 2\sigma x < 2\sigma M^{-2}\]
and $M \geq M_0$ as before.

As  before, \[\|N-S\| \leq \max\left(s(M_0+2), \left(s\frac{\pi M_0}{M_0-2}+2\right)\frac{M_0+2}{M_0}  \right)x^{1/2} = f(s,M_0)\left(\frac{\|[S^\ast, S]\|}{2\sigma}\right)^{1/2}.\]
If we perform the same change of variables $\tilde S = \frac{r}{\|S\|}S$ then the weights of $\tilde S$ have absolute value at least $\tilde \sigma = \frac{r\sigma}{\|S\|}$. So, as before we obtain normals $\tilde N$ and $N$ so that
\[\|N-S\| = \frac{\|S\|}r\|\tilde N - \tilde S\| \leq 
\frac{\|S\|}r\frac{f(r, M_0)}{(2\tilde \sigma)^{1/2}}\|[\tilde S^\ast, \tilde S]\|^{1/2} = \frac{f(r, M_0)}{(2r)^{1/2}}\left(\frac{\|S\|}{\sigma}\right)^{1/2}\|[S^\ast,  S]\|^{1/2}.\]
Choosing $M = 10.762$ and $r = 0.2897$ provides the estimate.

\end{proof}

\section{Gradual Exchange Process}
\label{GEP-Section}

We begin this section by motivating the construction in Lemma \ref{proto-gep}. The proceeding lemmas: Lemmas \ref{proto-gep2} and \ref{gep} are generalizations of this lemma that we will need for the main result of the paper. 

Recall that several of the counter-examples of almost
commuting matrices that are not nearly commuting have the same structure: a diagonal matrix $A$ and a weighted shift matrix $S$, where there is a lower bound on the absolute value of the weights of $S$ over a long span of the spectrum of $A$. Consider the following example, which is essentially Example 2.1 of Hastings and Loring's \cite{HastingsLoring}.
\begin{example}\label{repEx}
Let $A = \frac{1}{\lam}S^{\lam}(\sigma_3)$ and $S = \frac{1}{\lam}S^{\lam}(\sigma_+)$. Recall that by Equation (\ref{repAlmostCommuting}), $A$ and $S$ are almost commuting. 
Note that $A$ is self-adjoint and $S$ is almost normal.
Using an invariant called the Bott index, \cite{HastingsLoring} shows that there are no nearby commuting matrices $A', S'$ with $A'$ self-adjoint and $S'$ normal.

Written in matrix form, these are
\[A = 
\bp 
-1& & & & \\
 &\displaystyle-1+\frac1\lam& & & \\
 & & \displaystyle-1 + \frac2\lam & &\\
 & & &\ddots& \\ 
 & & & & 1 \\
\ep, \;\; S=  
\bp 
0  & & & &\\
\displaystyle\frac{d_{\lam,-\lam}}\lam&0& & &\\
 &\displaystyle\frac{d_{\lam, -\lam+1}}\lam&0& &\\
 & &\ddots& \ddots&\\
 & & &\displaystyle\frac{d_{\lam,\lam-1}}\lam&0
\ep.\]
Note that for $|i| \leq \lam/2-1$,
\[\displaystyle\frac{d_{\lam, i}}{\lam} = \sqrt{\frac{\lam(\lam+1)}{\lam^2} - \frac{i(i+1)}{\lam^2} } \geq \sqrt{1 - \frac14} = \frac{\sqrt3}2.\]
This sort of lower bound on the weights of $S$ is a crucial part of why $A$ and $S$ are not nearly commuting as we illustrate using the following construction.
\end{example}

\begin{example}\label{gepMotivation}
Suppose that $A = \diag(a_1, a_2, \dots, a_n)$ and $S = \ws(c_1, \dots, c_{n-1})$ where $a_1 < \dots < a_n$ and $c_i \geq 0$.
\begin{figure}[htp]     \centering
    \includegraphics[width=15cm]{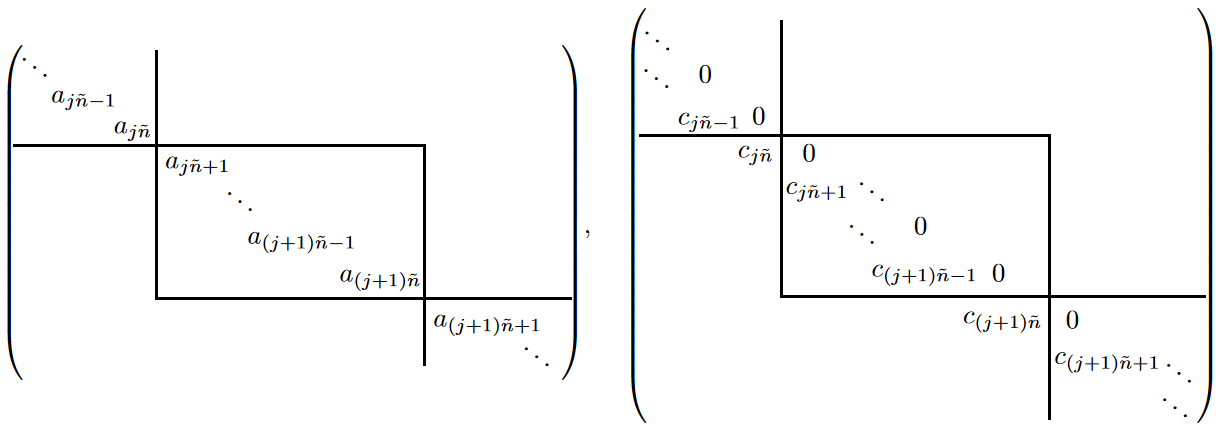}
    \caption{\label{AS-matrices}\dark
    Illustration of $A$ and $S$, respectively.}
\end{figure}
We will suppose further that $A$ and $S$ are nearly commuting: $\max_i |a_{i+1}-a_i||c_i|$ is small. So, if the $a_i$ are close then the $c_i$ are not required to be too small.

For the sake of the example, suppose that for some $\tilde{n} \ll n$ that divides $n$, it is true that 
$c_{\tilde{n}}, c_{2\tilde{n}}, \dots, c_{n-\tilde{n}}$ are no greater than some constant $D > 0$. Then
define $S'$ to be the linear operator where the weights $c_{j\tilde{n}}$ for $0< j < n/\tilde{n}$ of $S$ are replaced with zero. Then 
\[\|S'-S\| \leq D.\]

Let $J_j = [j\tilde{n}+1,  (j+1)\tilde{n}]$. Now, for $j \geq 0$, the subspaces $\mathcal W_j = \spn_{i \in J_j} e_i$ are invariant under $S'$. 
So, let $A'$ be an operator that is a multiple $a_j'$ of the identity when restricted to $\mathcal W_j$. If $F_j$ is the projection onto $\mathcal W_j$ then $A' = \sum_j a_j' F_j$. We choose then $a_j' = (a_{j\tilde{n}+1} + a_{(j+1)\tilde{n}})/2$ so that 
\[\|A' - A\| \leq \frac12\max_j \diam \{a_i: i  \in J_j\}.\]
We then see that if $D$ is small and the eigenvalues $a_i$ do not vary much for $a_i \in J_j$ then $A$ and $S$ are nearly commuting. The second condition can be restated as the property that the orbits of $S'$ do not span long stretches of the spectrum of $A$.

Expressed in matrix form, this construction replaces the almost commuting matrices $A$ and $S$ given in Figure \ref{AS-matrices} with the commuting matrices $A'$ and $S'$ given in Figure \ref{A'S'-matrices}, respectively.
\begin{figure}[htp]     \centering
    \includegraphics[width=15cm]{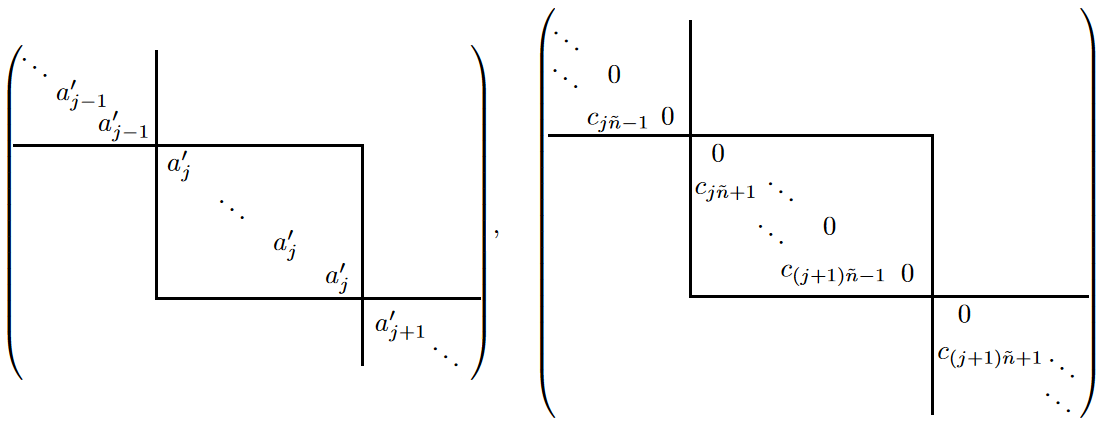}
    \caption{\label{A'S'-matrices}\dark
    Illustration of $A'$ and $S'$, respectively.}
\end{figure}

\end{example}

Example \ref{repEx} and the argument in Example \ref{gepMotivation} are illustrated in Figure \ref{WSheatmap_GEP_Motivation_SingleWS}. The weights $c_i$ for $i = 1, \dots, n=100$ in Figure \ref{WSheatmap_GEP_Motivation_SingleWS}(b) are $(1+\sin(i/\sqrt{n}))/2$.
\begin{figure}[htp]     \centering
    \includegraphics[width=12cm]{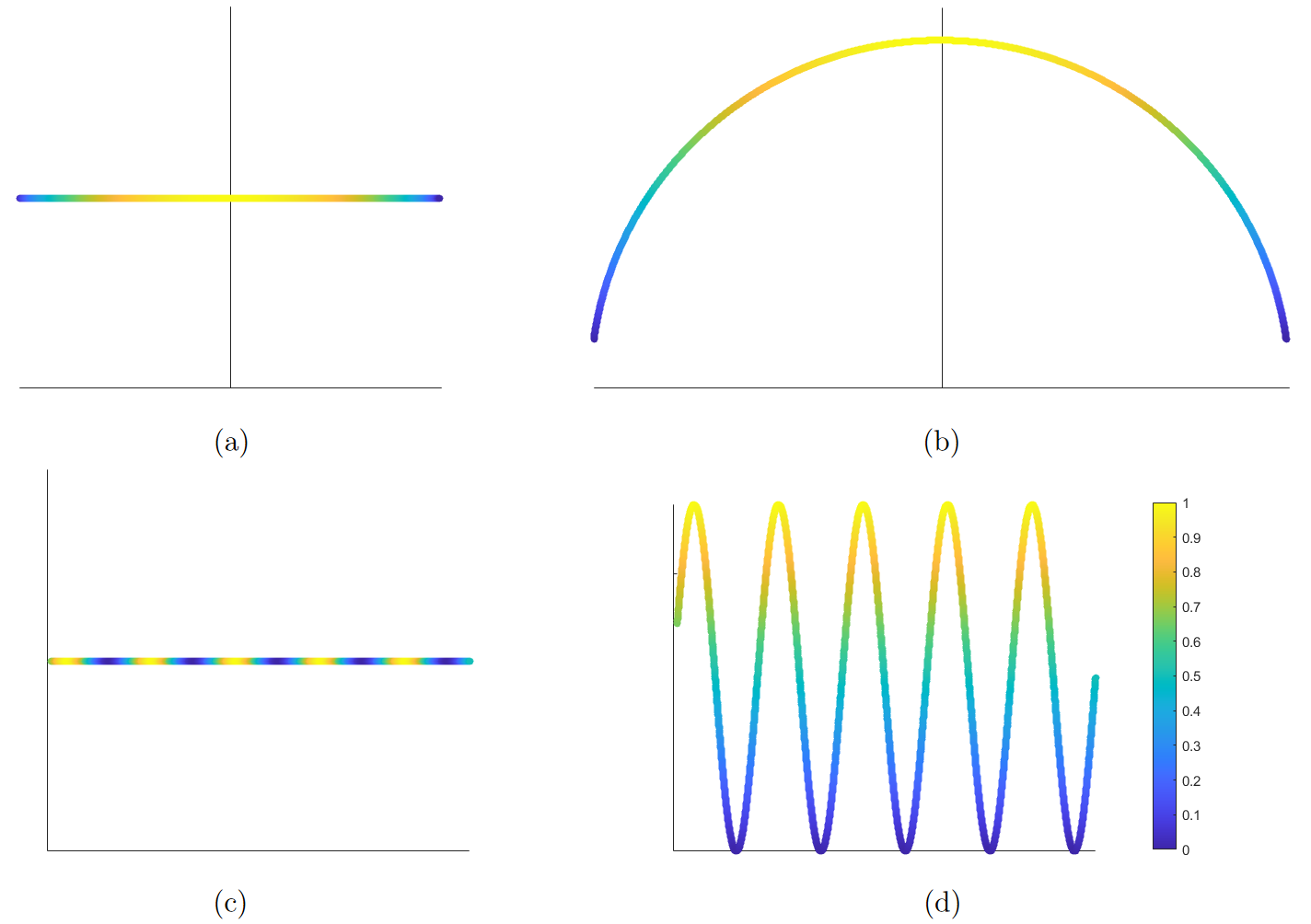}
    \caption{\label{WSheatmap_GEP_Motivation_SingleWS}\dark
    Example \ref{repEx} is illustrated in row (a) and (b). An example similar to Example \ref{gepMotivation} is illustrated in row (c) and (d). The graphs (a) and (c) on the left are weighted shift diagrams (without the arrow) colored according to the values of the weights at that point in the spectrum. The graphs (b) and (d) on the right illustrate these weights as graphs.}
\end{figure}
The construction of $S'$ from $S$ can be illustrated in weighted shift diagrams as in Figure \ref{WeightedShift_Decomposition}.

\begin{figure}[htp]     \centering
    \includegraphics[width=12cm]{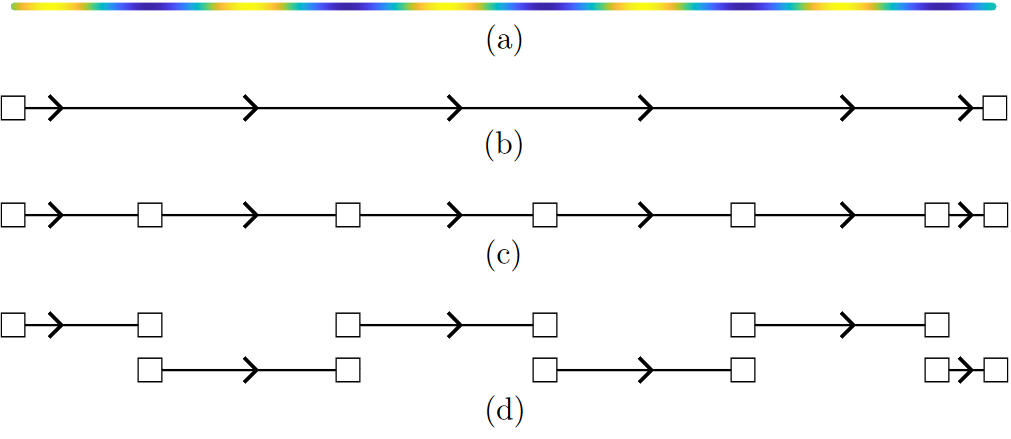}
    \caption{\label{WeightedShift_Decomposition}\dark
    The weighted shift diagram (without the arrow) from Figure \ref{WSheatmap_GEP_Motivation_SingleWS}(c) is illustrated in (a). 
    The rest of this figure illustrates how the construction in Example \ref{gepMotivation} changes the weighted shift matrix to obtain $S'$. A (standard) weighted shift diagram is seen in (b).
    In (c), we break the single orbit in (b) into smaller orbits by replacing  some of the small weights with zero. In (d), we then view this broken weighted shift diagram as the direct sum of numerous weighted shift diagrams with smaller orbits.}
\end{figure}

\vspace{0.1in}

To state the problem that we address in this section, suppose that $A$ and $S$ are given by block matrices:
{
\begin{align}\label{blockMatrices}
{\setlength{\arraycolsep}{0pt}A = 
\begin{pmatrix} 
\alpha_1 I_{k_1}&&&&\\
&\alpha_2 I_{k_2}&&&\\
&&\ddots&&\\
&&&\alpha_{n-1}I_{k_{n-1}}&\\&&&&\alpha_n I_{k_n}\end{pmatrix},\;} 
{\setlength{\arraycolsep}{2pt}
S = 
\begin{pmatrix} 0&&&&\\C_1&0&&&\\&C_2&\ddots&&\\&&\ddots&0&\\&&&C_{n-1}&0\end{pmatrix},}
\end{align}
}
where the $\alpha_i$ are distinct and each $C_i \in M_{k_{i+1}\times k_i}(\C)$ is ``diagonal'', with its only non-zero entries being those with the same row and column number.
We are trying to construct nearby commuting matrices $A'$ and $S'$. We also want to perturb $S'$ to $S''$ that is additionally normal.

If many of the blocks $C_i$ had only small entries (and hence has small operator norm), then we could apply the exact argument from Example \ref{gepMotivation}. In the case that the $C_i$ typically have small and large diagonal entries, we will develop a method to use the small diagonal entries to break $S$ into a direct sum (in a rotated basis) of weighted shifts for which the arguments in Example \ref{gepMotivation} apply. However, the estimates will depend on the distribution of values.

\begin{example}\label{gep-Example}

We now illustrate this mechanism for constructing projections analogous to those from Example \ref{gepMotivation} in an example when the $C_i$ all are square matrices. This is the case addressed by Lemma \ref{proto-gep}. 
\begin{figure}[htp]     \centering
    \includegraphics[width=16cm]{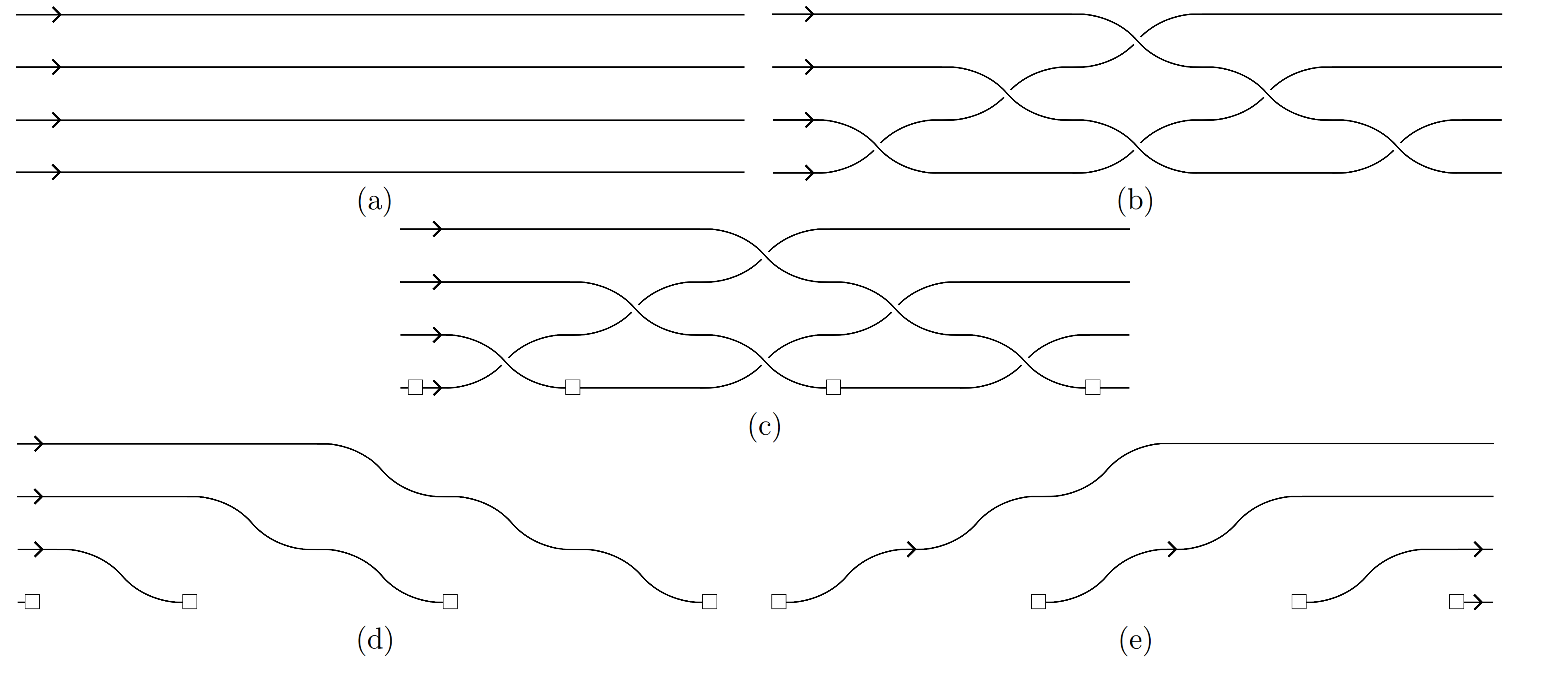}
    \caption{\label{GEP_Motivation1Rescaled}\dark
    Illustration of the gradual exchange process with $m=4$ for Example \ref{gep-Example}.}
\end{figure} 
In this example we focus on constructing only a single projection.
Let $A$ and $S$ be of the form of (\ref{blockMatrices}) where the $\alpha_i$ are strictly increasing real numbers and the identical matrix-valued weights $C_i$ of $S$ are
\[{\setlength{\arraycolsep}{1.5pt}
C_i = \bp 
c^4 & & & \\
& c^3 & & \\
& & c^2 & \\
& & & c^1  \\
\ep}\]
for some $c^r\geq 0$ and $c^1 \leq D$.  Note that the index $r$ is a superscript so that when the blocks $C_i$ are not identical as in Lemma \ref{proto-gep} then the diagonal entries of $C_i$ can be written with the similar notation: $c_i^r$. 

We now define weighted shift matrices $S_r = \ws(c^r)$ and will construct certain projections $F$ and $F^c$ for the direct sum of the $S_r$. We will later explain how $S$ can be seen as the direct sum of the $S_r$.

We describe the diagrams in Figure \ref{GEP_Motivation1Rescaled}.
Figure \ref{GEP_Motivation1Rescaled}(a) is weighted shift diagram for $S = \bigoplus_{r=1}^m S_r$ with $m = 4$ in the direct sum basis. In the diagram, the weighted shift diagram for $S_1$ is on the bottom of (a) and $S_4$ is illustrated on the top. Only a portion of the orbits is shown. For this example, we will apply our method within this window and outside of this window $S$ will not be changed.
Figure \ref{GEP_Motivation1Rescaled}(b) is an illustration of how we will apply the gradual exchange lemma. For the following discussion, please see Figure \ref{GEP_Size4_Columns} below for a description of what the ``columns'' are.

We first apply the gradual exchange lemma to $S_2, S_1$ over the span of $N_0+1$ vectors. We will have $N_0+1$ vectors in each orbit corresponding to where applications of the gradual exchange lemma occur in a column.

Later in the basis, we apply the gradual exchange lemma to $S_3, S_2$ over $N_0+1$ vectors. This is the second column of application(s) of the gradual exchange lemma. 
Later in the basis, we simultaneously (in the same column) apply the gradual exchange lemma to $S_4, S_3 $ and to $S_2, S_1$ in parallel. 

This is the end of the first stage. What we have done so far has changed the orbit of $S_1$ so that it ends up in the orbit of $S_m$ and the orbit of $S_m$ has finally been lowered to $S_{m-1}$. 
After the first stage, we continue to lower the orbits. We apply the gradual exchange lemma to $S_3, S_2$ in the next column. Then we apply it to $S_2, S_1$ in the last column.  

\begin{figure}[htp]     \centering
    \includegraphics[width=10cm]{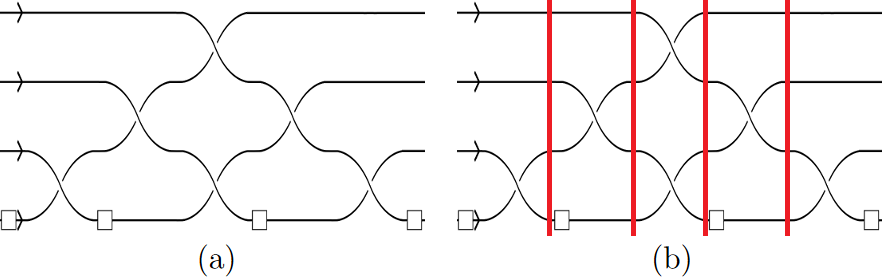}
    \caption{\label{GEP_Size4_Columns}\dark
    Part (a) of this Figure is equivalent to Figure \ref{GEP_Motivation1Rescaled}(c). Part (b) of this Figure has four vertical red bars inserted. The ``first column'' refers to the portion of the diagram to the left of the first bar. The second column refers to the portion between the first and second bars and so on.}
\end{figure}
In more generality (see Figure \ref{GEP_Motivation2Rescaled2}), the first stage has $m-1$ columns and the second stage has $m-2$ columns. In each column, the gradual exchange lemma is applied to pair(s) of weights shift operators in parallel. 
In the proof of Lemma \ref{proto-gep},  the column in which we apply the gradual exchange lemma is spanned by $\mathcal V_{3 + (N_0+1)(j-1)}$, $\dots$, $\mathcal V_{2 + (N_0+1)j}$.
\begin{figure}[htp]     \centering
    \includegraphics[width=16cm]{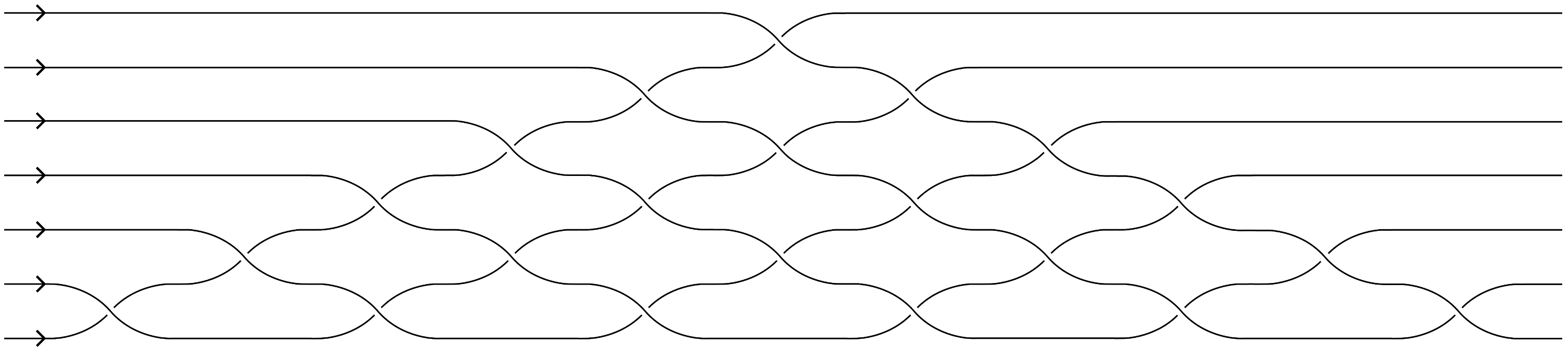}
    \caption{\label{GEP_Motivation2Rescaled2}\dark
    Illustration of the applications of the gradual exchange lemma as a part of the gradual exchange process for $m=7$ weighted shift matrices. \\
    When $1\leq j \leq m-1$,  we apply the gradual exchange lemma to the the pairs of weighted shift operators: $S_{j+1-e}, S_{j-e}$ in the $j$th column for all even $e \in [0, j)$ .
    When $m \leq j \leq 2m-3$,  we apply the gradual exchange lemma to the pairs of weighted shift operators: $S_{2m-j-1-e}, S_{2m-j-2-e}$ in the $j$th column for all even $e \in [0, 2m-j-2)$.
    }
\end{figure}

These applications of the gradual exchange lemma give the linear operator $\tilde{S}$, which is a direct sum of weighted shift matrices $\tilde{S}_r$ in a rotated basis. 
By following the orbit of the first basis vector of each of the direct summands $\tilde{S}_r$, we see that each of these orbits eventually lie in the orbit of $S_1$. 
The particular weaving done with the applications of the gradual exchange lemma was for this reason since we will assume that the weights of $S_1$ are small.

Choose a vector belonging to the portion of the orbit of $\tilde{S}_r$ that is in the orbit of $S_1$. 
We then break the orbit of $\tilde{S}_r$ into two orbits by replacing the weight of that vector with zero. This can be done by a perturbation of size at most $D = \|S_1\|$. We do this for each $r$ to  obtain $S'$.
Figure \ref{GEP_Motivation1Rescaled}(c) is an illustration of this.

Figure \ref{GEP_Motivation1Rescaled}(d) is an illustration of the orbit of the initial basis vector of each $S_r$ under $S'$. $S'$ acts as the direct sum of some weighted shift matrices which terminate within the window of \ref{GEP_Motivation1Rescaled}(a) that we began with. We define the projection $F$ to have range equaling the portion of the orbits illustrated in Figure \ref{GEP_Motivation1Rescaled}(d) that are within the window illustrated. 

Additionally, there are other weighted shift operators that form part of $S'$ that are illustrated in Figure \ref{GEP_Motivation1Rescaled}(e). The orbits of these operators begin within the window of \ref{GEP_Motivation1Rescaled}(a) that we began with and exit the window. The portion of the orbits illustrated in this window span the range of a projection that we call $F^c$. 

Observe a few key properties of $S', F,$ and $F^c$. First, $F$ is an invariant subspace of $S'$. When restricting $S'$ to $F$, we see that $S'$ has the structure of the direct sum of weighted shift operators. Notice that in the subspace corresponding to the window of the weighted shift diagrams, the range of $F^c$ is the orthogonal complement of the range of $F$. Moreover, although $F^c$ is not an invariant subspace of the entire domain of $S$, the image of  $R(F^c)$ under $S'$ is orthogonal to $R(F)$ and belongs to the span of $R(F)$ and the basis vectors of the weighted shift diagram that lie outside the window to the right. 
These properties will allow us to construct invariant subspaces when we apply the construction illustrated in this example later when forming various projections $F_i$ and $F_i^c$ for all windows as in Lemma \ref{gep}.

The estimates obtained will depend on the weights. The weights of consecutive $S_r$ contribute to the estimate through the gradual exchange lemma and the weights of $S_1$ contribute to the value of $\|S'-\tilde S\|$ when we break the orbits of $\tilde{S}_r$.
A key property of applying the gradual exchange lemma is that because the applications of the gradual exchange lemma are only applied to $S$ on orthogonal subspaces, the norms of perturbations do not add. Similarly, because the vectors in the orbit of $S_1$ whose weights of $\tilde{S}$ that we changed to zero were not affected by our application of the gradual exchange lemma, the norms of the perturbations of breaking up the orbits will not add either. We now estimate $\|S'-S\|$.

In the first column of applications of the gradual exchange lemma, we applied this lemma to $S_2, S_1$ incurring a perturbation of norm at most $|c^2-c^1| + \frac{\pi}{2N_0}\max(c^1, c^2)$. Next we applied the gradual exchange lemma to $S_3, S_2$, incurring an independent perturbation of norm at most $|c^3-c^2| + \frac{\pi}{2N_0}\max(c^2, c^3)$. Then we applied the gradual exchange lemma to $S_2, S_1$ and also $S_4, S_3$ in parallel, incurring independent perturbations of norm at most $|c^2-c^1| + \frac{\pi}{2N_0}\max(c^1, c^2)$ and $|c^4-c^3| + \frac{\pi}{2N_0}\max(c^3, c^4)$, respectively. Continuing this analysis, we observe that by applying the gradual exchange lemma in our construction of $\tilde{S}$ incurred a perturbation of norm at most 
\[G = \max_{1 \leq r\leq 3}\left(|c^{r+1}-c^r| + \frac{\pi}{2N_0}\max(c^{r+1}, c^r)\right).\]
Changing some of the weights to zero incurred an independent perturbation of norm $c^1 \leq D$. So,
\[\|S'-S\| \leq \max(G, D).\]

\vspace{0.1in}

We now return to the identification of $S$ as this direct sum of weighted shift matrices. We then describe the construction of $F$ in terms of basis vectors. 
Let the subspaces corresponding to the blocks $C_i$ be $\mathcal V_1, \dots, \mathcal V_n$. Write the standard basis vectors of $\C^{4n}$ as $e_1^4, e_1^3, e_1^2, e_1^1$, $\dots$, $e_n^4, e_n^3, e_n^2, e_n^1$ so that the subspace $\mathcal V_i$ is spanned by $e_i^4, e_i^3, e_i^2, e_i^1$.

Let $A_r = \diag(\alpha_1, \dots, \alpha_n)$ and $S_r = \ws(c^r, \dots, c^r)$ for $r = 1, \dots, 4$. 
By grouping the standard basis vectors of $\C^{4n}$ as $e_1^r, e_2^r, \dots, e_n^r$, 
we can express $A$ and $S$ as $A = \bigoplus_{r=1}^4 A_r$ and $S = \bigoplus_{r=1}^4 S_r$. 
In particular, the span of $e_1^r, e_2^r, \dots, e_n^r$ is invariant under $A$ and $S$ with
$Ae_i^r = a_i e_i^r$ and $S e_i^r = c^re_{i+1}^r$. This is the orbit of $e^r_1$ under $S$.

So, the formulation of $A$ and $S$ as block matrices of the form of Equation (\ref{blockMatrices}) with the same size is equivalent to expressing $A$ as a direct sum of the identical diagonal matrices $A_r$ and expressing $S$ as a direct sum of the weighted shift matrices $S_r$ by rearranging the direct sum basis. In the block matrix perspective,  
$e^r_i$ can be expressed as $0_4^{\oplus(i-1)}\oplus e_r \oplus 0_4^{\oplus(n-i)}$, where $0_4^{\oplus k}$ is the $k$-fold direct sum of the zero vector $0_4$ in $\C^4$.

After this set-up, we now state the required properties of $S', F,$ and $F^c$ as in the statement of Lemma \ref{proto-gep}.
Let $a, b \in \sigma(A)$ and $\alpha_1 \leq a < b \leq \alpha_n$. This specifies the window in which we focus.

We will require that the projections $F\leq E_{[a,b)}(A)$ and $F^c = E_{[a,b]}(A)-F$ satisfy $E_{\{a\}}(A) \leq F $, $R(F)$ is invariant under $S'$, and $S'$ maps $R(F^c)$ into $R(F^c) + R(E_{b+}(A))$, where $b+$ is the eigenvalue of $A$ that equals $\min \sigma(A) \cap (b, \infty)$, if it exists. If $\sigma(A) \cap (b, \infty)=\emptyset$, then $R(F^c)$ will just be an invariant subspace. These are conditions that we will use in Lemma 
\ref{proto-gep}.

\vspace{0.05in}

We will now describe the vectors spanning $F$. Note that our description of these vectors, some of which are obtained by many applications of the gradual exchange lemma, will not mention how negative signs are propagated in the sort of detail seen in Example \ref{GEL_Arrows}. We will instead use the statement of the gradual exchange lemma that we proved which takes care of the propagated negative signs after each application. Keeping track of the negative signs is not necessary to state what $F$ is, however it is necessary if we wanted to have an explicit description of the basis with respect to $S'$ breaks into a direct sum of weighted shift matrices with positive weights in order to apply Berg's construction in Theorem \ref{BergResult}.

So, we begin. The vectors 
\[e^4_1, e^3_1, e^2_1, e^1_1\] correspond to the first block because they form a basis for $\mathcal V_1$. Each $e^r_1$ corresponds to a point on each of the four orbits lying on a vertical line on the far left of Figure \ref{GEP_Size4_Columns}(b) to the left of the box at the bottom of this first column. Because we require $\mathcal V_1 \subset R(F)$, we include these vectors in our collection of spanning vectors of $R(F)$. For the sake of not perturbing the weights $c_1^r$ on the boundaries of this window, we need the subspace $\mathcal V_2$ to also be included:
\[e^4_2, e^3_2, e^2_2, e^1_2\]
since $Se^r_1 = c^re^r_2$.

When we continue our list of vectors, we drop the last vector to obtain
\[e^4_3, e^3_3, e^2_3, 0.\]
Now, these three vectors will also form a part of the basis for $R(F)$. Although 0 does not contribute to the span, we leave it there as a placeholder.
Because $Se^1_2 = c^1e^1_3$ and because we will set one of the weights $c^1$ equal to zero so that $S'e^1_2 = 0$, our dropping $e^1_3$ corresponds to a perturbation of $S$ of norm  $c^1 \leq D$ only on the the orbit of $S_1$. The box in the first column of \ref{GEP_Size4_Columns}(b) reflects that although $e^1_2$ belongs to the orbit of $S_1$ we made a weight equal to zero so that now $e^1_3$ is excluded from the orbit of $S'$.

We now apply the gradual exchange lemma to obtain orthonormal vectors $e_k'^2, e_k'^1$, orthogonal to all other vectors that we list, so that $e_{2+1}'^2=e_{2+1}^2$, $e_{2+(N_0+1)}'^2 = e_{2+(N_0+1)}^1$, $e_{2+1}'^1 = e_{2+1}^1$, $e_{2+(N_0+1)}'^1 = -e_{2+(N_0+1)}^2$.
Our list of vectors continues with (the first line is what we have listed above)
\[e^4_3, e^3_3, e'^2_3, 0,\]
\[e^4_4, e^3_4, e'^2_4, 0,\]
\[\dots\]
\[e^4_{2+(N_0+1)}, e^3_{2+(N_0+1)}, e'^2_{2+(N_0+1)}, 0,\]
which is
\[e^4_{2+(N_0+1)}, e^3_{2+(N_0+1)}, 0, e^1_{2+(N_0+1)}.\]
This application of the gradual exchange lemma happens in the first column of \ref{GEP_Size4_Columns}(b).

We then apply the gradual exchange lemma to obtain vectors $e_k'^3, e_k'^2$ so that $e_{3+(N_0+1)}'^3$ $=e_{3+(N_0+1)}^3$, $e_{2+2(N_0+1)}'^3 = e_{2+2(N_0+1)}^2$, $e_{3+(N_0+1)}'^2 = e_{3+(N_0+1)}^2$, $e_{2+2(N_0+1)}'^2 = -e_{2+2(N_0+1)}^3$.
Our list of vectors continues as follows. Note that we drop the lowest weight vector as well in the third step.
\[e^4_{3+(N_0+1)}, e'^3_{3+(N_0+1)}, 0, e^1_{3+(N_0+1)},\]
\[e^4_{4+(N_0+1)}, e'^3_{4+(N_0+1)}, 0, e^1_{4+(N_0+1)},\]
\[e^4_{5+(N_0+1)}, e'^3_{5+(N_0+1)}, 0, 0,\]
\[\dots\]
\[e^4_{2+2(N_0+1)}, e'^3_{2+2(N_0+1)}, 0, 0,\]
which is
\[e^4_{2+2(N_0+1)}, 0, e^2_{2+2(N_0+1)}, 0.\]
This application of the gradual exchange lemma happens in the second column of \ref{GEP_Size4_Columns}(b). The dropping a vector in the orbit of $S_1$ corresponds to the box in the second column of \ref{GEP_Size4_Columns}(b).

Now that there are not any consecutive non-zero vectors in our list of vectors, we apply the gradual exchange lemma twice to ``lower'' all the non-zero vectors. 
Now, we obtain vectors  $e_k'^2, e_k'^1$ so that $e_{3+2(N_0+1)}'^2=e_{3+2(N_0+1)}^2$, $e_{2+3(N_0+1)}'^2 = e_{2+3(N_0+1)}^1, e_{3+2(N_0+1)}'^1 = e_{3+2(N_0+1)}^1, e_{2+3(N_0+1)}'^1 = -e_{2+3(N_0+1)}^2$ as well as vectors
$e_k'^4, e_k'^3$ so that $e_{3+2(N_0+1)}'^4=e_{3+2(N_0+1)}^4$, $e_{2+3(N_0+1)}'^4 = e_{2+3(N_0+1)}^3$, $e_{3+2(N_0+1)}'^3 = e_{3+2(N_0+1)}^3$, $e_{2+3(N_0+1)}'^3 = -e_{2+3(N_0+1)}^4$.

Our list of vectors continues with
\[e'^4_{3+2(N_0+1)}, 0, e'^2_{3+2(N_0+1)}, 0,\]
\[e'^4_{4+2(N_0+1)}, 0, e'^2_{4+2(N_0+1)}, 0,\]
\[\dots\]
\[e'^4_{2+3(N_0+1)}, 0, e'^2_{2+3(N_0+1)}, 0,\] which is
\[0, e^3_{2+3(N_0+1)}, 0, e^1_{2+3(N_0+1)}.\]
These two applications of the gradual exchange lemma happen in the third column of \ref{GEP_Size4_Columns}(b).

Then we apply the gradual exchange lemma to obtain vectors  $e_k'^2, e_k'^3$ with the expected properties so that our list of vectors continues with
\[0, e'^3_{3+3(N_0+1)}, 0, e'^1_{3+3(N_0+1)},\]
\[0, e'^3_{4+3(N_0+1)}, 0, e'^1_{4+3(N_0+1)},\]
\[0, e'^3_{5+3(N_0+1)}, 0, 0,\]
\[\dots\]
\[0, e'^3_{2+4(N_0+1)}, 0, 0,\]
which is
\[0, 0, e^2_{2+4(N_0+1)}, 0.\]
This application of the gradual exchange lemma happens in the fourth column of \ref{GEP_Size4_Columns}(b). 
The dropping a vector in the orbit of $S_1$ corresponds to the box in the fourth column of \ref{GEP_Size4_Columns}(b).

Then we apply the gradual exchange lemma again to continue our list as
\[0, 0, e'^2_{3+4(N_0+1)}, 0.\]
\[\dots\]
\[0, 0, e'^2_{2+5(N_0+1)}, 0,\]
which is
\[0, 0, 0, e^1_{2+5(N_0+1)}.\]
We finally drop the last vector to obtain
\[0, 0, 0, 0\]
in the next block.
This corresponds to the box in the last column of \ref{GEP_Size4_Columns}(b).
We also  include another 
\[0, 0, 0, 0\]
for the last block
because the dropping of the vector corresponds to setting a weight to zero and we want to not change the first or last weights to facilitate calculating the change to the norm of the self-commutator by allowing us to restrict to each window.
This completes the construction of $F$ using $5(N_0+1)+4$ blocks.

Because $m = 4$, the constant $5$ (the number of columns) is the $2m-3$ that appears in the statement of Lemma \ref{proto-gep}. The $m-1$ comes from the first stage, consisting of the first three columns and $m-2$ comes from the second stage, consisting of the last two columns. 

If we follow the orbits of the vectors that were dropped, we obtain a basis for $R(F^c)$.  We will refer these vectors forming the orbits of $S'$ and the basis of $R(F)$ and $R(F^c)$ by $v_i^r$. \end{example}

Figure \ref{GEP_Motivation2Rescaled2} is an illustration of the method for $m = 7$ and Figure \ref{GEP_Motivation2a} illustrates breaking of the diagram into orbits that terminate and begin in this window in the construction of $F$ and $F^c$. 

\begin{figure}[htp]     \centering
    \includegraphics[width=16cm]{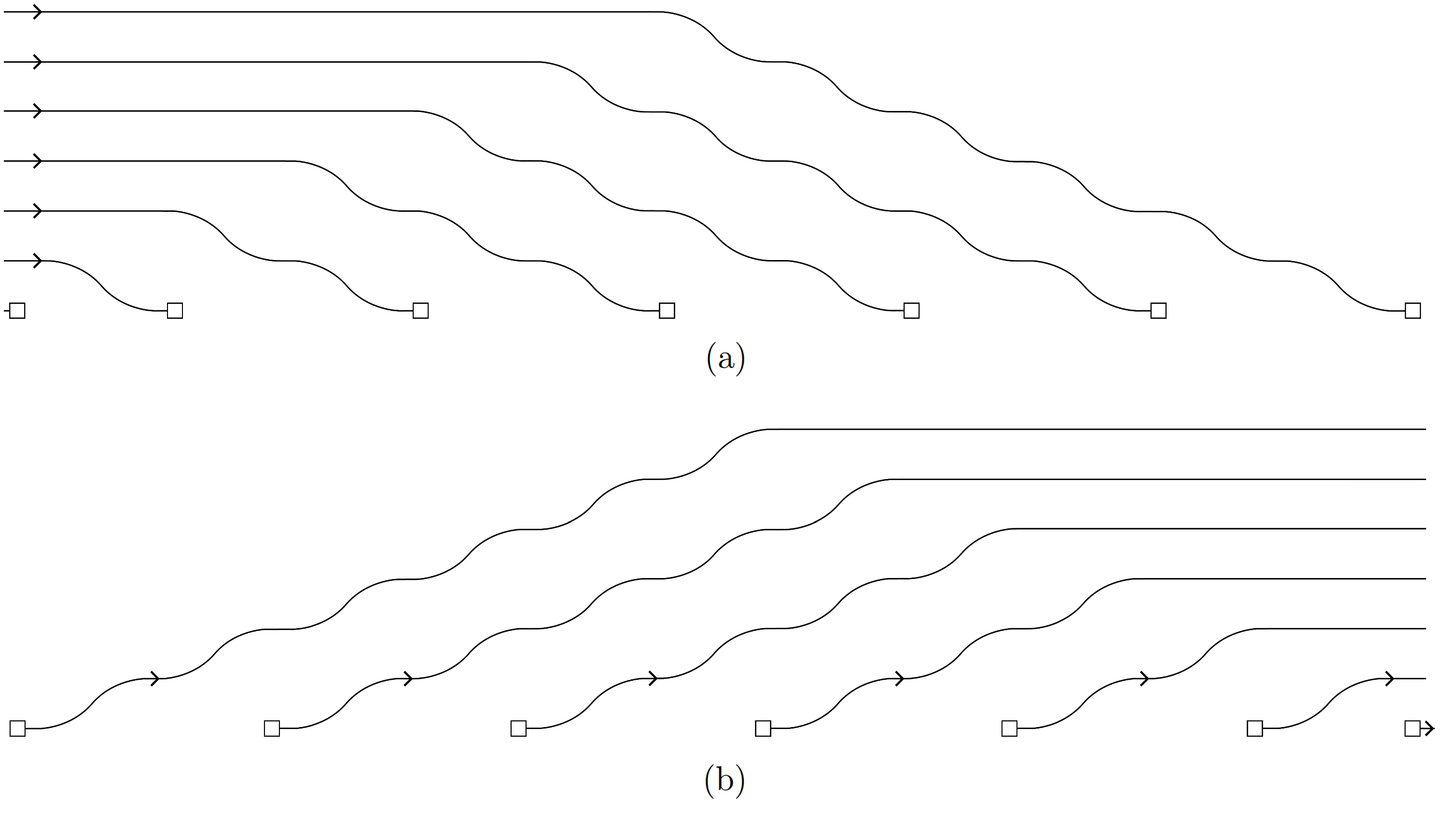}
    \caption{\label{GEP_Motivation2a}\dark
    Illustration of the decomposed weighted shifts as a continuation of   Figure \ref{GEP_Motivation2Rescaled2}. Compare to Figure \ref{GEP_Motivation1Rescaled}.
    The portion seen in this window of the orbits of the weighted shifts in (a) form a basis of the range of the projection $F$ and those of (b) in this window form a basis of the range of the projection $F^c$.}
\end{figure}

\vspace{0.1in}

The next three lemmas should be thought of as composing a single lemma but are stated independently to make the construction clearer.  Along the way we include more examples to illustrate the ideas of the proofs. The following is the gradual exchange process for constant-sized but not identical blocks.
\begin{lemma} \label{proto-gep}
Let $A_r=\diag(\alpha_i), S_r = \ws(c_i^r)$ with respect to some orthonormal basis of $M_{n}(\C)$ for $r = 1, \dots, m$. Suppose that the $\alpha_i$ are real and strictly increasing.  Define $A = \bigoplus_r A_r, S = \bigoplus_r S_r$. 

Let $a, b \in \sigma(A)$ with $a < b$. Let $N_0\geq 2$ be a natural number such that \\$n_0= \#\sigma(A)\cap [a,b]\geq (2m-3)(N_0+1)+4$ and $n_0 \geq 3$ in the case that $m = 1$.

Then there is a projection $F$ such that $E_{\{a\}}(A) \leq F \leq E_{[a,b)}(A)$ and a perturbation $S'$ of $S$ with $S'-S$ having support and range in $E_{(a,b]}(A)$ such that $S'$ is a direct sum of weighted shift matrices in a different eigenbasis of $A$, $F$ is an invariant subspace for $S'$, and 
\[\|S'-S\| \leq  \max(G_{[a,b]},\, D_{[a,b]})\]
\[\|\,[S'^\ast, S']\,\| \leq \max\left(\|[S^\ast, S]\|+T_{[a,b]},\, D_{[a,b]}^2\right)\]
where
\begin{align}
G_{[a,b]}&=\max_{1 \leq r \leq m-1}\max_{a \leq \alpha_i\leq b} \left(||c_{i}^{r+1}|-|c_i^r| | + \frac{\pi}{2N_0}\max(|c_i^r|,|c_i^{r+1}|)\right),
\nonumber
\\
D_{[a,b]}&=\max_{a\leq \alpha_i \leq b}|c_i^1|,
\\
T_{[a,b]}&=\frac1{N_0}\max_{1 \leq r \leq m-1}\max_{a \leq \alpha_i\leq b}||c_i^{r+1}|^2-|c_i^r|^2|.\nonumber
\end{align}
Additionally, define $F^c = E_{[a,b]}(A)-F$. Then $S'$ maps $R(F^c)$ into $R(F^c) + R(E_{\{b+\}}(A))$, where $b+ = \min \sigma(A) \cap (b,\infty)$ if $\sigma(A) \cap (b,\infty)\neq\emptyset$ or $b+=b$ otherwise.

If the $c_i^r$ are all real then there is an orthonormal basis of vectors $v_i^r$ that are real linear combinations of the given basis vectors such that $F$ and $F^c$ are each the span of a collection of these vectors and $S'$ is a direct sum of weighted shift matrices with real weights in this basis. The $v_i^r$ are also eigenvectors of $A$.

Note that if $m = 1$ then we use the convention that $G_{[a,b]}=T_{[a,b]}=0$.
\end{lemma}
\begin{remark}
We briefly explain the variable names. The term $G_{[a,b]}$ is the maximal error accrued due to an application of the gradual exchange lemma. The term $D_{[a,b]}$ bounds the weights that are set to zero and hence allow us to ``drop'' vectors from the range of $F$. The term $T_{[a,b]}$ is an additional ``term'' of the norm of the self-commutator that takes into account the interchange of orbits.

Define
\begin{align}
\varepsilon_{[a,b]} &= \max_{1\leq r\leq m-1}\max_{a \leq \alpha_i \leq b}||c_{i}^{r+1}|-|c_i^r||,
\nonumber
\\
R_{[a,b]}&=\frac{\pi}{2N_0}\max_{r}\max_{a \leq \alpha_i \leq b}|c_i^r|,
\nonumber
\end{align}
where
$\varepsilon_{[a,b]}$ is the maximal error due to the small difference in weights inherit in $S$  and $R_{[a,b]}$ is the maximal rotational error from proof of the gradual exchange lemma. It follows that $G_{[a,b]} \leq\varepsilon_{[a,b]}+ R_{[a,b]}$, although this inequality may be strict.
\end{remark}
\begin{proof}
We re-index the $\alpha_i$ in $[a,b]$ and choose
$n_0$ so that $\alpha_1 = a$ and $\alpha_{n_0}=b$.  Without loss of generality, we can assume that $c_i^r \geq 0$ by a change of basis as indicated in Example \ref{ws-basisChange}. Note that this change of basis is done only by multiplying the basis vectors by phases, so it does not affect the structure of $A$ and $S$ as direct sums of diagonal matrices and weighted shift matrices, respectively. The phases are $\pm1$ when the $c_i^r$ were real.

We first consider the trivial case of $m = 1$. With the relabeling given above, $S=\ws(c_{n_\ast}, \dots, c_{n^\ast})$ for $n_\ast\leq 1$ and $n^\ast \geq n_0 \geq 3$.
 We define $S'$ to equal $S$ except $c_2$ is replaced with zero.
Define $F = E_{\{a_1, a_2\}}(A)$ and $F^c = E_{[a_3,a_{n_0}]}(A)$. 

So, \[\|S'-S\| = c_2\leq \max(c_1, \dots, c_{n_0})= D_{[a,b]}.\]
Also,
\begin{align*}
\|\,[S'^\ast, S']\,\| &= \max(c_{n_\ast}^2, |c_{n_\ast+1}^2-c_{n_\ast}^2|,\dots, |c_1^2-c_0^2|, |0^2-c_1^2|, |c_3^2-0^2|, |c_4^2-c_3^2|, \dots,\\
&\;\;\;\;\;\;\;\;\;\;\;\;\;\;\;\;\;|c _{n^\ast}^2-c_{n^\ast-1}^2|, c _{n^\ast}^2) \\
&\leq \max(\|[S^\ast, S]\|, c
_1^2, c_3^2)\leq \max(\|[S^\ast, S]\|, \, D_{[a,b]}^2).
\end{align*}
The rest of the lemma then follows for this case.
 
\vspace{0.05in}

We now do the case that $m \geq 2$.
Let $e_r$ be the standard basis vectors of $\C^m$ and $e^r_i = 0_m^{\oplus(i-1)}\oplus e_r \oplus 0_m^{\oplus(n-i)}$ so that $Se_i^r = c_i^re_{i+1}^r$ for $i < n$ and $e_i^1, \dots, e_i^m$ form a basis for $\mathcal{V}_i = R(E_{a_i}(A))$.  Note that $(2m-3)(N_0+1)+3 < n_0$.

We now group the subspaces $\mathcal V_i$ as follows. The first grouping will consist of $\mathcal V_1, \mathcal V_2$. The second grouping will consist of $2m-3$ subgroupings of the $N_0+1$ subspaces
$\mathcal V_{3+(N_0+1)(j-1)},\dots, \mathcal V_{2+(N_0+1)j}$, $j = 1, \dots, 2m-3$. Let $\mathcal U_j = \bigoplus_{i=3+(N_0+1)(j-1)}^{2+(N_0+1)j}\mathcal V_i$. The third grouping is formed from the subspaces $\mathcal V_{3+(N_0+1)(2m-3)}, \dots, \mathcal V_{n_0}$. Note that the first and third groupings each consist of at least two of the subspaces $\mathcal V_i$.

For $j = 1, \dots, 2m-3$, we apply the gradual exchange lemma, Lemma \ref{GELws}, 
to pairs of weighted shift operators on the $N_0+1$ subspaces that compose  $\mathcal U_j$.
When $1\leq j \leq m-1$, the pairs of weighted shift operators that we apply the gradual exchange lemma to over the $N_0+1$ subspaces of $\mathcal U_j$ are $S_{j+1-e}, S_{j-e}$ for all even $e \in [0, j)$. 
When $m \leq j \leq 2m-3$, we apply the gradual exchange lemma over those latter $N_0+1$ subspaces of $\mathcal U_j$ to the operators $S_{2m-j-1-e}, S_{2m-j-2-e}$ for all even $e \in [0, 2m-j-2)$. 

Notice that the last pairs of operators in the first range are $S_{m-e}, S_{m-1-e}$ and the first pairs of operators in the second range are $S_{m-1-e}, S_{m-2-e}$. This means that if we have interchanged the orbits of some $S_{t}, S_{t-1}$ over $\mathcal U_{m-1}$ and $t-1>1$ then over $\mathcal U_{m}$ we will interchange of orbits of $S_{t-1}, S_{t-2}$. 
So, we will continue lowering the orbit of $S_t$ to $S_{t-1}$ then to $S_{t-2}$ across the value $j = m$. Because the indices $2m-j-2$ and $2m-j-1$ decrease by one for each increase of $j$ by one, we see that the orbit of $S_{t-2}$ will continue to be lowered if $t-2 > 1$. This will be useful later in the proof.

Let $\tilde{S}$ be the operator obtained from these modifications of $S$. Consider an orbit of $\tilde S$ while it is interchanging the orbits of two operators $S_t, S_{t-1}$ over the interval of indices $[i_0, i_1] = [3+(N_0+1)(j-1),2+(N_0+1)j]$. 
By Lemma \ref{GELws}(iii), when interchanging one orbit to the other, the weight at $i_0$ is the weight of $S$ corresponding to the former orbit and the weight at $i_1$ is the weight of $S$ corresponding to the latter orbit.
So, using the fact that the applications of the gradual exchange lemma are done independently over orthogonal subspaces, we see that with the arguments used in Example \ref{GEL_Arrows} that 
\begin{align}\label{GTildeEst}
\|\tilde{S} - S\| &\leq G_{[a,b]}
\\\label{TTildeEst}\|\,[\tilde S^\ast, \tilde S]\,\|&\leq \|[S^\ast, S]\|+T_{[a,b]}.
\end{align}

Now consider the orbit of $e_1^r$ under $\tilde{S}$. We know that $\tilde{S}$ is a direct sum of weighted shift operators whose orbits each start with a $e_1^r$. We claim that for each $r$, the orbit of $e_1^r$ under $\tilde S$ is eventually in the orbit of $S_1$. The following discussion is devoted to discussing this and finding particular weights $c_k^1$ in the orbit of $S_1$ that we will set equal to zero.

First, suppose that $r = 1$. In this case, we can just choose $k= 2$ just as in the case that $m=1$. The basis vector then belongs to the second subspace of the first grouping of subspaces. Suppose now that $2 \leq r \leq m$. Notice that the action of $S$ and $\tilde{S}$ on $e^r_1$ are identical on the $\mathcal U_j$ for $j < r-1$. 
When $j = r-1$, the gradual exchange lemma is applied to $ S_r, S_{r-1}$ over $\mathcal U_{r-1}$. 
So, the orbit of $e_1^r$ under $\tilde{S}$ moves from the orbit of $S_{r}$ to the orbit of $S_{r-1}$ by the beginning of $\mathcal U_{r}$. 
Then upon each application of the gradual exchange lemma, the orbit of $e_1^r$ under $\tilde{S}$ moves to $S_t$ with decreasing values of $t$. This clearly continues while both $j \leq m-1$ and the orbit is still not in the orbit of $S_1$. 

Observe that since $e^r_1$ begins to be lowered over $\mathcal U_{r-1}$ and $r-1$ orbits must be lowered, the orbit is finally lowered to the orbit of $S_1$ over $\mathcal U_{j_r}$ when $j_r = r-2+r-1 = 2r-3$.
Note that $S_m$ is the last orbit to begin to be lowered and, by construction, once it is lowered to $S_1$ over $\mathcal U_{2m-3}$, no more applications of the gradual exchange lemma are applied.
Note also that for all $r$ but $r=m$, the orbit of $e_1^r$ under $\tilde{S}$ will move back upward into the orbit of $S_t$ for some increasing values of $t$ as the result of the subsequent applications of the gradual exchange lemma.

In particular, if the orbit is moved from $S_2$ into $S_1$ over $\mathcal U_j$ then no application of the gradual exchange lemma is applied to $S_1$ over $\mathcal U_{j+1}$. 
More specifically, when $j$ is even, the gradual exchange lemma is not applied to $S_1$. 
So, for $j=2r-2$ with $2 \leq r\leq m-1$, we replace $c_i^1$ with zero for the second value of $i$ in $[3+(N_0+1)(j-1), 2+(N_0+1)j]$.
Denote this value of $i$ by $i_r=4+(N_0+1)(2r-3)$. So, we see that $e_1^r$ is annihilated by the $i_r$-th application of $\tilde{S}$ after this modification.

We extend this property to $r = 1, m$ by also replacing $c_2^1$ and $c_{3+(N_0+1)(2m-3)}^m$ with zero and defining with $i_1 = 2$ and $i_{m}=3+(N_0+1)(2m-3)$. 
So, all the $i_r$ are greater than $1$ and less than $n_0$. 

Let $S'$ be the operator gotten by applying these modifications to $\tilde{S}$. 
The estimate for $\|S'- S\|$ follows from  Equation (\ref{GTildeEst}) and the way that we set weights equal to zero that are bounded by $D_{[a,b]}$, just as in the case when $m = 1$.

Now, $S'$ is a direct sum of weighted shift operators in different $n$-dimensional orthogonal subspaces of $M_n(\C)^{\oplus m}$. Hence, we can obtain vectors ${v_i^r}$ due to the applications of the gradual exchange lemma with respect to the summands of $S'$ are weighted shift matrices.  So, $S'{v^r_i} = {c^r_i}'{v^r_{i+1}}$ for $i < n$ and $v^1_i, \dots, v^{m}_i$ form a basis for $\mathcal{V}_i = R(E_{a_i}(A))$, having the same span as $e^1_i,\dots, e^m_i$.
Define $F$ to be the span of \[v_i^r:\, 1 \leq r \leq m,\, 1\leq i\leq i_r.\] We see that $F$ is an orthogonal projection such that $R(F)$ is an invariant subspace for $S'$ and the other desired properties hold.
By this definition, we have that $F^c$ is the span of 
\[v_i^r:\, 1 \leq r \leq m,\, i_r < i \leq n_0.\]
We then see that $S'(R(F^c))$ is orthogonal to $R(F)$. So, because $S'$ maps $R(E_{[a,b]}(A))$ into $R(E_{[a,b+]}(A))$, the desired property of $F^c$ is obtained. 

When the $c_i^r$ are real, the desired properties follow from the use of real phases and the real coefficient properties from Lemma \ref{GELws}(i).

We now justify the estimate of the self-commutator of $S'$. Observe that replacing weights $d_{k}$ for $k$ in the index set $\mathcal I$ of a weighted shift matrix $T=\ws(d_i)$ with zero to create a weighted shift $T'$ will produce the estimate
\[\|[T'^\ast, T']\| \leq \max\left(\|[T^\ast, T]\|,\, \max_{k \in \mathcal I}\max(|d_{k-1}|^2, |d_{k+1}|^2)\right) \]
by the argument used in the case where $m=1$. When going from $\tilde S$ to $S'$ we are doing exactly this for the weighted shift operator summands of $\tilde S$.
By construction, the weights before and after the weight set to zero are weights of $S_1$ and hence are bounded by $D_{[a,b]}$. By this argument and Equation (\ref{TTildeEst}), we obtain the desired estimate for $\|[S'^\ast, S']\|$. 

When considering the support and range of $S'-S$, we see that the perturbations due to the gradual exchange lemma have support and range in the $\mathcal U_j$:
\begin{align}\label{GEParrow}
\mathcal V_{3+(N_0+1)(j-1)}\overset{}{\rightarrow}\mathcal V_{4+(N_0+1)(j-1)}\overset{\ast}{\rightarrow}\mathcal V_{5+(N_0+1)(j-1)}\overset{}{\rightarrow}\cdots\overset{}{\rightarrow} \mathcal V_{2+(N_0+1)j}\overset{}{\rightarrow}
\end{align}
where we have illustrated the action of either $S$ or $S'$ using the arrows between subspaces. Because $\mathcal V_1, \mathcal V_{n_0}$ are not included in the $\mathcal U_j$, 
the range and support of the perturbation $\tilde S - S$  is within the range of $E_{(a,b)}(A).$

The $\ast$ in (\ref{GEParrow}) indicates where the weights in the orbit of $S_1$ may be potentially set to zero.  The contribution to the perturbation $S'-S$ of setting the weight equal to zero within the $\mathcal U_j$ then has support and range in $E_{(a,b)}(A)$ as well.

Likewise, consider where the first and last weight is set equal to zero outside the $\mathcal U_j$ as indicated by the $\ast$'s:
\begin{align}\label{GEParrow2}
\mathcal V_1 \overset{}{\rightarrow} \mathcal V_2 \overset{\ast}{\rightarrow} \mathcal U_1\overset{}{\rightarrow}\cdots \overset{}{\rightarrow} \mathcal U_{2m-3}\overset{}{\rightarrow}\mathcal V_{3+(N_0+1)(2m-3)}\overset{\ast}{\rightarrow} \mathcal V_{4+(N_0+1)(2m-3)}\overset{}{\rightarrow}
\end{align}
We see that because $n_0 \geq 4+(N_0+1)(2m-3)$ that the support of $S'-S$ is in the range of $E_{(a,b)}(A)$ and the range of $S'-S$ is in the range of $E_{(a,b]}(A)$.
So, in total, the support and range of $S'-S$ is as stated in the lemma.
\end{proof}

Now, we illustrate the following result concerning when the blocks $C_i$  are not all  the same size. This is equivalent to the statement of the previous lemma when the matrices $A_r$ have spectrum growing in $r$. The idea is that if $\sigma(A_r)$ for some $r$ does not contain the entire spectrum of $A = \bigoplus_r A_r$ in the interval that we are looking at then $S_r$ already has an invariant subspace that we just include.

For example, suppose that $2 < n_1 < n_2 = \cdots = n_5$ and consider  $A_r = \diag(\alpha_1, \dots, \alpha_{n_r})$ where $\alpha_1 < \cdots < \alpha_{n_1} < \cdots < \alpha_{n_2}=\cdots=\alpha_{n_5}$ and $S_r = \ws(c_i^r)$ in $M_{n_r}(\C)$. Then $\sigma(A_1) = [\alpha_1, \alpha_{n_1}] \cap \sigma(A)\subsetneq \sigma(A_2) = \cdots = \sigma(A_5)=\sigma(A)$. In this example, $r_0 = 2$ as defined  in the lemma below. Figure \ref{GEP_Motivation3} illustrates the method that is used in the following lemma. 
\begin{figure}[htp]     \centering
    \includegraphics[width=16cm]{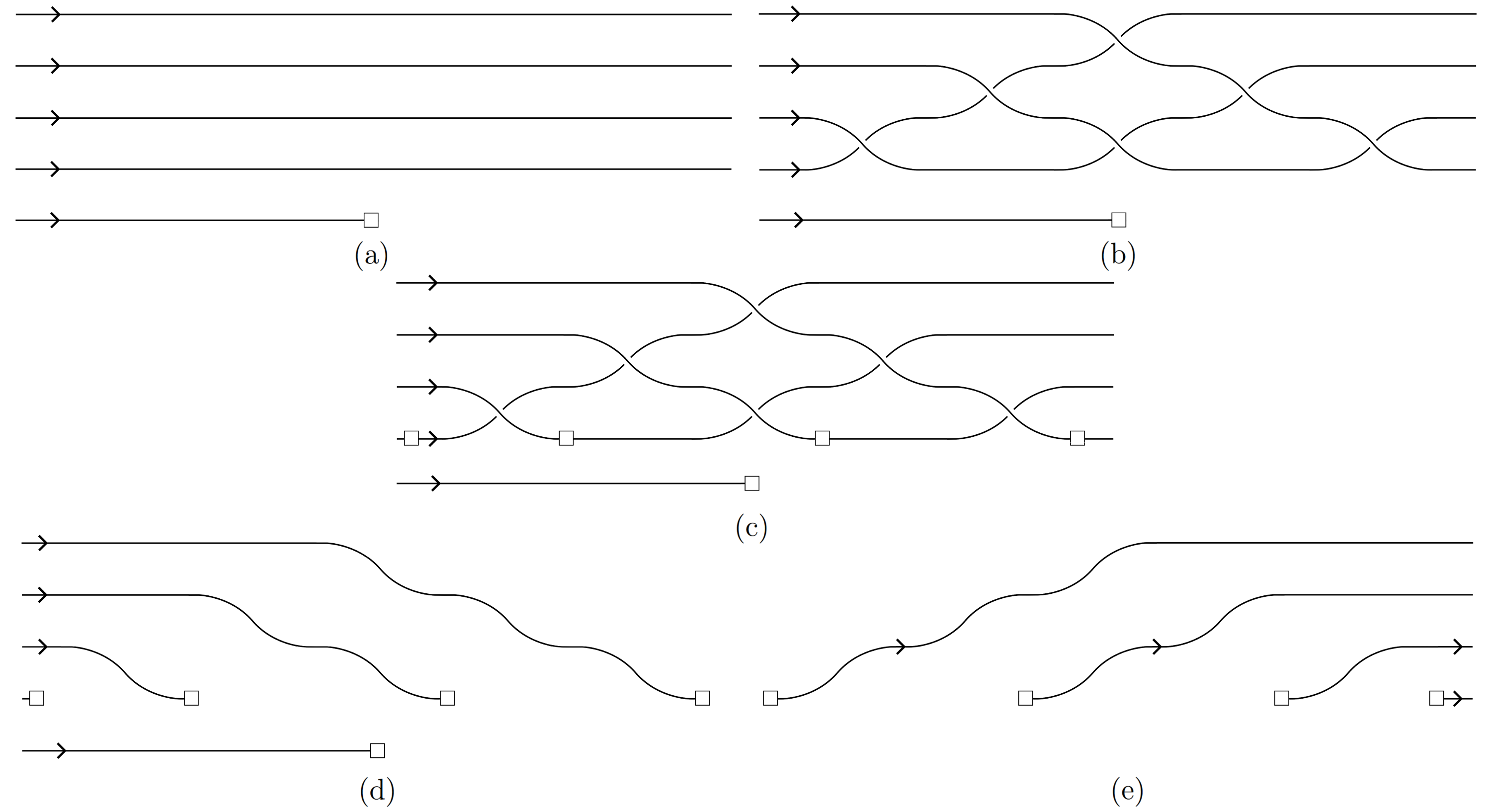}
    \caption{\label{GEP_Motivation3}\dark
    Illustration of appending another weighted shift operator whose orbit does not span the window to that of Example \ref{gep-Example}. Compare to Figure \ref{GEP_Motivation1Rescaled}.}
\end{figure}

Although more general forms of this lemma can be imagined, we only state what we will find useful in later sections. Note that if  $\underline{i}_r$ and $\overline{i}_r$ are constant, this lemma follows from the previous lemma. 
\begin{lemma} \label{proto-gep2}
Let $A_r=\diag(\alpha_i), S_r = \ws(c_i^r)$ with respect to some orthonormal basis of $M_{n_r}(\C)$ for $r = 1, \dots, m$ and $i= \underline{i}_r, \dots, \overline{i}_r$, where $[\underline{i}_r, \overline{i}_r]\subset [\underline{i}_{r+1}, \overline{i}_{r+1}]$.
Suppose that the $\alpha_i$ are real and strictly increasing.  Define $A = \bigoplus_r A_r, S = \bigoplus_r S_r$. 

Let $a, b \in \R$ with $a < b$. Let $N_0\geq 2$ be a natural number such that 
\[\#\sigma(A)\cap [a,b]\geq \max(3, (2m-3)(N_0+1)+4).\]
Let $\mathscr R_{I} = \{r: \sigma(A)\cap I\subset \sigma(A_r)\}$. Consequently, $\mathscr R_{[a,b]}$ is empty or equal to $r_0, r_0+1, \dots, m$ for some $r_0 \geq 1$. Let $a^\sigma = \min \sigma(A) \cap [a,b]$ and $b^\sigma = \max \sigma(A) \cap [a,b]$.

Then there is a projection $F$ such that $E_{\{a^\sigma\}}(A) \leq F \leq E_{[a^\sigma,b^\sigma)}(A)$ and a perturbation $S'$ of $S$ with $S'-S$ having support and range in $E_{(a^\sigma,b^\sigma]}(A)$ such that $S'$ is a direct sum of weighted shift matrices in a different eigenbasis of $A$, $F$ is an invariant subspace for $S'$, and 
\begin{align}\|S'-S\| &\leq  \max(G_{[a,b]}, D_{[a,b]}), \nonumber \\
\|\,[S'^\ast, S']\,\| &\leq \max\left(\|[S^\ast, S]\|+T_{[a,b]},\, D_{[a,b]}^2\right),\nonumber
\end{align}
where
\begin{align}
G_I &= \max_{\substack{r<m\\ r\in\mathscr R_{I} }}\max_{\alpha_i \in I} \left(||c_{i}^{r+1}|-|c_i^r| | + \frac{\pi}{2N_0}\max(|c_i^r|,|c_i^{r+1}|)\right),\nonumber
\\
D_{I}&=\max_{\alpha_i \in I}|c_i^{r_0}|\nonumber,\\
T_{I}&=\frac1{N_0}\max_{\substack{r<m\\ r\in\mathscr R_{I} }}\max_{\alpha_i \in I}||c_i^{r+1}|^2-|c_i^r|^2|.\nonumber
\end{align}
If $\mathscr R_{[a,b]}$ is empty then $S' = S$.
Additionally, define $F^c = E_{[a,b]}(A)-F$. Then $S'$ maps $R(F^c)$ into $R(F^c) + R(E_{\{b+\}}(A))$, where $b+ = \min \sigma(A) \cap (b,\infty)$ if $\sigma(A) \cap (b,\infty)\neq\emptyset$ or $b+=b^\sigma$ otherwise.

If the $c_i^r$ are all real then there is an orthonormal basis of vectors $v_i^r$ that are real linear combinations of the given basis vectors such that $F$ and $F^c$ are each the span of a collection of these vectors and $S'$ is a direct sum of weighted shift matrices with real weights in this basis. The $v_i^r$ are also eigenvectors of $A$.
\end{lemma}
\begin{proof}
Note that $n_r = \overline{i}_r - \underline{i}_r + 1$.

Let $G$ be the projection in $\mathcal M=\bigoplus_{r}M_{n_r}(\C)$ onto \[\mathcal G= \bigoplus_{r < r_0}0^{\oplus n_r}\oplus \bigoplus_{r \geq r_0}M_{n_r}(\C).\] Note that $\mathcal G$ is clearly an invariant subspace of $A$ and $S$. 
Now, we apply Lemma \ref{proto-gep} to $A_r, S_r$ for $r = r_0,\dots, m$ over $[a^\sigma,b^\sigma]$. This provides an operator $S'$ and projection $F$ on $\mathcal G$ with the desired properties with the exception that $F$ contains the projection onto $R(E_{\{a\}}(A))\cap \mathcal G$ and the estimate we have for $S'$ is 
\[\|(S'-S)G\| \leq \max(G_{[a^\sigma,b^\sigma]}, D_{[a^\sigma,b^\sigma]})\]
for $G_I$ and $D_I$ in the statement of the lemma.

We will identify $S'$, $F$, and $F^c$ with the operators on $\mathcal M$ that are gotten by trivially extending them to be zero on $\mathcal M \ominus \mathcal G$. However, the operator $S_{\mathcal M}$ and projections $F_{\mathcal M}$ and $F_{\mathcal M}^c$ that we construct for the first part of the statement of this lemma will in general be non-trivial extensions.

If $r_0 = 1$, then $G = I$ so the proof is complete.  So, suppose that $r_0 > 1$. 
Define 
\begin{align}\label{F_Mdef}
F_{\mathcal M} = F + \sum_{\substack{r < r_0\\ a^\sigma \in \sigma(A_r)}}E_{[a,b]}(A_r),\; F_{\mathcal M}^c = F^c + \sum_{\substack{r < r_0\\ a^\sigma \not\in \sigma(A_r)}}E_{[a,b]}(A_r).
\end{align}
Note that $F_{\mathcal M} - F$ and $F_{\mathcal M}^c - F^c$ are both projections into $R(I-G)$.

Recall the following basic property of $A_r = \diag(\alpha_i)$ and $S_r=\ws(c^r_i)$. If $v_i \in R(E_{\alpha_i}(A_r))$ then $S_r^k v_i \in R(E_{\alpha_{i+k}}(A_r))$. The following statements about $E_{[a,b]}(A_r)$ are then straightforward consequences of the assumptions on the $A_r$.
For $r < r_0$, there is an $\alpha \in [a,b]\cap \sigma(A)$ such that $\alpha \not \in \sigma(A_r)$. Because $\sigma(A_r) = \{\alpha_i: i \in [\underline{i}_r, \overline{i}_r]\}$ and $\sigma(A) = \{\alpha_i: i \in [\underline{i}_m, \overline{i}_m]\}$, it is not possible that $\sigma(A_r)$ contains both $a^\sigma$ and $b^\sigma$.

For each $r < r_0$ such that $a^\sigma \in \sigma(A_r)$, since $b^\sigma \not \in \sigma(A_r)$, we see that there is a $b_r \in [a^\sigma,b^\sigma)$ such that $[a,b] \cap \sigma(A_r) = [a^\sigma,b_r]$. Consequently, $R(E_{[a,b]}(A_r)) = R(E_{[a,b_r]}(A_r))$ is invariant under $S_r$.  
Likewise, consider $r < r_0$ such that $a^\sigma \not\in \sigma(A_r)$. If $[a,b]\cap\sigma(A_r)= \emptyset$, then $E_{[a,b]}(A_r) = 0$. Otherwise, there is an $a_r \in (a^\sigma, b^\sigma]$ such that $[a,b]\cap \sigma(A_r) \subset [a_r, b^\sigma]$. 
So, we see that $R(E_{[a,b]}(A_r))=R(E_{[a_r, b^\sigma]}(A_r))$ is mapped into $R(E_{[a_r, b+]}(A_r))$ by $S_r$.  
So, we obtain $E_{\{a^\sigma\}}(A) \leq F_{\mathcal M} \leq E_{[a^\sigma, b^\sigma)}(A)$ and  $F_{\mathcal M}^c = E_{[a^\sigma, b^\sigma]}(A)-F_{\mathcal M}$. 

We now extend $S'$ from $\mathcal G$ to $S'_{\mathcal M} = S'G + S(1-G)$ on $\mathcal M$. We then have $\|S'_{\mathcal M} - S\| = \|(S'-S)G\|$ with the above estimate. The estimate for the self-commutator of $S'$ holds similarly.  By the discussion above, $F_{\mathcal M}$ is invariant under $S'_{\mathcal M}$.

Therefore the desired property for $F^c_{\mathcal M}$ follows from that of $F^c$ from Lemma \ref{proto-gep} and each summand $E_{[a,b]}(A_r)$ in the definition of $F^c_{\mathcal M}$.
 
\end{proof}

\begin{remark}
We can instead assume that the spectrum of $A$ lies on a nice simple curve homeomorphic to an interval in 
$\R$. For instance, instead of increasing real numbers on a line, the $\alpha_i$ could be complex numbers on the unit circle with increasing argument. In this case, $A$ would be unitary and the $S_r$ could be either unilateral or bilateral weighted shifts. There are other generalizations possible.
\end{remark}

We now give an example of the construction of the following lemma.
\begin{figure}[htp]     \centering
    \includegraphics[width=16.5cm]{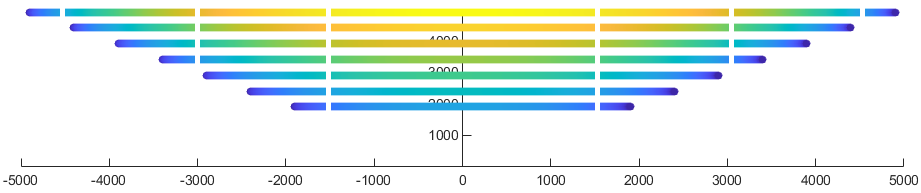}
    \caption{\label{WSheatmap_GEP_Motivation4}\dark
    Illustration of the weights of $S$ in Example \ref{S''Motivation}. The vertical gaps in the graph are shown to illustrate the windows in which we apply the gradual exchange method.}
\end{figure}
\begin{figure}[htp]     \centering
    \includegraphics[width=16cm]{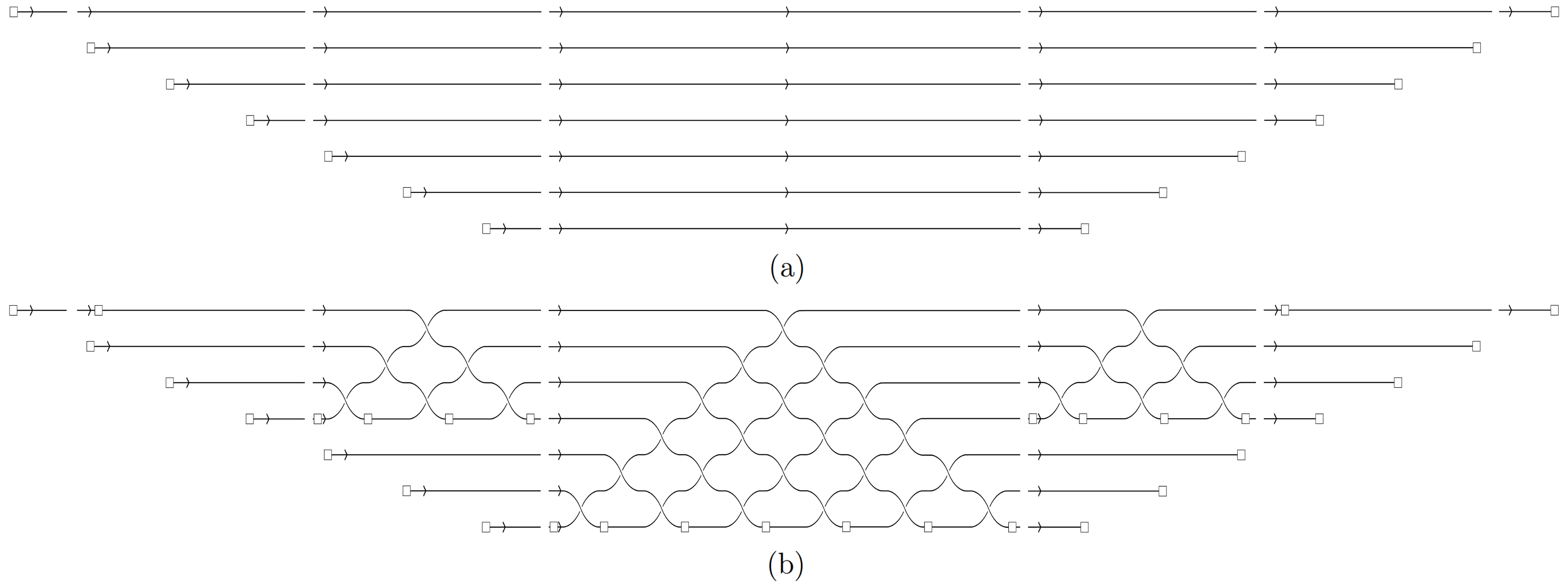}
    \caption{\label{GEP_Motivation4a}\dark
    Illustration of the applications of the gradual exchange lemma during the construction of $S'$ in Example \ref{S''Motivation}.}
\end{figure}
\begin{example}\label{S''Motivation}
Here we illustrate the construction of $S'$ and the $F_i, F_i^c$. Consider 
\[A = \frac1{4900}\left(S^{1900}\oplus S^{2400} \oplus \cdots \oplus S^{4900}(\sigma_3)\right),\; S = \frac1{4900}\left(S^{1900}\oplus S^{2400} \oplus \cdots \oplus S^{4900}(\sigma_+)\right).\]
A weighted shift diagram for $S$ is provided in Figure \ref{WSheatmap_GEP_Motivation4}. Note that the vertical gaps in the orbits are included to illustrate the windows that we deal with using the prior lemma and not that the orbits terminate.

\begin{figure}[htp]     \centering
    \includegraphics[width=16cm]{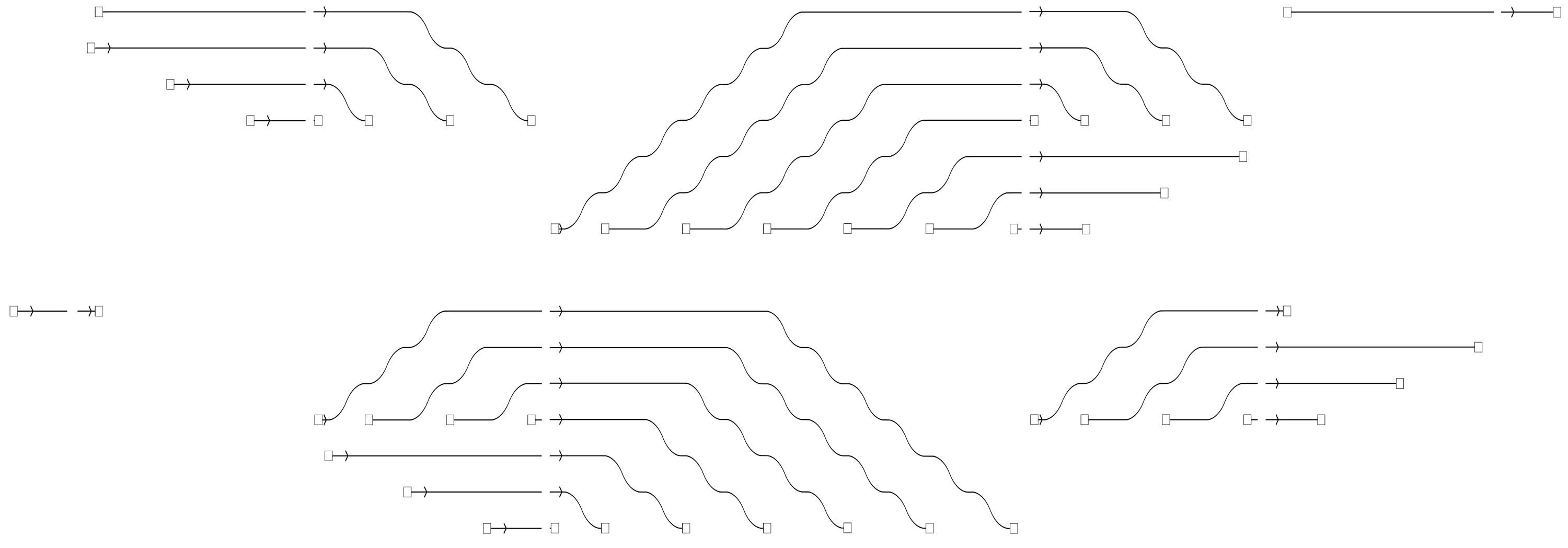}
    \caption{\label{GEP_Motivation4b}\dark
    Illustration of decomposed weighted shift operators of $S'$ in Example \ref{S''Motivation}.}
\end{figure}
\begin{figure}[htp]     \centering
    \includegraphics[width=16cm]{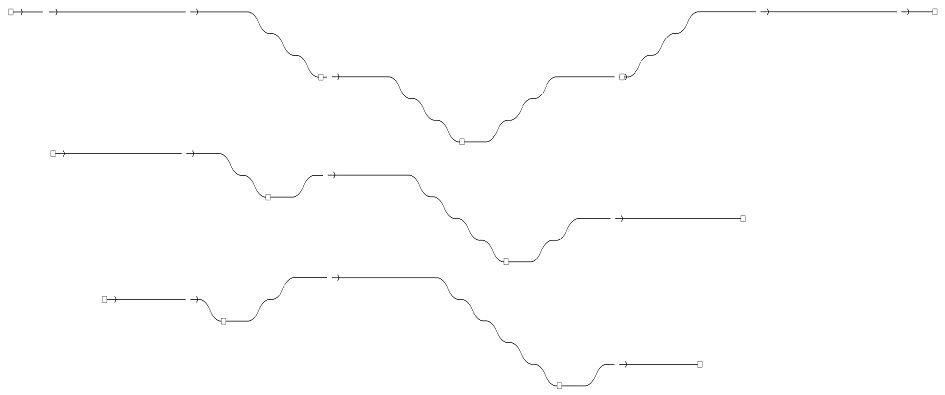}
    \caption{\label{GEP_Motivation4bOrbits}\dark
    Illustration of three orbits of $\tilde{S}$ in Example \ref{WSheatmap_GEP_Motivation4}(b).}
\end{figure}
\begin{figure}[htp]     \centering
    \includegraphics[width=16.2cm]{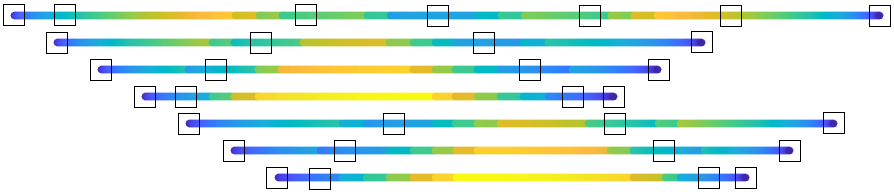}
    \caption{\label{WSheatmap_GEP_MixDecomp}\dark
    Illustration of decomposed weighted shift operators of $S'$ in Example \ref{S''Motivation}. Marked with boxes are the weights that are dropped in the construction detailed above.\\ Note that generating the colors was done using a different version of the gradual exchange lemma that does not continuously change the values of weights between orbits. }
\end{figure}

\begin{figure}[htp] \centering
    \includegraphics[width=16.2cm]{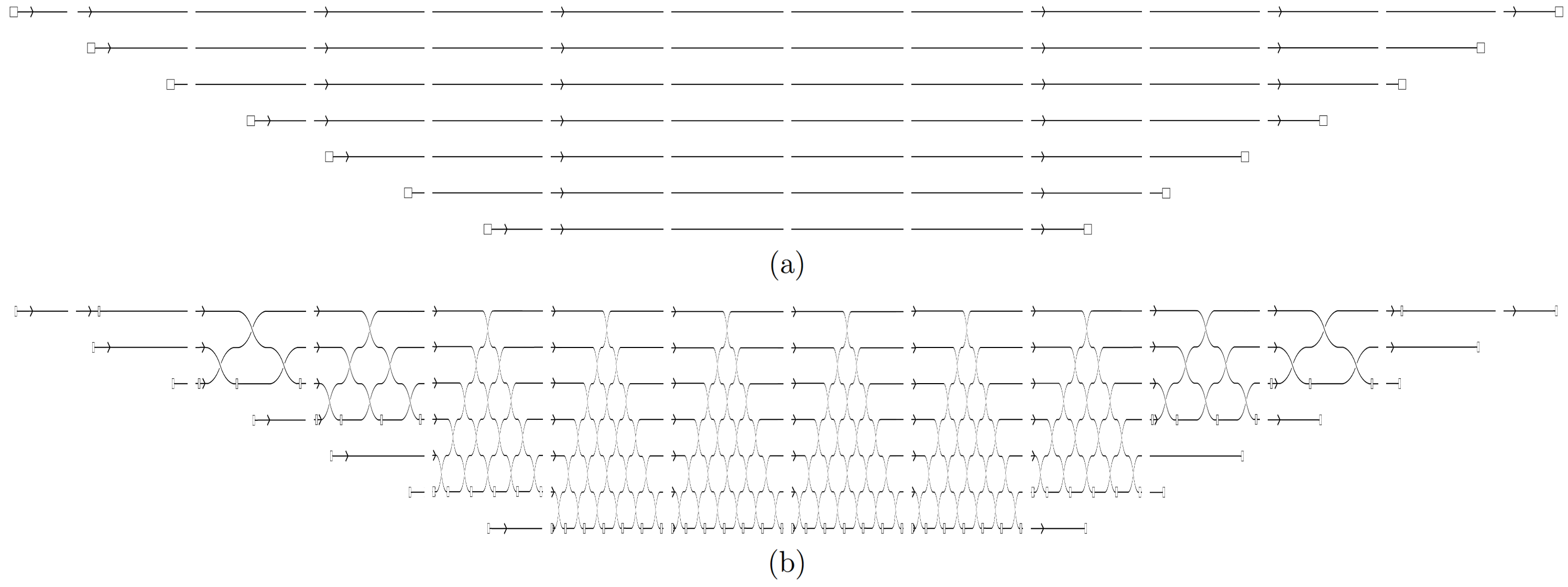}
    \caption{\label{GEP_Motivation_MoreOsc}\dark
    Illustration of applying gradual exchange process with a smaller window size.}
\end{figure}
\begin{figure}[htp] \centering
    \includegraphics[width=16.2cm]{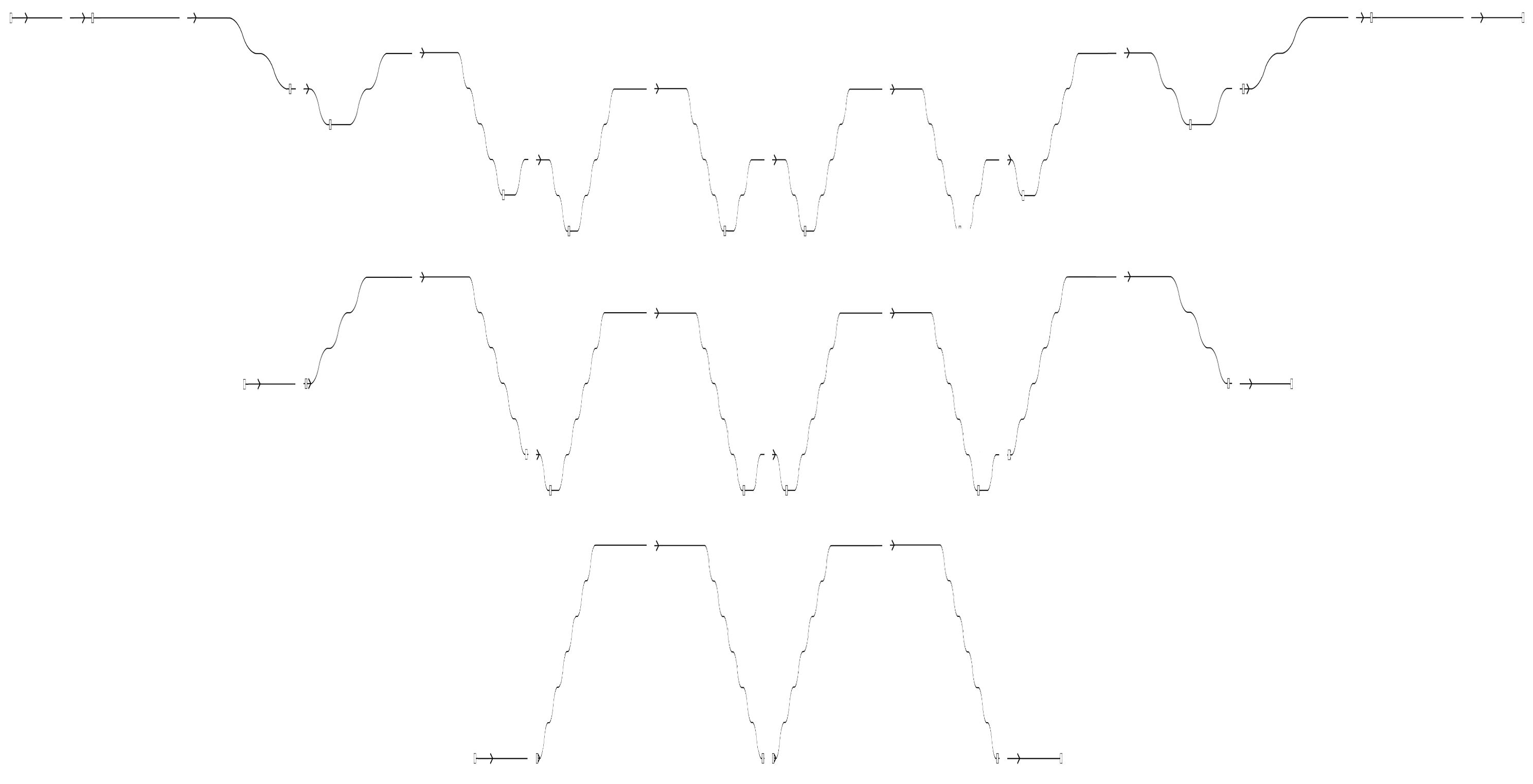}
    \caption{\label{GEP_Motivation_MoreOscOrbits}\dark Illustration of several of the orbits of $\tilde{S}$.}
\end{figure}
\begin{figure}[htp] \centering
    \includegraphics[width=16.2cm]{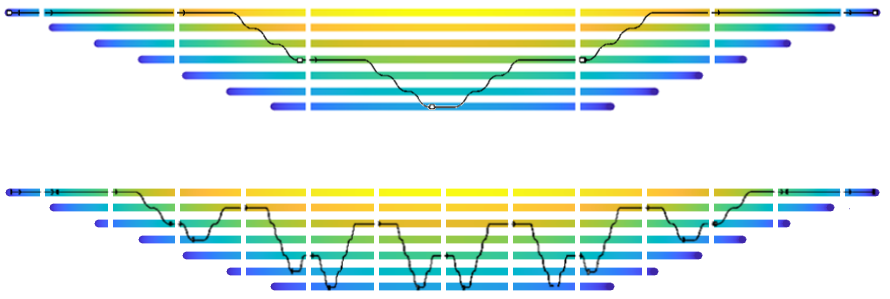}
    \caption{\label{GEP_MotivationOrbitsColor}\dark Illustration of several of the orbits of $\tilde{S}$ for different window sizes superimposed upon the colorbar graph of the values of the weights.}
\end{figure}
Figure \ref{GEP_Motivation4a}(a) is a depiction of $S$. 
Figure \ref{GEP_Motivation4a}(b) depicts the gradual exchange process that we developed earlier in each window. For each window, we explored earlier that $S'$ is a direct sum (in a rotated basis) of weighted shift operators whose orbits are broken in the window. Because we do this in each window, we can piece together these orbits.
Figures \ref{GEP_Motivation4b}(a) and \ref{GEP_Motivation4b}(b) illustrate these orbits. We then use these orbits to construct projections $E_j$ so that $A'$ has spectral projections $E_j$. Because each orbit belongs to at most two consecutive windows, $A'$ will be approximately equal to $A$ if the window length is small. 

For each orbit, we construct a nearby normal using Theorem \ref{BergResult}. Then putting these normals together gives $S''$.

\end{example}

We repeat the notation from the previous lemma in the statement of the next lemma. This result completes the construction of nearby commuting matrices using the gradual exchange process. The use of projections to construct nearby commuting matrices is motivated by the constructions in \cite{Hastings} and \cite{Davidson}.
\begin{lemma}\label{gep}
Let $A_r=\diag(\alpha_i), S_r = \ws(c_i^r)$ with respect to some orthonormal basis of $M_{n_r}(\C)$ for $r = 1, \dots, m$ and $i= \underline{i}_r, \dots, \overline{i}_r$, where $[\underline{i}_r, \overline{i}_r]\subset [\underline{i}_{r+1}, \overline{i}_{r+1}]$. 
Suppose that the $\alpha_i$ are real and strictly increasing.  Define $A = \bigoplus_r A_r, S = \bigoplus_r S_r$.  Let $\mathscr R_{I} = \{r: \sigma(A)\cap I\subset \sigma(A_r)\}$. Consequently, $\mathscr R_{I}$ is empty or equal to $r_0, r_0+1, \dots, m$ for some $r_0 = r_0(I) \geq 1$ which may depend on $I$. 

Let $a_k \in \R$, $a_1 < a_2 < \dots < a_{n_0}$, $I_k= [a_k, a_{k+1})$ for $k+1 < n_0$ and $I_{n_0-1} = [a_{n_0-1}, a_{n_0}]$, satisfying $\sigma(A) \subset \bigcup_k I_k$.
Let $m_k = m+1-r_0(I_k) \leq m$ and let $N_{I_k} = N_k$ be natural numbers such that
\begin{align}\label{minspectrum}
\#\sigma(A)\cap I_k\geq \max\left(3, (2m_k-3)(N_k+1)+4\right).
\end{align}
Let
\begin{align}
G_I &= \max_{\substack{r<m\\ r\in\mathscr R_{I} }}\max_{\alpha_i \in I} \left(||c_{i}^{r+1}|-|c_i^r|| + \frac{\pi}{2N_I}\max(|c_i^r|,|c_i^{r+1}|)\right)
\\
D_{I}&=\max_{\alpha_i \in I}|c_i^{r_0}|\\
T_{I}&=\frac1{N_I}\max_{\substack{r<m\\ r\in\mathscr R_{I} }}\max_{\alpha_i \in I}||c_i^{r+1}|^2-|c_i^r|^2|.
\end{align}

Then there is a self-adjoint matrix $A'$ commuting with a matrix $S'$ that is a direct sum of weighted shift matrices in an eigenbasis of $A'$ such that 
\begin{align}\|A' - A\| &\leq \max_k\diam I_k,\\
\|S'-S\| &\leq \max_k\max(G_{I_k}, D_{I_k}),\noindent
\\
\|\,[S'^\ast, S']\,\| &\leq \max_k \max\left(\|[S^\ast, S]\|+T_{I_k},\, D_{I_k}^2\right).
\end{align}
Moreover, there is a normal $S''$ that is a direct sum of weighted shift matrices in an eigenbasis of $A'$ such that
\begin{align}
 \|S''-S'\| \leq C_\alpha \|S\|^{1-2\alpha}\|\,[S'^\ast, S']\,\|^\alpha \label{S''ineq}
\end{align}
where $\alpha, C_\alpha > 0$ are constants such that a nearby normal matrix can be obtained by Theorem \ref{BergResult}.

If the $c_i^r$ are real then using $\alpha = 1/3, C_{1/3} = 5.3308$ allows $S''$ to be real.
Moreover, there is a real change of basis that makes $S''$ (and also $S'$) a direct sum of weighted shift matrices with real weights.
\end{lemma}
\begin{remark}
If we estimate
\begin{align}
\varepsilon_{I} &= \max_{\substack{r<m\\ r\in\mathscr R_{I} }}\max_{\alpha_i \in I}||c_{i}^{r+1}|-|c_i^r| |,\nonumber\\
R_{I}&=\frac{\pi}{2N}\max_{r \in \mathscr R_{I}}\max_{\alpha_i \in I}|c_i^r|\nonumber
\end{align}
separately then we obtain the bounds for $G_I$: $\max(\varepsilon_{I},R_{I}) \leq G_I \leq \varepsilon_{I} + R_{I}$.
\end{remark}
\begin{proof}
\underline{Construction of and estimates for $A'$ and $S'$}: Let $a^\sigma_k = \min \sigma(A)\cap [a_k, a_{k+1})$ and $b^\sigma_k = \max \sigma(A)\cap [a_k, a_{k+1})$. 
Let $F_k$ be the projection gotten by applying the construction in Lemma \ref{proto-gep2} for $[a^\sigma_k,b^\sigma_k]$, let $S_k'$ be the constructed perturbation of $S$, and $F_k^c = E_{[a^\sigma_k,b^\sigma_k]}(A) - F_k$. 
Note that $E_{\{b^\sigma_k\}}(A)\leq F_k^c \leq E_{(a^\sigma_k, b^\sigma_k]}(A)$. 
Define \[S' = S+\sum_k (S'_k - S).\]
The definition that we give here for $S'$ is the same as applying all these perturbations from the previous lemma in each window separately. Because the perturbations $S'_k-S$ are supported on and have range in the orthogonal subspaces $R(E_{I_k}(A))$, we obtain the desired estimate for $\|S'-S\|$.

Consider the orthogonal projections $E_k$ defined to be the \[F_1, F_1^c + F_2, \dots, F_{k}^c+F_{k+1}, \dots, F_{n_0-1}^c+F_{n_0}, F_{n_0}^c.\]   
Because the $F_k$ are invariant under $S'$ and $S'$ maps $R(E_{[a_k,b_k]}(A))$ into $R(E_{[a_k,a^\sigma_{k+1}]}(A))$, 
we see that $S'$ maps $R(F_{k}^c)$ into $R(F_{k}^c)+R(F_{k+1})$. 
Hence, the projections $E_k$ commute with $S'$. 
Note that $E_k\leq E_{[a_{k-1}, a_{k+1}]}(A)$ if $a_0$ is defined to be $a_1$ and $a_{n_0+1}$ is defined to be $a_{n_0}$. So, letting $A' = \sum_k a_{k}E_k$, we see that $[S',A']=0$ and $\|A'-A\|\leq \max_k (a_{k+1}-a_k)$. 

\vspace{0.05in}

\underline{Construction of and estimates for $S''$}: We now take advantage of the structure of $S'$ through the operators $S'_k$, which were called  $S'_{\mathcal M}$ in the proof of Lemma \ref{proto-gep2}. Please recall the construction of what was called $S'$ in Lemma \ref{proto-gep}, in particular the statement about the support and range of $S'-S$ illustrated in Equations (\ref{GEParrow}) and (\ref{GEParrow2}). These contribute to the construction of each $S'_k$. 

We know that $S'$ is a direct sum of weighted shift operators. Because the construction of $S'_k$ in each window did not change the weights of the weighted shifts on the boundaries, we see that the differences of the squares of the $S'$ weights between windows are the same as those of $S$ between windows. Within windows, the differences of squares of $S'$ weights are bounded by the estimates for the self-commutator of the $S_k'$ in Lemma \ref{proto-gep2}. So, the desired estimate for the self-commutator of $S'$ holds.

Because $S'$ commutes with $A'$, we can view the orbits of $S'$ as lying within the eigenspaces of $A'$. We then apply Theorem \ref{BergResult} to each such weighted shift orbit to obtain $S''$. If the $c_i^r$ are real then the additional structure follows from that of Lemma \ref{proto-gep2}.

\end{proof}

\begin{remark}
We now discuss the utility of the estimates gotten in this construction. 

We first discuss the term $D_I$.
Under some mild conditions, we need the singular values $\min_{r}\min_{\alpha_i \in I}|c_i^r|$ to be small in order for there to exist structured nearby commuting matrices by a generalization of Voiculescu's argument in \cite{Voiculescu}. This suggests that the estimate of $D_I = \max_{\alpha_i\in I}|c_i^{r_0}|$ might be small for situations where we want to construct nearby commuting matrices. 

The construction in Lemma \ref{proto-gep} strictly speaking does not make use of the fact that all $|c_i^{r_0}|$ are small for $\alpha_i \in I$ since only $m$ weights are set equal to zero in the construction of the invariant subspace. A different choice of which weights to set equal to zero based on the particular problem at hand might be able to improve this estimate when the values of $|c_i^r|$ vary rapidly in $i$. However, if each $S_r$ is almost normal then we expect such variation to be controlled by the self-commutator of $S$. 

We now discuss the term $G_I$. This term is a consequence of the application of the gradual exchange lemma to consecutive weighted shift operators $S_t, S_{t-1}$. 
Based on the details of this construction, the term $G_I$ can be changed by reordering the weighted shift operators $S_{r_1}$, $S_{r_2}$ in the direct sum given that $A_{r_1}=A_{r_2}$. In our application to Ogata's theorem in the next section, the weights $c_i^r$ will be increasing in $r$ so the natural ordering based on the spin of the representations is optimal.

The only contribution to $G_I$ that depends explicitly on $A$ is the appearance of the $N_I$ in the term corresponding to $R_I$. In applications, we will choose the points $a_i$ first so that then $N_I$ is chosen to be as large as possible. 
There is a trade-off between how small the spacing of the $a_i$ can be and how large $N_I$ can be. The spacing of the $a_i$ may directly affect all the terms $\varepsilon_I, R_I, G_I, D_I$ while the size of $N_I$ only directly affects $R_I$.

Because we assume that $[A,S]$ is small, we know that \[|\alpha_{i+1}-\alpha_i||c_i^r| \leq \|[A,S]\|\]
is small. Assuming that the norm of $S$ on $E_{I}(A)$ is of order $1$, we know that $\max_{\alpha_i \in I} |c_i^r|$ is bounded and so $|\alpha_{i+1}-\alpha_i|$ is at most a constant multiple of $\|[A,S]\|$. 
So, we choose the $a_i$ so that $\diam I_k$ is much larger than the spacing of the eigenvalues of $A$ and hence $N_I$ is large. Exactly how large $N_I$ will be will depend on the situation, but we will want balance the size of the various components of the estimate to obtain the optimal result.

We now discuss the term $T_I$. The norm of the self-commutator of $S$, $\|[S^\ast, S]\|$, reflects the sizes of the differences of the squares of the absolute values of the weights of $S$ along individual orbits.  When applying the gradual exchange lemma, we then need to take into account that the weights of $S_{t}, S_{t-1}$ are blended together. The term $T_I$ reflects the size of the differences of the squares of the absolute values of the weights of $S$ between the consecutive orbits of $S_{t}, S_{t-1}$, reduced by the factor $N_I^{-1}$ due to how many vectors we have to smooth out the weights over. 
So, we expect that if the weights of the weights shifts $S_r$ do not vary much in $r$ then $T_I$ should not be too large. 
\end{remark}
\begin{remark}\label{refine}
As discussed previously, given any collection of $A_r, S_r$, we can refine the direct sum over all $r$ by partitioning the set of possible values of $r$ then apply this lemma to each partition of direct summands separately. 

An example of why one might want to do this is that it is easily possible that $m$ is comparable to  (or even larger than) $\# \sigma(A)$. In this case, $N_I$ cannot be large so the estimate of $R_I$ is not small. 
Conversely, making the refinements too sparse  conversely may increase the size of $\varepsilon_I$ and $T_I$.

For instance, take any non-trivial example of $A, S$ and repeatedly form direct sums with themselves. Having repeated summands only makes the estimate for $\|S'-S\|$ worse. This is because none of the estimates from the lemma change if the repeated summands are listed together in the lemma except that $N_I$ necessarily must decrease due to the increase of $m$.

This sort of difficulty is relevant for our application to Ogata's theorem.
In fact, it is on its face impossible to use this result without refinement for Ogata's theorem as in the next section due to the $N$-fold tensor product of $S^{1/2}$ being decomposed into many more than $N$ subrepresentations.
Our approach in the next section will be to refine the direct sum to then apply this lemma.
We also  obtain optimal results using the only freedom we have in this construction: the partition chosen and the windows $I_k$.
\end{remark}

\section{Main Theorem}
\label{MainTheorem-Section}

We assume that $\lam_1 \leq  \dots \leq \lam_m$.
In Lemma \ref{Snearby} we will obtain nearby commuting self-adjoint matrices $A_i'$ for $A_i = \frac1N S^{\lam_1}\oplus\cdots\oplus S^{\lam_m}(\sigma_i)$. 

Let $A_r = \diag(i/N)$ for $-\lam_r \leq i \leq \lam_r$ and $S_r = \ws(d_{\lam_r, i}/N)$ for $-\lam_r \leq i < \lam_r$.  Then for $A = \bigoplus_r A_r$ and $S = \bigoplus_r S_r$, we have that $A_1 = \Re(S)$, $A_2 = \Im(S)$, and $A_3 = A$. The proof of Lemma \ref{Snearby} relies upon using the estimates in Lemma \ref{d-ineq} for the construction from Lemma \ref{gep}. We later optimize the result by choosing the lengths of the intervals $I_j$ optimally.

Dividing by $N$ here is referred to ``normalizing'' these operators.
For the moment we will focus only on the unnormalized weights $d_{\lam_r, i}$ and unnormalized spectrum. 
We assume that both $\lam_1$ and the maximum gap between the $\lam_r$ are not too small but also not too large. 
See Figure \ref{smallremoved}.
\begin{figure}[htp]  
    \centering
    \includegraphics[width=8cm]{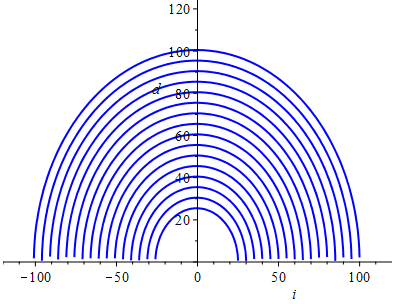}
    \caption{\label{smallremoved}\dark
    Illustration of the weights $d_{\lam_r, i}$ for  $\lam_r=25, 30, 35, \dots, 100$.}
\end{figure}
For this discussion, and hence the proof of Ogata's theorem, the estimates obtained in Lemma \ref{d-ineq} for $d_{\lam, i}$ are central to the calculation of the estimates for the nearby commuting matrices and influence the use of words such as ``small'' and ``large''. 

When calculating the estimate for $D_{I}$, 
\begin{figure}[htp]  
    \centering
    \includegraphics[width=8cm]{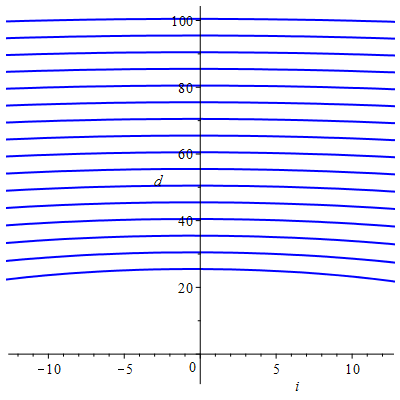}
    \caption{\label{near0}\dark
    Illustration of Figure \ref{smallremoved} focused on a small unnormalized interval $I = [-12,12]$ near $0$.}
\end{figure}
one is concerned with the largest value of the weight of the representation $S^{\lam_{r_0}}(\sigma_+)$ in the interval $I$, where $r=r_0$ is the smallest index so that the spectrum of $S^{\lam_{r}}(\sigma_3)$ spans the interval $I$.
See Figure \ref{near0} for an interval near $0$. In this example, $r_0 = 1$ and $D_I$ corresponds to the largest (unnormalized) weight of $S^{\lam_1}$, which is about $25$. 

In the proof of Ogata's theorem later in this paper, representations $S^\lam$ with small values of $\lam$ need to be dealt with separately due to the distribution of the multiplicities of the irreducible subrepresentations of the tensor representation. The reason that $\lam_{r+1}-\lam_r$ cannot be made very small and hence reduce the size of the $\varepsilon_I$ contribution to $G_I$ is also that it requires $m$ to be very large. 

As another example, consider the interval illustrated in Figure \ref{nearmiddle} that is not near $0$ or the boundary of the spectrum of $S^{\lam_m}$.
\begin{figure}[htp]  
    \centering
    \includegraphics[width=8cm]{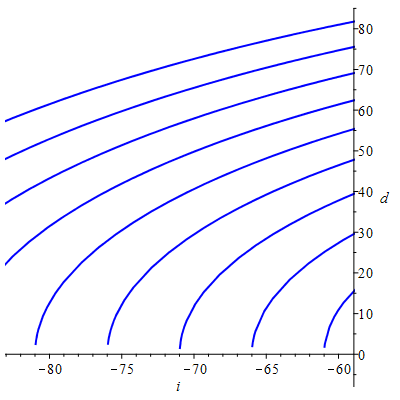}
    \caption{\label{nearmiddle}\dark
    Illustration of Figure \ref{smallremoved} focused on a small unnormalized interval $I=[-82, -60]$.}
\end{figure} 
In this case, $S^{\lam_{r_0}}$ corresponds to the arc passing the vertical axis a little more than $60$. For each $r < r_0$, the spectrum of $S^{\lam_r}(\sigma_3)$ does not span the interval and for each $r \geq r_0$ the spectrum does span the interval.

Because the gradual exchange process will be applied for all $r \geq r_0$, the estimate for $D_I$ will involve the largest weight of $S^{\lam_0}(\sigma_+)$, which is slightly larger than $60$. 
For an interval in this position, it is important that the length of the interval not be too large since although the smallest weight of $S^{\lam_{r_0}}(\sigma_+)$ may be small, its largest weight may be large based on the growth of the weights within an orbit. 
The length of the interval and the spacing of the $\lam_r$ give an inequality of the form $|\lam_{r_0} - |i|| \leq M$ so that $D_I$ is controlled by Lemma \ref{d-ineq}$(ii)$.

In this illustration, the smallest weight of $S^{\lam_{r_0}}(\sigma_+)$ is about 25 and if the interval were extended to the right, the largest weight of $S^{\lam_{r_0}}(\sigma_+)$ would grow. If the interval were only extended to the left, then at some point $r_0$ would necessarily increase by multiples of $5$ which then increases the largest weight of $S^{\lam_{r_0}}(\sigma_+)$ to about $70$ and so on.
So, we see that the length of the interval $I$ cannot be too large. Alternatively, the length of $I$ cannot be too small since then the spectrum of the $S^{\lam_{r}}(\sigma_3)$ in that interval will be small. So, the $R_I$ contribution to $G_I$ will be large through $N_I$ being small. These estimates get larger the farther this interval is from $0$.

\vspace{0.1in}

We now proceed to constructing nearby commuting matrices with various parameters in the estimates.
\begin{lemma} \label{Snearby}
Let $S = \frac1NS^{\lam_1}\oplus \cdots \oplus S^{\lam_m}$ where $S^\lam$ is the irreducible $(2\lam+1)$-dimensional spin representation of $su(2)$ with
$0\leq \lam_{r+1}-\lam_r \leq L, \lam_m= \Lam$ and the $2\lam_r$ are all even or all odd. 
Let $l, \Delta > 0$ with $4 \leq N\Delta \leq 2\Lam$.

Then there are commuting self-adjoint matrices $A_i'$ such that
\begin{align}\|A_1'-S(\sigma_1)\|, \|A_2'-S(\sigma_2)\| &\leq \max(G,D)+C_\alpha\left(\frac{\Lam+1/2}{N}\right)^{1-2\alpha}\max(T^\alpha,D^{2\alpha}), \nonumber\\
\|A_3'-S(\sigma_3)\| &\leq c_\Delta, \end{align}
where
\begin{align}
c_\Delta &= \frac{2\Lam}{N\lfloor 2\Lam / N\Delta \rfloor} \leq \frac{2\Lam}{2\Lam \Delta^{-1}-N}\\
N_0 &= \left\lfloor \frac{N\Delta-5}{2m-3} \right\rfloor-1\geq \frac{N\Delta-5}{2(m-1)-1} -2 \\
T&=\left(2+\frac{2L}{N_0}\right)\frac{\Lam}{N^2}
\\
G=\frac1N\max\left( \sqrt{\Lam}\frac{2L}{\sqrt{l}}+ \frac{\pi}{2N_0}\right.&\left.\left(\Lam+1/2\right), \sqrt{2\Lam L}+\frac{\pi}{2N_0}\sqrt{2\Lam(l+1)}\right)
\end{align}
\begin{align}
D= \max\left(\sqrt{\frac{2\Lam}N\left(\frac{L+1}{N}+c_\Delta\right)},\frac{\lam_{1}+1/2}N
, \frac{c_{\Delta}}2+\frac{L+1/2}N\right),
\end{align}
$\alpha \in(0, 1/2], C_\alpha > 0$ are constants as in Theorem \ref{BergResult}, and $A_3'$ is real. Consequently, when using $\alpha = 1/3, C_{1/3} = 5.3308$, we have that $A_1', iA_2', A_3'$ are real.
\end{lemma}
\begin{proof}
We wish to apply Lemma \ref{gep} with 
\[
A_r  = \frac1NS^{\lam_r}(\sigma_3) = \diag\left(-\frac{\lam_r}N, \frac{-\lam_r+1}N, \dots, \frac{\lam_r}N\right)\]
\[S_r = \frac1NS^{\lam_r}(\sigma_+) = \ws\left(\frac{d_{\lam_r, -\lam_r}}N, \frac{d_{\lam_r, -\lam_r+1}}N, \dots, \frac{d_{\lam_r, \lam_r-1}}N\right)\]
so that
\[A_r = \diag\left(\frac{i}{N}\right), \; i = -\lam_r, -\lam_r+1, \dots, \lam_r\]
\[S_r = \ws\left(\frac{d_{\lam_r,i}}{N}\right), \; i = -\lam_r, -\lam_r+1, \dots, \lam_{r}-1.\]
Set $A = \bigoplus_r A_r$ and $S = \bigoplus_r S_r$ and $\alpha_i = i/N$, $c_i^r = d_{\lam_r,i}/N\geq0$ in accordance with the assumptions of Lemma \ref{gep}.
So, the estimates of $c_i^r$ and $c_i^{r+1}-c_i^r$ needed to apply Lemma \ref{gep} will be obtained from the inequalities for $d_{\lam_r, i}$ and $d_{\lam_{r+1},i}-d_{\lam_r, i}$ in Lemma \ref{d-ineq}.
We will then obtain nearby commuting $A', S''$ such that $A'$ is Hermitian and $S''$ is normal. We then set $A_1' = \Re(S''), A_2' = \Im(S''),$ and $A_3' = A'$.
 
We choose an increasing sequence of real numbers $a_i$ to satisfy the conditions of Lemma \ref{gep} with $a_{1} = -\Lam/N$ and $a_{n_0}=\Lam/N$ satisfying
\begin{align*}
a_{k+1}-a_k= 
c_\Delta,
\end{align*}
where 
\begin{align}
n_\Delta = \left\lfloor \frac{2\Lam/N}{\Delta} \right\rfloor, \, c_\Delta = \frac{2\Lam/N}{n
_\Delta}\geq \Delta,\nonumber
\end{align}
requiring $2\Lam/N \geq \Delta$ so $2\Lam \geq N\Delta$.
So, the intervals $I_k$ have the same length, which is at least $\Delta$ and is asymptotically equal to $\Delta$ as $N\Delta/\Lam \to 0$.
Note that
\begin{align}\label{num-size}
N\Delta -1\leq \#\sigma(A_r) \cap [a_k, a_{k+1})
\end{align}
and we require that $N\Delta -1 \geq 3$ so $N\Delta \geq 4$.

We now move to calculating the various estimates in Lemma \ref{gep}. 

\vspace{0.1in}

\noindent \underline{Estimating $D_{I_k}$}: There are two types of intervals $I=I_k$. If $n_\Delta$ is odd, then $I_{(n_\Delta+1)/2} = [-c_\Delta/2, c_\Delta/2]$. All other intervals are of the form  $[-b,-b+c_\Delta]$ or $[b-c_\Delta,b]$ for $b \geq c_\Delta$.

We first deal with the exceptional case. Recall that $\sigma(A_r)$ consists of $-\lam_r/N, \dots, \lam_r/N$. So, the sets $\sigma(A_r)$ are nested consecutive and symmetric intervals in $\frac1N\Z$. 
Recall that $r_0 = \min \mathscr R_{I}$ is the smallest $r$ so that $\sigma(A_r)$ contains $\sigma(A) \cap I$. We then bound
\[D_I\leq \max_i c_i^{r_0} \leq \frac{\lam_{r_0}+1/2}N\]
by Lemma \ref{d-ineq}$(i)$.
If $r_0 = 1$, then we obtain
\[D_I \leq \frac{\lam_{1}+1/2}N.\]
So, suppose that $r_0 > 1$. Because 
\[\frac{\lam_{r_0-1}}N< c_{\Delta}/2 \leq \frac{\lam_{r_0}}N\] and $\lam_{r_0} \leq \lam_{r_0-1}+L$, we see that $\lam_{r_0}\leq Nc_{\Delta}/2+L$. So,
\[D_I \leq \frac{Nc_{\Delta}/2+L+1/2}N=\frac{c_{\Delta}}2+\frac{L+1/2}N.\]

So, suppose that $I$ is not the central interval of the previous case. 
If $r_0 = 1$ we apply the same bound as before. So, suppose that $r_0 > 1$.
If $I = [-b, -b+c_\Delta]$ or $I = [b-c_\Delta,b]$ then \[\frac{\lam_{r_0-1}}N < b \leq \frac{\lam_{r_0}}N.\] Because $\lam_{r_0}\leq \lam_{r_0-1}+ L$, we obtain 
\[\lam_{r_0}-N|x| \leq L+Nc_\Delta, \; x \in I.\]
So, suppose $x=|i|/N \in I$ so that $i \in [-\lam_{r_0}, \lam_{r_0}]$.
Using $M= L+Nc_\Delta$ in  Lemma \ref{d-ineq}$(ii)$, we have
\[d_{\lam_{r_0}, i}\leq \sqrt{2\lam_{r_0}(M+1)} \leq \sqrt{2\Lam(L+Nc_\Delta+1)}\]
and hence
\[c_{i}^{r_0} \leq \frac1N\sqrt{2\Lam(L+Nc_\Delta+1)}.\]

Therefore, we obtain the bound from the statement of the lemma: $D_I \leq D$.

\vspace{0.1in}

\noindent \underline{Estimating $G_I$}:
Note that in order to apply Lemma \ref{gep}, we need $(2m-3)(N_0+1)+4 \leq \#\sigma(A)\cap I$, where we choose $N_k = N_0$ for all $k$. The definition of 
$N_0$ in the statement of the lemma was made to satisfy this inequality through Equation (\ref{num-size}).

The estimate of $G_I$ involves estimating the sum of the two terms $c^{r+1}_i-c^r_i$ and $\frac{\pi}{2N_0}\max(c^{r+1}_i,c^r_i)$. Using Lemma \ref{d-ineq}$(iv)$ and $\lam_{r+1} \leq \Lam$, we obtain the bound
\begin{align}\nonumber
|c^{r+1}_i-c^r_i| &+ \frac{\pi}{2N_0}\max(c^{r+1}_i,c^r_i)\leq G.
\end{align}

\vspace{0.1in}

\noindent \underline{Estimating Equation (\ref{S''ineq})}:
By Lemma \ref{d-ineq}$(vi)$, \[\|[S^\ast, S]\| \leq \frac{2\Lam}{N^2}.\] By Lemma \ref{d-ineq}$(v)$, for all the weights  \[|(c^{r+1}_i)^2-(c^r_i)^2| \leq \frac{2\Lam L}{N^2}.\]
By Lemma \ref{d-ineq}$(i)$, $\|S\| \leq (\Lam+1/2)/N$. Note that we require $\alpha \leq 1/2$ so that $1-2\alpha \geq 0$.
The desired estimate then follows from the estimates of $\|S'-S\|$ and $\|S''-S'\|$ from Lemma \ref{gep}.

When using $\alpha = 1/3, C_{1/3} = 5.3308$, we have $A'$ and $S''$ real so $\Re(S'')$  and $i\Im(S'')$ are as well. We now collect what we showed into the statement of the lemma.  
\end{proof}

\begin{example}\label{Ex1}
We assume that the constants in the statement of Lemma \ref{Snearby} satisfy the asymptotic estimates
\begin{align}\label{asymptBounds}
\lam_1\leq c_0N^{\gamma_0}, m-1 \leq c_1 N^{\gamma_1}, \underline{c_2} N^{\underline{\gamma_2}}\leq \lam_m \leq c_2 N^{\gamma_2}, \nonumber\\
L \leq c_3 N^{\gamma_3}, l = c_4 N^{\gamma_4}, \Delta = c_5 N^{-\gamma_5}.
\end{align}
We assume $N \geq N_\ast \geq 1$. Note that $N$ will be an integer, though $N_\ast$ is not assumed to be. Although we will prove more in this discussion, what we will use from it for Ogata's theorem is expressed in Lemma \ref{Ex2Lemma}.

We now explore some mild assumptions on the exponents to obtain nearby commuting matrices using Lemma \ref{Snearby}.
First, $\gamma_0,  \gamma_1, \gamma_2, \gamma_3, \gamma_5 > 0$.  
Because $\lam_1 \leq \lam_m$, we expect $\gamma_0 \leq \gamma_2$. 
Because $\lam_m-\lam_1 \leq (m-1)L$ and often $\lam_1 = o(\lam_m)$, we will often have $ \gamma_2 \leq \gamma_1 + \gamma_3$. For reasons explained below, we expect $\gamma_1 \leq \gamma_2$ as well. We will assume that $\gamma_1+\gamma_5\leq1$ so that $N_0$ can be large. To make the term coming from $\|S\|$ bounded by a constant, we will assume that $\gamma_2\leq 1$.

The constants $l$ and $\Delta$ are chosen, while the others are given. 
In particular, $l$ will be chosen so that the first and fourth term in the estimate of $G$ are equalized and negligible. 
Because the optimal value of $l$ is not a simple expression, we elect to choose $l$ after the estimate for $G$ is expressed in terms of the $c_i$, $\gamma_i$, and $N_\ast$. 

Choosing the optimal constant and exponent for $\Delta$ in this generality requires knowing more information about the relative sizes of the exponents in the definitions of $G$, $D$ and $T$. 
We make further assumptions about the exponents after having done as much simplification as possible.
The necessary condition $4 \leq N\Delta \leq 2\Lam$ becomes
\[4 \leq c_5N^{1-\gamma_5} \leq 2\underline{c_2}N^{\underline{\gamma_2}}\]
\[
4 \leq c_5N^{1-\gamma_5}, \;\; c_5 \leq 2\underline{c_2}N^{\underline{\gamma_2}+\gamma_5-1}.
\]
So, we further assume that $\gamma_5\leq 1$ and $\underline{\gamma_2}+\gamma_5\geq1$.

We first find the optimal exponent for $\max(G,D)+C_{\alpha, \Lam, N}\max(T^\alpha,D^{2\alpha})$. Note with $\alpha \leq 1/2$, we will use $\Lam \leq Const. N$ so that $C_{\alpha, \Lam, N}$ is bounded by a constant.
It should be noted that we will not consider the asymptotics of $c_\Delta$ for the matrix $A_3'$ during the optimization of the exponent because $D > c_\Delta/2$.

After finding the optimal exponent, we then bound all the terms by a constant factor multiplied by a single power of $N$. In particular, for $N \geq N_\ast$, all terms that are negligible will contribute to the constant factor in a way that depends on $N_\ast$ as follows. 
The primary inequality that will be used to choose optimal constant factors will be repeated applications of the following simple observation that if $a \geq b, N \geq N_\ast$ then
\[N^b = N^{b-a}N^a \leq N_\ast^{b-a}N^a\]
In particular, if $a \geq 0$ then \[1 \leq N_\ast^{-a}N^a.\]

\vspace{0.05in}

\noindent We now proceed to the calculations.\\ \underline{$c_\Delta$}:
Because $x\mapsto x/(ax-b)$ is decreasing as a function of $x > b/a$, we have
\begin{align}
c_\Delta &\leq \frac{2\lam_m}{2\lam_m\Delta^{-1}-N     }\leq \frac{2\underline{c_2}N^{\underline{\gamma_2}}}{\frac{2\underline{c_2}}{c_5}N^{\underline{\gamma_2}+\gamma_5}-N} \leq \frac{2\underline{c_2}N^{\underline{\gamma_2}}}{\frac{2\underline{c_2}}{c_5}N^{\underline{\gamma_2}+\gamma_5}-N_\ast^{1-\underline{\gamma_2}-\gamma_5}N^{\underline{\gamma_2}+\gamma_5}} \nonumber\\
&= c_5\left(\frac{2\underline{c_2}}{2\underline{c_2}-c_5N_\ast^{1-\underline{\gamma_2}-\gamma_5}}\right)N^{-\gamma_5} = d_\Delta N^{-\gamma_5},\label{dDelta}
\end{align}
where we assume that $d_\Delta > 0$ (or equivalently $c_5 < 2\underline{c_2}N_\ast^{\underline{\gamma_2}+\gamma_5-1}$).
Note that the upper bound for $c_\Delta$ through that of $d_\Delta$ is the only place in our calculations where we use the lower bound for $\lam_m$. This guarantees that $\lam_m$ is much larger than $N\Delta$ so that $c_\Delta$ is approximately equal to $\Delta = c_5N^{-\gamma_5}$.

\vspace{0.05in}

\noindent \underline{$N_0$}:
\begin{align}
N_0 &\geq  \frac{c_5N^{-\gamma_5+1}-5}{2c_1N^{\gamma_1}-1} -2 \geq \frac{c_5N^{-\gamma_5+1}-5N_\ast^{\gamma_5-1}N^{-\gamma_5+1}}{2c_1N^{\gamma_1}} -2N_\ast^{\gamma_1+\gamma_5-1}N^{-\gamma_1-\gamma_5+1}\nonumber\\
&= \left(\frac{c_5-5N_\ast^{\gamma_5-1}}{2c_1} -2N_\ast^{\gamma_1+\gamma_5-1}\right)N^{-\gamma_1-\gamma_5+1} = d_0 N^{-\gamma_1-\gamma_5+1},\label{d0}
\end{align}
where we used the assumption that $\gamma_1+\gamma_5 \leq 1$. We further assume that $d_0>0$ and $2c_1N^{\gamma_1}>1$.

\vspace{0.05in}

\noindent \underline{$T$}:
\begin{align*}
T&\leq\left(2+\frac{2c_3 N^{\gamma_3}}{d_0 N^{-\gamma_1-\gamma_5+1}}\right)\frac{c_2 N^{\gamma_2}}{N^2} = 2c_2N^{\gamma_2-2}+\frac{2c_2c_3}{d_0 }N^{\gamma_1+\gamma_2+\gamma_3+\gamma_5-3} \end{align*}

\vspace{0.05in}

\noindent \underline{$G$}:
\begin{align*}
G&\leq\frac1N\max\left( \sqrt{c_2N^{\gamma_2}}\frac{2c_3N^{\gamma_3}}{\sqrt{c_4}}N^{-\gamma_4/2}+ \frac{\pi}{2d_0 N^{-\gamma_1-\gamma_5+1}}\left(c_2N^{\gamma_2}+\frac12\right
),\right.\\
&\;\;\;\;\;\;\;\;\;\;\;\;\;\;\;\;\;\;\left.\sqrt{2c_2N^{\gamma_2}(c_3N^{\gamma_3})}+\frac{\pi}{2d_0 N^{-\gamma_1-\gamma_5+1}}\sqrt{2c_2N^{\gamma_2}(c_4N^{\gamma_4}+1)}\right)\\
&\leq \frac1N\max\left( 2c_3\sqrt{\frac{c_2}{c_4}}N^{\gamma_3+(\gamma_2-\gamma_4)/2}+ \frac{\pi}{2d_0 }N^{\gamma_1+\gamma_5-1}\left(c_2N^{\gamma_2}+\frac12N_\ast^{-\gamma_2}N^{\gamma_2}\right
),\right.\\
&\;\;\;\;\;\;\;\;\;\;\;\;\;\;\;\;\;\;\left.\sqrt{2c_2c_3}N^{(\gamma_2+\gamma_3)/2}+\frac{\pi}{2d_0} N^{\gamma_1+\gamma_5-1}\sqrt{2c_2c_4N^{\gamma_2+\gamma_4}+2c_2N_\ast^{-\gamma_4}N^{\gamma_2+\gamma_4}}\right)\\
&= \max\left( 2c_3\sqrt{\frac{c_2}{c_4}}N^{\gamma_3+(\gamma_2-\gamma_4)/2-1}+ \frac{\pi}{2d_0 }\left(c_2+\frac12N_\ast^{-\gamma_2}\right)N^{\gamma_1+\gamma_2+\gamma_5-2},\right.\\
&\;\;\;\;\;\;\;\;\;\;\;\;\;\;\;\;\;\;\left.\sqrt{2c_2c_3}N^{(\gamma_2+\gamma_3)/2-1}+\frac{\pi}{2d_0} \sqrt{2c_2c_4+2c_2N_\ast^{-\gamma_4}}N^{\gamma_1+\gamma_5+(\gamma_2+\gamma_4)/2-2}\right)
\end{align*}

With the choice of $\gamma_4 = -\gamma_1+\gamma_3-\gamma_5+1$, we equalize the exponents in the first and fourth terms, obtaining
\begin{align*}
G&\leq \max\left( 2c_3\sqrt{\frac{c_2}{c_4}}N^{(\gamma_1+\gamma_2+\gamma_3+\gamma_5-3)/2}+ \frac{\pi}{2d_0 }\left(c_2+\frac12N_\ast^{-\gamma_2}\right)N^{\gamma_1+\gamma_2+\gamma_5-2},\right.\\
&\;\;\;\;\;\;\;\;\;\;\;\;\;\;\;\;\;\;\left.\sqrt{2c_2c_3}N^{(\gamma_2+\gamma_3)/2-1}+\frac{\pi}{2d_0} \sqrt{2c_2c_4+2c_2N_\ast^{-\gamma_4}}N^{(\gamma_1+\gamma_2+\gamma_3+\gamma_5-3)/2}\right).
\end{align*}
Note that the first and fourth terms are not asymptotically larger than the third term because $\gamma_1+\gamma_5 \leq 1$. Later we will have a strict inequality so that these two terms become negligible as $N \to \infty$.

\vspace{0.05in}

\noindent \underline{$D$}:
\begin{align*}
D\leq \max&\left(\sqrt{\frac{2c_2N^{\gamma_2}}N\left(\frac{c_3 N^{\gamma_3}+1}{N}+d_\Delta N^{-\gamma_5}\right)},\frac1N\left(c_0N^{\gamma_0}+\frac12\right),\right. \\ &\;\;\;\;\;\;\;\;\;\;\;\;\;\left. \frac{d_\Delta }2N^{-\gamma_5}+\frac{c_3N^{\gamma_3}+1/2}N\right)\\
\leq \max&\left(\sqrt{2c_2c_3N^{\gamma_2+\gamma_3-2}+2c_2N^{\gamma_2-2}+2c_2d_\Delta N^{\gamma_2-\gamma_5-1}}, c_0N^{\gamma_0-1}+\frac12N^{-1},\right. \\ &\;\;\;\;\;\;\;\;\left. \frac{d_\Delta }2N^{-\gamma_5}+c_3N^{\gamma_3-1}+\frac12N^{-1}\right)
.
\end{align*}
Note that the first term in the bound for $D$ has three components, the first of which is asymptotically equal to the third term of $G$, considering the square root. 

\vspace{0.05in}

\noindent \underline{Optimal Asymptotics}:\\
Recall that $\alpha\leq 2\alpha \leq 1$. We see that the slowest decaying term of $\max(G,D)+C_{\alpha, \Lam, N}\max(T^\alpha,D^{2\alpha})$ has exponent
\begin{align*}
-\gamma &= \max\left( \gamma_1+\gamma_2+\gamma_5-2, 2\alpha\left(\frac{\gamma_2+\gamma_3}{2}-1\right), 2\alpha\left(\frac{\gamma_2-\gamma_5-1}{2}\right), 2\alpha(\gamma_0-1),\right. \\
&\;\;\;\;\;\;\;\;\;\;\;\;\;\;\;\;\; \left.  -2\alpha\gamma_5,2\alpha(\gamma_3-1), \alpha(\gamma_2-2), \alpha(\gamma_1+\gamma_2+\gamma_3+\gamma_5-3)^{\tcw{|}}\right).
\end{align*}
So, $-\gamma$ is the largest of several exponents that, minimally, we wish to choose to be negative. We will then minimize  $-\gamma$. Note that its optimal value will depend on $\alpha$ as well as the appropriate choice of the $\gamma_i$.

We now impose additional assumptions on the exponents $\gamma_i$.
We further assume that we have $\gamma_2 = \gamma_1+\gamma_3$. So, we assume that $\gamma_3 \leq \gamma_2$. This corresponds to having a bound for the spacing $\lam_{r+1}-\lam_r$ that is asymptotically equal to the bound of the average spacing $(\lam_m - \lam_1)/(m-1)$ if additionally $m-1 \geq Const. N^{\gamma_1}$.

Substituting $\gamma_1 = \gamma_2-\gamma_3$, we obtain
\begin{align}\nonumber
-\gamma &= \max\left( 2\gamma_2-\gamma_3+\gamma_5-2, \alpha\left(\gamma_2+\gamma_3-2\right), \alpha\left(\gamma_2-\gamma_5-1\right), 2\alpha(\gamma_0-1),  -2\alpha\gamma_5,\right. \\
&\;\;\;\;\;\;\;\;\;\;\;\;\;\;\;\;\; \left.2\alpha(\gamma_3-1), \alpha(\gamma_2-2), \alpha(2\gamma_2+\gamma_5-3)\right).\label{gamma}
\end{align}
Note that the requirement $\gamma_1+\gamma_5\leq1$ becomes $\gamma_2-\gamma_3+\gamma_5\leq1$.

We now bound our estimates for $G, D, D^{2\alpha}, T^{\alpha}$ by a constant multiple of $N^{-\gamma}$. Note that by definition, if $a$ is an exponent such that $a \leq -\gamma$ then
\[N^a = N^{a+\gamma}N^{-\gamma} \leq N_0^{a+\gamma}N^{-\gamma}\]
since $a+\gamma \leq 0$.

So,
\begin{align}\label{Gestimate}
G&\leq \max\left( 2c_3\sqrt{\frac{c_2}{c_4}}N_\ast^{\frac{2\gamma_2+\gamma_5-3}2+\gamma}+ \frac{\pi}{2d_0 }\left(c_2+\frac12N_\ast^{-\gamma_2}\right)N_\ast^{2\gamma_2-\gamma_3+\gamma_5-2+\gamma},\right.\\
&\;\;\;\;\;\;\;\;\;\;\;\;\;\;\;\;\;\;\left.\sqrt{2c_2c_3}N_\ast^{\frac{\gamma_2+\gamma_3}2-1+\gamma}+\frac{\pi}{2d_0} \sqrt{2c_2c_4+2c_2N_\ast^{\gamma_2-2\gamma_3+\gamma_5-1}}N_\ast^{\frac{2\gamma_2+\gamma_5-3}2+\gamma}\right)N^{-\gamma},\nonumber
\end{align}
\begin{align}\label{Destimate}\nonumber
D\leq \max&\left(\sqrt{2c_2c_3N_\ast^{\gamma_2+\gamma_3-2+2\gamma}+2c_2N_\ast^{\gamma_2-2+2\gamma}+2c_2d_\Delta N_\ast^{\gamma_2-\gamma_5-1+2\gamma}}\right.,
\\ &\;\;\;\;\;\;\;\; c_0N_\ast^{\gamma_0-1+\gamma}+\frac12N_\ast^{-1+\gamma}, \\ &\;\;\;\;\;\;\;\;\left. \frac{d_\Delta }2N_\ast^{-\gamma_5+\gamma}+c_3N_\ast^{\gamma_3-1+\gamma}+\frac12N_\ast^{-1+\gamma}\right)N^{-\gamma},\nonumber
\end{align}
\begin{align}\label{Testimate}
T^\alpha&\leq \left(2c_2N_\ast^{\gamma_2-2+\frac\gamma\alpha}+\frac{2c_2c_3}{d_0 }N_\ast^{2\gamma_2+\gamma_5-3+\frac\gamma\alpha}\right)^\alpha N^{-\gamma}, \end{align}
\begin{align}\label{Dalphaestimate}\nonumber
D^{2\alpha}\leq \max&\left(\sqrt{2c_2c_3N_\ast^{\gamma_2+\gamma_3-2+\frac\gamma{\alpha}}+2c_2N_\ast^{\gamma_2-2+\frac\gamma{\alpha}}+2c_2d_\Delta N_\ast^{\gamma_2-\gamma_5-1+\frac\gamma{\alpha}}}\right.,
\\ &\;\;\;\;\;\;\;\; c_0N_\ast^{\gamma_0-1+\frac\gamma{2\alpha}}+\frac12N_\ast^{-1+\frac\gamma{2\alpha}}, \\ &\;\;\;\;\;\;\;\;\left. \frac{d_\Delta }2N_\ast^{-\gamma_5+\frac\gamma{2\alpha}}+c_3N_\ast^{\gamma_3-1+\frac\gamma{2\alpha}}+\frac12N_\ast^{-1+\frac\gamma{2\alpha}}\right)^{2\alpha}N^{-\gamma}.\nonumber
\end{align}

\end{example}

We write the result of the previous example as a lemma.
\begin{lemma}\label{Ex2Lemma}
Let $S = \frac1NS^{\lam_1}\oplus \cdots \oplus S^{\lam_m}$ where $S^\lam$ is the irreducible $(2\lam+1)$-dimensional spin representation of $su(2)$ with $0 \leq\lam_{r+1}-\lam_r \leq L$  and the $2\lam_r$ are all even or all odd. 

Suppose further that
\begin{align*}
\lam_1\leq c_0N^{\gamma_0}, m-1 \leq c_1 N^{\gamma_2-\gamma_3}, \underline{c_2} N^{\underline{\gamma_2}}\leq \lam_m \leq c_2 N^{\gamma_2}, \\
L \leq c_3 N^{\gamma_3}, l = c_4 N^{ -\gamma_2+2\gamma_3-\gamma_5+1}, \Delta = c_5 N^{-\gamma_5},
\end{align*}
where $\gamma_i, c_i, \underline{c_2} > 0, \gamma_0 < 1, \gamma_i, \underline{\gamma_2} \leq 1, \gamma_3 \leq \gamma_2$, $\underline{\gamma_2}+\gamma_5\geq1$, and $\gamma_2-\gamma_3+\gamma_5\leq 1$.
Suppose that the $c_i$ and $N_\ast$ satisfy the inequalities \begin{align}\label{ineq-requirements}
1 &< 2c_1N_{\ast}^{\gamma_2-\gamma_3}, \;\; 4c_1N_\ast^{\gamma_2-\gamma_3+\gamma_5-1}+5N_\ast^{\gamma_5-1} < c_5, \nonumber \\
4 &\leq c_5N_\ast^{1-\gamma_5}, \;\; c_5 < 2\underline{c_2}N_\ast^{\underline{\gamma_2}+\gamma_5-1}\nonumber
\end{align}

Let 
$C_{\alpha, \Lam, N} = C_\alpha\left(\frac{\lam_m+1/2}N\right)^{1-2\alpha}\leq Const.$, where $\alpha, C_\alpha$ are as in Theorem \ref{BergResult} with additionally $\alpha\leq 1/2$. Let $d_\Delta$ and $d_0$ be defined by Equations (\ref{dDelta}) and (\ref{d0}) and let $\gamma=\gamma(\alpha, \gamma_i)$ be defined by Equation (\ref{gamma}). 

Then we have the bounds for $G, D, T^\alpha, D^{2\alpha}$ from Lemma \ref{Snearby} of the form \\
$C(\alpha, c_i, \gamma_i, \underline{c_2}, \underline{\gamma_2}, N_\ast)N^{-\gamma}$ in Equations (\ref{Gestimate}), (\ref{Destimate}), (\ref{Testimate}), and (\ref{Dalphaestimate}) so that there are commuting self-adjoint matrices $A_i'$ such that
\begin{align}
\|A_1'-S(\sigma_1)\|, \|A_2'-S(\sigma_2)\| &\leq \max(G,D)+C_{\alpha, \Lam, N}\max(T^\alpha,D^{2\alpha}) \leq Const. N^{-\gamma}, \nonumber\\
\|A_3'-S(\sigma_3)\| &\leq d_\Delta N^{-\gamma_5}.\nonumber
\end{align}
Moreover, when using $\alpha = 1/3, C_{1/3} = 5.3308$, we have that $A_1', iA_2', A_3'$ are real.
\end{lemma}

\begin{example}
With the set-up of the previous example, suppose that we are interested in the optimal exponent and the constant obtained as $N_\ast\to \infty$ when $\alpha=1/3$.

For this example, we will assume that $c_1c_3 \geq c_2$. In the next lemma below, we treat the details of this constraint which approximately holds when $N_\ast$ is large, $\lam_1 = o(\lam_m)$, and $\lam_{r+1}-\lam_r$ is constant in $r$. Due to this assumption, we can easily remove the dependence of $c_1$ as follows:
The only 
occurrence of $c_1$ in our inequalities is in $G$ and $T$ through $d_0^{-1}$. 
We see that both $G$ and $T$ are decreased when $c_1$ is decreased, so we choose $c_1 = c_2/c_3$.

For this calculation, we assume that $\lam_m = N/2$ so that $\underline{c_2}=c_2 = 1/2, \underline{\gamma_2}=\gamma_2 = 1$. The condition $\gamma_2-\gamma_3+\gamma_5 \leq 1$ then becomes $\gamma_5 \leq \gamma_3$. We choose $\lam_1 = O(N^{1/2})$ by taking $\gamma_0 =1/2$.

For $\alpha = 1/3$, the optimal choices of $\gamma_3 = 4/7, \gamma_5 = 3/7$ give $\gamma = 1/7$. Then the exponents in Equation (\ref{gamma}) are
\[-\frac17, -\frac17, -\frac17, -\frac13, -\frac27, -\frac27, -\frac13, -\frac4{21}.\]
So, the slowest decaying terms have exponents $2\gamma_2 - \gamma_3 + \gamma_5 - 2, \alpha(\gamma_2 + \gamma_3 - 2), \alpha(\gamma_2 - \gamma_5 - 1)$ which equal $-1/7$.
We note that as $N_\ast \to \infty$, we obtain that $d_0 \sim c_5/(2c_1) = c_3c_5/2c_2, d_\Delta \sim c_5$.

So asymptotically,
\begin{align*}
G&\leq \frac{\pi c_2}{2d_0}N^{-1/7}+o(N^{-1/7})= \frac{\pi }{4c_3c_5}N^{-1/7}+o(N^{-1/7}),\\
T^{1/3}&=o(N^{-1/7}),\\
D \ll D^{2/3}&\leq (2c_2c_3+2c_2d_\Delta)^{1/3} N^{-1/7}+o(N^{-1/7}) = (c_3+c_5)^
{1/3} N^{-1/7}+o(N^{-1/7}).
\end{align*}
To approximately optimize our estimate of $\left(\frac{\pi }{4c_3c_5}+5.3308\left(\frac12\right)^{1/3}(c_3+c_5)^
{1/3}\right)N^{-1/7}$, we choose $c_3, c_5 = 0.95$. So, for $N$ large, there are nearby commuting matrices $A_i'$ satisfying the following inequalities 
\begin{align}\nonumber
\|A_1' - S(\sigma_1)\|, \|A_2' - S(\sigma_2)\|&\leq 6.111\, N^{-\frac17}\\
\|A_3' - S(\sigma_3)\|&\leq 0.951\, N^{-\frac37}\nonumber
\end{align}

This estimate  shows that we might as well assume that $N_\ast$ is at least $(2\cdot 6.111)^7 > 4.07\times 10^7$. This is because $\|S(\sigma_i)\|=\frac12$
so it is only when $N \geq (2\cdot 6.111)^7$ that the obtained estimate is better than trivially choosing $A_1' = A_2' = 0, A_3'=A_3$.  

\end{example}

We now prove the following lemma that is closer to what will be used for Ogata's theorem. This result is a modification of the previous example that holds for all $N$.
\begin{lemma}\label{bigLstepLemma}
Let $N \geq 1$, $\Lambda_0 \leq \frac12N^{1/2}+\frac32$, and $L = \lfloor 1.045\, N^{4/7} \rfloor$.
Let $S = \frac1NS^{\lam_1}\oplus \cdots \oplus S^{\lam_m}$ with  $\lam_1 \leq \Lambda_0+2L$, $\lam_m \leq N/2$, and $\lam_{r+1}-\lam_r = L$. 

Then there are commuting self-adjoint matrices $A_i'$ such that
\begin{align}\|A_1'-S(\sigma_1)\|, \|A_2'-S(\sigma_2)\| &\leq 6.286\,N^{-\frac17}, \nonumber\\
\|A_3'-S(\sigma_3)\| &\leq 1.083\,N^{-\frac37}\nonumber\end{align}
and $A_1', iA_2', A_3'$ are real.
\end{lemma}
\begin{proof}
Note that the variables $N_\ast, c_3,$ and $\underline{c_2}$ will be left undetermined until the end of the proof. We also at this point define $L = \lfloor c_3N^{4/7}\rfloor \leq c_3N^{\gamma_3}$ with $\gamma_3 = \frac47$. We will obtain estimates for three cases then choose the optimal values for these constants to obtain the result of the lemma.

\vspace{0.1in}

\noindent \underline{$\lam_m < \underline{c_2}N^{6/7}$}: \\
This case only relies the value of the variable $\underline{c_2}$. By Equation (\ref{lamnorm}) we have
\[ \|S(\sigma_i)\| < \frac{\underline{c_2}N^{6/7}}{N} = \underline{c_2}N^{-1/7}.\] So, we may safely choose $A_1' =  A_2' = 0$ and $A_3' = S(\sigma_3)$. 
The estimates in the statement of the lemma that we obtain are $\|A_3'-S(\sigma_3)\| = 0$ and for $i = 1, 2$,
\[\|A_i'-S(\sigma_i)\| = \|S(\sigma_i)\|  < \underline{c_2}N^{-1/7}.\]

\vspace{0.1in}

\noindent \underline{$N < N_\ast$}: \\
This case only relies the value of the variable $N_\ast$.

As in the previous case, we choose $A_1' = A_2' = 0$ and $A_3' = S(\sigma_3)$. Because $N^{1/7} < N_\ast^{1/7}$, we have
\[\|S(\sigma_i)\| \leq \frac{N/2}{N}< \frac12 N_\ast^{1/7}N^{-1/7}.\]

\noindent \underline{$N \geq N_\ast$, $\underline{c_2}N^{6/7} \leq \lam_m$}:
This is the only non-trivial case and it  relies on the values of $N_\ast, c_3,$ and $\underline{c_2}$. Due to our use of Lemma \ref{Ex2Lemma}, we will also have other constants 

We will apply Lemma \ref{Ex2Lemma} with exponents $\gamma_0 = \frac47, \underline{\gamma_2} = \frac67, \gamma_2 = 1, \gamma_3 = \frac47,  \gamma_5 = \frac37, \gamma=\frac17$ and with $c_2 = \frac12$. 

First note that 
\[\lam_0 \leq \Lambda_0+2L \leq 2c_3N^{\frac47}+\frac12N^{\frac12}+\frac32\leq \left(2c_3+\frac12N_\ast^{-\frac1{14}}+\frac32N_\ast^{-\frac47}\right)N^{\frac47}=c_0N^{\gamma_0}.\]
Also, because $\lam_{r+1}-\lam_r$ is constant, we see that 
\[m-1 = \frac{\lam_m - \lam_1}{L} \leq \frac{\lam_m}{c_3N^{\frac47}-1} \leq\frac{N}{2c_3N^{\frac47}-2N_\ast^{-\frac47}N^{\frac47}} =  \frac{1}{2c_3 - 2N_\ast^{-\frac47}}N^{\frac37} = c_1N^{\gamma_1}.\]
Observe that the exponent provided here is $\gamma_1 = \gamma_2-\gamma_3 = 1-\frac47 = \frac37.$

\vspace{0.05in}

\noindent \underline{Choice of constants}: So, at this point we only need to choose the values for $\underline{c_2}, c_3, c_4, c_5,$ and $N_\ast$ for the estimate. We choose the approximately optimal $c_3 = 1.045, c_4 = 18.65, c_5 = 1.082, \underline{c_2}=6.285,$ and $N_\ast = 4.962\times10^7$.
We then obtain the results of the lemma from all these cases, noting that the required conditions on the constants hold.
\end{proof}

\begin{example}\label{OgataEx}
Using the following example, we will illustrate how we prove our extension of
Ogata's theorem (Theorem \ref{OgataTheorem}) over the next two theorems. Consider the 
scaled representation
\begin{align*}
S = \frac{1}{28}\left (2S^1\right.&\left.\oplus4S^2\oplus7S^3\oplus 8S^4\oplus7S^5\oplus6S^6\oplus4S^7\oplus4S^8\right.\\
&\left.\oplus4S^9\oplus3S^{10}\oplus3S^{11}\oplus2S^{12}\oplus S^{13}\oplus S^{14} \right)
\end{align*}
with multiplicities illustrated in Figure \ref{PartitionOfReps}(a).
Recall that, just as in the next two results, the $1/28$ is a multiplicative factor while the constant $n_i$ of $n_iS^{\lam_i}$ indicates the multiplicity of $S^{\lam_i}$ in the (unscaled) representation $28S$.

\end{example}
\begin{figure}[htp]  
    \centering
    \includegraphics[width=14cm]{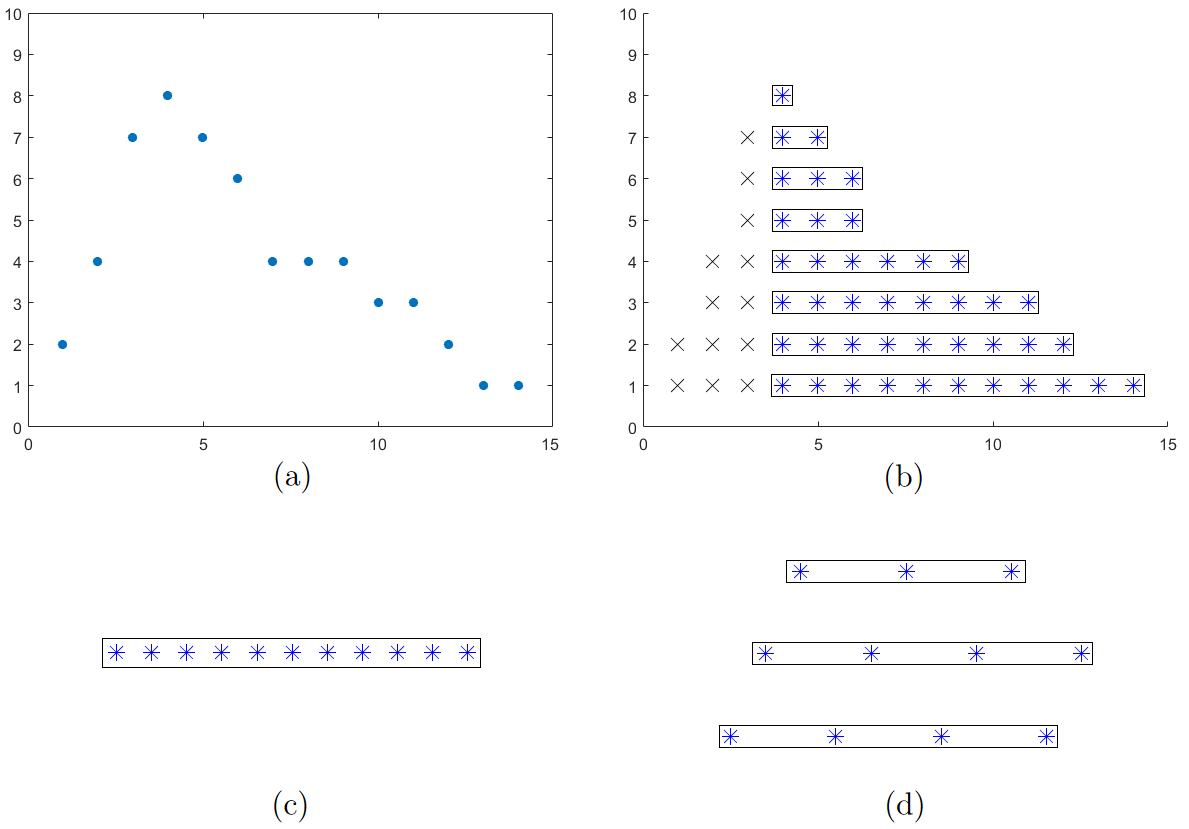}
    \caption{\dark \label{PartitionOfReps}  Illustration of irreducible representations for Example \ref{OgataEx}.}
\end{figure}

Figure \ref{PartitionOfReps}(a) is a graph of the multiplicities of the irreducible representations in $S$. We construct the almost commuting matrices $A_i'$ nearby the $S(\sigma_i)$ as follows. 
We first partition the direct sum appropriately, which gives us subrepresentations acting on orthogonal invariant subspaces. For each of these subrepresentations we construct nearby commuting matrices.
Then the nearby commuting matrices $A_i'$ are formed by taking the direct sum of the commuting matrices formed in all the invariant subspaces. The distance $\|A_i'-S(\sigma_i)\|$ will be the maximal distance in each of the invariant subspaces corresponding to the partition.

We now discuss the partitions and how we construct their nearby commuting matrices. First, refine the representations illustrated in (b) into two subsets illustrated with $\times$'s and $\ast$'s. 
One such partition will correspond to the $\times$ irreducible representations.
Because the spins of the $\times$ representations are at most $3$, we will ``discard'' all of these by choosing trivial nearby commuting matrices as in the previous lemma. This provides an error of $\frac3{28}$.

We chose which representations were $\times$'s and $\ast$'s in such a way that the multiplicities of the $\ast$ irreducible representations were monotonically decreasing. We then can form a ``level set'' decomposition illustrated by some long and some short horizontal boxes that group the $\ast$ representations as in (b). 

A sample horizontal grouping of representations is given in (c).
Each such horizontal grouping of representations will be itself partitioned as follows. We choose a value of $L$, which is $3$ in this example. We partition each horizontal grouping of $\ast$ representations so that the spins in each partition increase by exactly $L$. These are illustrated in (d).

The way that this is described in the proof of the Theorem \ref{mainthm} is by choosing the arithmetic progression of spins $\mu_1, \mu_2, \dots, \mu_K$ where $\mu_{i+1}-\mu_i=L$ and $\mu_K$ is one of the last $L$ spins to the far right of the grouping in (c).
These provide the partitions of the $\ast$ representations for which we obtain nearby commuting matrices by Lemma \ref{bigLstepLemma}. 

Note that, strictly speaking, in order to apply Theorem \ref{mainthm}, we do not need the representations to be monotonically increasing in the sense that $n_i \geq n_{i+1}$ after some point. What is actually needed is that the multiplicities are monotonically decreasing with steps of size $L$: $n_i \geq n_{i+L}$.

\begin{thm}\label{mainthm}
Let $N \geq 1$ and $\lam_1, \dots, \lam_m$ be given with $\lam_{i_\ast} \leq \frac12\sqrt{N}+1$, $\lam_m \leq \frac12N$, and $\lam_{r+1}-\lam_r =1$. Define $L = \lfloor 1.045\, N^{4/7} \rfloor$.\\
Let $S = \frac1N\left( n_1S^{\lam_1}\oplus \cdots \oplus n_mS^{\lam_m}\right)$, where $n_{i} \geq n_{i+L}$ for $i \geq i_\ast$. 

Then there are commuting self-adjoint matrices $A_i'$ such that
\begin{align}\|A_1'-S(\sigma_1)\|, \|A_2'-S(\sigma_2)\| &\leq 6.286\,N^{-\frac17}, \nonumber\\
\|A_3'-S(\sigma_3)\| &\leq 1.083\,N^{-\frac37}\nonumber
\end{align}
and $A_1', iA_2', A_3'$ are real.

The same result applies if instead $\lam_{r+1}-\lam_r = 1/2$.
\end{thm}
\begin{proof}
We first relabel the  indices of the weights so that $i_\ast=1$ and the weights are $\lam_i$ for $i_0\leq i\leq m$ with $i_0 \leq 1$ being possibly negative. To avoid the trivial case, we can assume that $N \geq (2\cdot6.2)^7 \approx 4.5 \times 10^7$.

Because the differences $\lam_{r+1} - \lam_r$ are an integer, all the $\lam_r$ are integers or half-integers. 
 If we had instead $\lam_{r+1} - \lam_r = 1/2$ then we decompose $S$ into a direct sum of the representations with $\lam_r$ integers and $\lam_{r}'$ half-integers and apply the construction for each separately with $\lam_1 \leq \frac12\sqrt{N}+1, \lam_1' \leq \frac12\sqrt{N}+\frac32$. 
So, we assume that $\lam_{r+1}-\lam_r = 1$ and $\lam_1 \leq \Lambda_0$, where $\Lambda_0 = \frac12\sqrt{N}+\frac32$.

We now break the representation into subrepresentations as follows. If $\lam_m$ is an integer, let $Z$ be the set of integers. If $\lam_m$ is a half-integer, let $Z$ be the set of half-integers. Then the collection of all $\lam_r$ is equal to $[\lam_{i_0}, \lam_m] \cap Z$. We first partition $[\lam_{i_0}, \lam_m] \cap Z$ into
$[\lam_{i_0}, \Lambda_0+2L) \cap Z$ and $[\Lambda_0+2L, \lam_m] \cap Z$. 

For each $\mu_K \in (\lam_m-L, \lam_m] \cap Z$, we form a disjoint (with indices relabeled) arithmetic progression $\mu_1, \dots, \mu_K$, where $\Lambda_0+L< \mu_1 \leq \Lambda_0+2L$ and $\mu_{i+1}-\mu_i = L$. The set $[\Lambda_0+2L, \lam_m]\cap Z$ is thus contained in the union of these disjoint arithmetic progressions. 

We now focus on forming nearby commuting self-adjoint matrices for subrepresentations of the representation $N\cdot S$ corresponding to the arithmetic progressions and also to the representations not accounted for by one of the arithmetic progressions. Then the desired matrices $A_i'$ are formed from the appropriate direct sums.

Suppose $\lam_j\in [\lam_{i_0}, \Lambda_0+2L) \cap Z$ does not belong to one of the above constructed arithmetic progressions. Then 
\[\lam_j \leq \frac12\sqrt{N}+\frac32+2L \leq 5N^{\frac47},\]
hence \[\|S^{\lam_j}(\sigma_i)\| \leq \frac{5}{N}N^{\frac47} = 5\, N^{-\frac37}.\] So, on this summand we choose the component of $A_1'$ and of $A_2'$ to be zero and the component of $A_3'$ to be $\frac{1}{N}S^{\lam_j}(\sigma_3)$. This guarantees a contribution of at most $5N^{-3/7}$ to $\|A_1' - S(\sigma_1)\|$ and $\|A_2' - S(\sigma_2)\|$ on this summand and no contribution to  $\|A_3' - S(\sigma_3)\|$ on this summand.

Now, consider one of the above constructed arithmetic progression $\mu_1, \dots, \mu_K$. For simplicity of notation, let $n_{\mu}$ be the multiplicity of the representation $S^\mu$ in the representation $N\cdot S$. Then
\[\bigoplus_{i=1}^K n_{\mu_i}S^{\mu_i} = n_{\mu_K}\left(S^{\mu_1}\oplus \cdots \oplus S^{\mu_K}\right)\oplus \bigoplus_{r=2}^{K}(n_{\mu_{r-1}}-n_{\mu_r})(S^{\mu_1}\oplus \cdots\oplus S^{\mu_{r-1}}).\]
This is well-defined because $n_i - n_{i+L}\geq 0$ so $n_{\mu_{r-1}}- n_{\mu_r} \geq 0$.

So, we focus on obtaining nearby commuting matrices for the representation of the form $S^{\mu_1}\oplus \cdots\oplus S^{\mu_{r}}$.
Nearby commuting matrices are obtained  by applying Lemma \ref{bigLstepLemma} since $\mu_1 \leq \Lambda_0 + 2L, \mu_K \leq \lam_m\leq \frac12N, \mu_{i+1}-\mu_i=L$. 
So, we conclude the proof of the lemma by taking direct sums of the nearby commuting matrices obtained in each summand.
\end{proof}

We now prove Theorem \ref{OgataTheorem}, giving a constructive proof of Lin's Theorem for $2\times 2$ matrices with an explicit estimate and additional structure.
\begin{proof}[Proof of Theorem \ref{OgataTheorem}]
We begin with the first statement. Consider the representation $(S^{1/2})^{\otimes N}$ decomposed as a direct sum of irreducible representations as discussed in  Section \ref{Prelim}. We write
\[(S^{1/2})^{\otimes N} \cong n_{0}S^0 \oplus n_{1/2}S^{1/2}\oplus \cdots \oplus n_{N/2}S^{N/2}.\]
As discussed in Section \ref{Prelim}, this decomposition as well as the unitary operator on $M_{2^N}(\C)$ that realizes this equivalence can be obtained  constructively. Moreover, we choose the unitary to be real.

Depending on whether $N$ is even or odd, the $n_\lam$ are only non-zero when the $\lam$ are all integers or are all half-integers, respectively. By Lemma \ref{1/2mult}, we know that $n_\lam \geq n_{\lam+1}$ for $\lam \geq \frac12\sqrt{N}$. So, we apply Theorem \ref{mainthm} with
\[T_N = \frac{1}{N}(S^{1/2})^{\otimes N} \cong \frac{1}{N}\left(n_{0}S^0 \oplus n_{1/2}S^{1/2}\oplus \cdots \oplus n_{N/2}S^{N/2}\right)\]
 to obtain commuting real self-adjoint matrices $Y_{1,N}, iY_{2,N}, Y_{3,N}$ that satisfy
\[\|T_N(\sigma_i)-Y_{i,N}\| \leq 6.286\,N^{-\frac17}\]
for $i = 1, 2$ and
\[\|T_N(\sigma_3)-Y_{3,N}\| \leq 1.083\,N^{-\frac37}.\]

To obtain the estimate for a more general operator, we proceed as discussed at the end of Section \ref{Intro2}. If $A$ is given by $c_1\sigma_1 + c_2\sigma_2 + c_3\sigma_3 + c_4I_2$ then define $Y_{N}(A) = c_1Y_{1, N} + c_2Y_{2, N} + c_3Y_{3, N} + c_4I_{2^N}$. Recall that $T_N(I_2) = I_{2^N}$ and by Equation (\ref{pauliNorm}),
\[\sqrt{|c_1|^2 + |c_2|^2 + |c_3|^2} \leq 2\|A\|.\]
 So, by the Cauchy-Schwartz inequality,
\begin{align*}
\|T_N(A) - Y_{N}(A)\| &\leq \sum_{i=1}^3 |c_i|\|T_N(\sigma_i)-Y_{i,N}\| \leq 6.286(|c_1|+|c_2|)N^{-\frac17} + 1.083|c_3|N^{-\frac37}\\
&\leq 2\sqrt{2(6.286^2) + 1.083^2N^{-4/7}}\|A\|N^{-\frac17} \leq  17.92\|A\|N^{-1/7}.
\end{align*}

Recall that by Equation (\ref{pSpin}),
\[\sigma_1= \frac12\bp 0 & 1\\1&0  \ep,\;\; \sigma_2 = \frac12\bp 0 & i\\ -i &0 \ep,\;\; \sigma_3=\frac12\bp -1&0\\0&1 \ep.\]
So, $\sigma_1, \sigma_3, I_2$ are symmetric self-adjoint $2\times 2$ matrices and $\sigma_2$ is an antisymmetric self-adjoint matrix. Because $Y_{1,N}, Y_{3,N}, Y_{N}(I_2)$ are real and self-adjoint, they are symmetric. Because $Y_{2, N}$ is imaginary and self-adjoint, it is antisymmetric. Therefore,
\[Y_{N}(A^\ast) = Y_{N}(\overline{c_1}\sigma_1 + \overline{c_2}\sigma_2 + \overline{c_3}\sigma_3 + \overline{c_4}I_2) = \overline{c_1}Y_{1,N} + \overline{c_2}Y_{2,N} + \overline{c_3}Y_{3,N} + \overline{c_4}Y_{I,N} = Y_{N}(A)^\ast\]
and
\[Y_{N}(A^T) = Y_{N}(c_1\sigma_1 -c_2\sigma_2 +c_3\sigma_3 +c_4I_2) = c_1Y_{1,N}-c_2Y_{2,N} + c_3Y_{3,N} + c_4Y_{I,N} = Y_{N}(A)^T.\]

The theorem then follow from these observations.
\end{proof}

\begin{remark}\label{useful}
For a 3 dimensional grid of $10^5$ particles along each axis, one sees that $N=10^{15}$ is a reasonable value of $N$ to apply our result to. We then would have the estimates
$\|T_N(\sigma_i) - Y_{i,N}\| \leq 0.046$ and for more general operators $\|T_N(A) - Y_{i,N}\| \leq 0.13\|A\|$.

For $N = (10^{7})^3$,
$\|T_N(\sigma_i) - Y_{i,N}\| \leq 0.0063$ and $\|T_N(A) - Y_{i,N}\| \leq 0.018\|A\|$.

For $N = (10^{10})^3$,
$\|T_N(\sigma_i) - Y_{i,N}\| \leq 0.00033$ and $\|T_N(A) - Y_{i,N}\| \leq 0.00093\|A\|$.
\end{remark}
\begin{remark}
Loring and S{\o}rensen in \cite{LS} extend Lin's theorem to respect real matrices. They show that two almost commuting real self-adjoint matrices are nearby commuting real self-adjoint matrices. We have shown that this result is true for Ogata's theorem for $d = 2$. 
The result that we found of the additional structure for $Y_{i,N}$ corresponds to what \cite{KTheoryandS} calls Class D in 2D (Section 5.2), which is the case of two real self-adjoint matrices and one imaginary self-adjoint matrix that are almost commuting and for which we want to find nearby commuting approximants with the same structure.
\end{remark}

It should be remarked that the suboptimal exponent $\alpha = 1/3$ was used because it provided the real structure of the $A_i'$ and a small explicit constant $C_{1/3}$. Using $\alpha=\frac12, \gamma_3 = \frac35, \gamma_5=\frac25, \gamma=\frac15$ and similar arguments as above, one can obtain the following result. Because $C_{1/2}$ is undetermined we state this result with the best asymptotic decay that our method provides but without an explicit constant.
\begin{thm}\label{OptimalResult}
There is a linear map $Y_N:M_2(\C)\to M_{2^N}(\C)$ such that the $Y_{N}(A)$ commute for all $A\in M_2(\C)$,
\[Y_{N}(A^\ast)=Y_{N}(A)^\ast,\]
and
\[\|T_N(A)-Y_{N}(A)\| \leq Const. \|A\|\,N^{-1/5}.\]

Consequently, $Y_{N}$ preserves the property of being self-adjoint or skew-adjoint.
\end{thm}

\newpage

\vspace{1in}

\textbf{ACKNOWLEDGEMENTS}. The author would like to thank Eric A. Carlen for introducing the problem to the author, providing continued guidance during the writing and revising of this paper, and for providing context for the useful size of estimates.
The author would also like to thank Terry A. Loring for his feedback on the benefits and limitations of the construction presented in this paper.

This research was partially supported by NSF grants DMS-2055282 and DMS-1764254.

\end{document}